\documentclass[aps,prx,twocolumn,superscriptaddress,longbibliography,floatfix]{revtex4-2}
\usepackage{placeins}
\usepackage{amsmath,amssymb,amsthm}
\usepackage{graphicx}
\usepackage{dcolumn}
\usepackage{bm}
\usepackage{hyperref}
\hypersetup{hypertexnames=false}
\usepackage{enumitem}
\usepackage{algorithm}
\usepackage{algpseudocode}
\usepackage{tikz}

\usepackage{float}
\usepackage{booktabs}
\usetikzlibrary{positioning,arrows.meta,shapes.geometric,
decorations.pathreplacing,calc}

\newcommand{\ket}[1]{|#1\rangle}
\newcommand{\bra}[1]{\langle#1|}
\newcommand{\braket}[2]{\langle#1|#2\rangle}
\newcommand{\PR}{\mathrm{PR}}
\newcommand{\Tr}{\mathrm{Tr}}
\newcommand{\Var}{\mathrm{Var}}

\theoremstyle{plain}
\newtheorem{theorem}{Theorem}
\newtheorem{lemma}{Lemma}

\newtheorem{corollary}{Corollary}
\newtheorem{remark}{Remark}
\newtheorem{definition}{Definition}

\begin{document}

\title{Basis-Adaptive Sparse-State Simulation of Quantum Circuits}

\author{Ch Nihar Kartikeya}
\email{kartikeya949243@gmail.com}
\affiliation{Department of Physics, Indian Institute of Technology Hyderabad, Kandi, Sangareddy, Telangana 502285, India}

\author{Anjana K}
\affiliation{Department of Physics, Indian Institute of Technology Hyderabad, Kandi, Sangareddy, Telangana 502285, India}

\author{Bijita Sarma}
\affiliation{Chennai Mathematical Institute, H1, SIPCOT IT Park, Siruseri, Tamil Nadu 603103, India}

\author{Sangkha Borah}
\email{sangkha.borah@phy.iith.ac.in}
\affiliation{Department of Physics, Indian Institute of Technology Hyderabad, Kandi, Sangareddy, Telangana 502285, India}

\date{\today}

\begin{abstract}
Classical simulation of many-body quantum systems remains economical only when wavefunction amplitudes stay localized in the working basis. Fixed-basis sparse-state simulators scale memory as $\mathcal{O}(k)$ by keeping the largest computational-basis amplitudes; however, fidelity drops once entanglement or basis rotations spread weight across the Hilbert space. In this work, we introduce an algorithm called Basis-Adaptive Sparse-State Simulation (BASS), which updates each qubit's local representation basis during execution rather than locking the computational basis for the entire circuit. Before truncation, each qubit is rotated into the eigenbasis of its single-qubit reduced density matrix, following the natural-orbital idea from quantum chemistry, so the retained amplitudes stay clustered. We prove that top-$k$ selection is uniquely optimal for one-step truncation in any fixed basis and that the one-body reduced-density-matrix eigenbasis is a stationary product basis for the inverse participation ratio (PR), with a residual bounded by local entanglement coherence. We perform a systematic benchmarking over a variety of quantum circuits and demonstrate that the ratio \(k/\text{PR}_Z\) (sparse budget over computational participation ratio) serves as an indicator for regimes in which adaptive measurement bases provide a performance advantage. On structured brickwork circuits, BASS achieves substantially higher fidelity than the fixed-basis approach, while incurring only a moderate increase in wall-clock time in the memory-limited regime. Moreover, for disordered Ising circuits, BASS systematically provides an improvement of approximately one order of magnitude in state overlap at a fixed computational budget.
\end{abstract}

\maketitle

\section{\label{sec:intro}Introduction}

Algorithm prototyping, benchmark design, and quantum hardware verification -- all practically need approximate classical simulations. Although existing quantum processors already function in regimes where perfect simulation is not available~\cite{boixo2018,huang2020,markov2008}, exact state-vector simulation scales as $\mathcal{O}(2^N)$ in memory, placing the practical ceiling around $N \approx 50$ on the largest accessible supercomputers. Classical emulators have improved in recent years~\cite{begusic2024,tindall2024,kechedzhi2024,pan2022}, but hardware groups still depend on approximate simulators to prototype near-term algorithms and check quantum-advantage claims~\cite{arute2019,wu2021,kim2023}.

Approximate classical simulation is organized according to three main paradigms, each of which takes advantage of a distinct structural characteristic of quantum states. In particular, matrix product states (MPS)~\cite{white1992,vidal2003,schollwock2011,fishman2022} represent the state as a chain of tensors with bounded bond dimension~$\chi$ and achieve efficiency whenever the bipartite entanglement entropy follows an area law.
MPS techniques are effective for one-dimensional systems with moderate entanglement~\cite{tindall2024,begusic2024,ayral2023}, but they become exponentially more expensive for two-dimensional geometries, volume-law entanglement, or circuits nearing quantum criticality~\cite{orus2014,verstraete2008,zhou2020}, which are precisely the regimes that quantum advantage experiments target.

\textit{Stabilizer-based methods} use the Gottesman-Knill theorem to effectively simulate Clifford circuits~\cite{gottesman1998,aaronson2004}. More recent expansions handle Clifford-dominated circuits by decomposing them into stabilizer sums~\cite{bravyi2016,bravyi2019}. However, their computational cost increases exponentially with the number of non-Clifford ($T$) gates, which restricts their use to circuits with low magic~\cite{haug2023,leone2022}. \textit{Decision-diagram methods}~\cite{zulehner2019,hillmich2020} compress the state through algebraic redundancy and perform well for structured circuits but quickly deteriorate for generic entangling operations. As a result, each paradigm has a basic blind spot: decision diagrams fail for generic unitaries, stabilizer techniques fail for high-magic circuits, and tensor networks fail for highly entangled states. Quantum advantage circuits are deliberately designed to evade all three~\cite{movassagh2023,dalzell2022,brandao2016}.

Sparse-state simulation targets a different resource:
\textit{computational-basis concentration}. Rather
than representing the full $2^N$-dimensional state vector, a sparse simulator
retains only the $k \ll 2^N$ basis states with the largest squared amplitudes.
Fixed-basis sparse simulators that implement this top-$k$ rule, including the
recent formulation of Miller et al.~\cite{miller2026}, demonstrate
$\mathcal{O}(k)$ memory and $\mathcal{O}( \log k)$ time per gate. The resource cost is controlled
by a single user-chosen parameter~$k$ and is independent of entanglement
structure or circuit geometry---a property shared by no other simulation
paradigm. Related ideas have appeared in photonic quantum
computing~\cite{heurtel2023} and in quantum chemistry's selected configuration
interaction (CI) methods~\cite{tubman2016,tubman2020}, where analogous
sparsity structures arise naturally.

The bottleneck is simpler: \textit{the computational $Z$ basis is fixed}. When a quantum
state is not naturally sparse in this basis---as occurs after Hadamard-rich
layers, near quantum critical points, in circuits designed for quantum
advantage, or when the natural physics resides in a rotated frame---the $k$
retained amplitudes capture a vanishingly small fraction of the total
probability. The paramagnetic ground state $\ket{+}^{\otimes N}$, which is a single state in the $X$ basis but distributes uniformly across all $2^N$ computational-basis states, is the canonical example. More generally, the sparse budget is quickly depleted by any circuit that produces significant superposition relative to the computational basis, and the issue worsens exponentially with system size. In the regimes of highest physical interest, this basis-dependence bottleneck has restricted the usefulness of sparse-state simulation.

This bottleneck has a close analog in quantum chemistry. The dependence of
configuration-interaction expansions on the molecular-orbital basis was
recognized by L\"owdin~\cite{lowdin1955}, who showed that the eigenbasis of
the one-body reduced density matrix - the \textit{natural-orbital basis} - is
a stationary and highly compact representation for leading configurations.
Ratini et al.\ recently demonstrated that natural orbitals also improve
mutual-information sparsity in quantum chemistry
simulations~\cite{ratini2024}. This connection between basis choice and
representation efficiency has driven decades of research in orbital
optimization~\cite{roos1980,tubman2016,tubman2020}. The present work transfers
the same principle to online sparse simulation of gate-model quantum circuits,
where the relevant local density matrices are single-qubit RDMs rather than
molecular one-body density matrices.

We ask whether the representation basis can be updated during the circuit and,
if so, whether the single-qubit RDM eigenbasis is a sensible choice. We study
top-$k$ sparse-state simulation in adaptive product bases and report theory,
code, and benchmarks for when that helps. Our contributions are:

\indent \textit{(i) A theoretical separation between support selection and basis selection.} We demonstrate two foundational lemmas and a central theorem that formulate the mathematical boundaries of what can and cannot be enhanced within the sparse-state paradigm.

Lemma~\ref{lem:topk}: \textit{Top-$k$ optimality.} In any fixed orthonormal
basis and in any fixed budget~$k$, maintaining the $k$ states with the largest
squared amplitudes maximizes the single-step retained probability; the
maximizer is eccentric when there is no tie at the truncation boundary.

This result, a direct consequence of the classical rearrangement inequality,
closes the support-selection problem in every step: no reweighting, diversity
penalty, or entanglement-aware selection can improve on the simple top-$k$
rule without first changing the basis. In spite of being simple, this result is helpful because, when combined with Lemma~\ref{lem:diversity} and Theorem~\ref{thm:stationarity}, it distinguishes between the still-adjustable representation basis and the already-optimal greedy support-selection step.

Lemma~\ref{lem:diversity}: \textit{Probability-diversity tradeoff.}
This rigorously eliminates the natural class of ``entanglement-aware'' truncation strategies
that sacrifice amplitude for structural diversity, clarifying why such
strategies have not succeeded in practice.

Theorem~\ref{thm:stationarity}: \textit{BASS stationarity.} The single-qubit
reduced-density matrix (RDM) eigenbasis is an \emph{exact} stationary point of
the participation-ratio objective in the product-state limit and a stationary
point up to a controlled residual---bounded by the covariance between local
amplitude concentration and entanglement coherence---for entangled inputs.
This residual vanishes continuously as the single-qubit entanglement entropy
approaches zero, identifying the RDM eigenbasis as the natural choice for
sparse simulation in direct mathematical analogy to L\"owdin's natural
orbitals in quantum chemistry. Whether stationarity implies minimality in full generality is left open; in our benchmarks the do-no-harm guard reverts fewer than 0.5\% of proposed rotations. The do-no-harm guard checks the resulting representation and undoes a rotation if it accidentally makes the state less sparse; see Section~\ref{sec:algorithm} for full details.

These three results work together rather than independently.
Lemmas~\ref{lem:topk} and~\ref{lem:diversity} prove that within any
fixed basis the \emph{single-step} support-selection problem has a
top-$k$ optimum (unique up to boundary ties), and no amplitude-diversity
alternative improves single-step retained probability at fixed
budget. The natural remaining degree of freedom is the basis, and
Theorem~\ref{thm:stationarity} identifies the principled, provably stationary
choice (with the residual bounded by single-qubit entanglement coherence).
Whether \emph{cumulative} retention across many gates admits a non-greedy
optimum, and whether the RDM eigenbasis is also a local or global
minimum, remain open; we mark these as conjectures supported by exhaustive
empirical evidence (Table~\ref{tab:proofs}).

\begin{figure*}[!htb]
  \centering
  \resizebox{0.9\linewidth}{!}{%
  \begin{tikzpicture}[
      >=Stealth,
      barZ/.style    ={fill=blue!35,  draw=blue!70!black,  thick},
      barZcut/.style ={fill=blue!10,  draw=blue!40!black,  thick},
      barB/.style    ={fill=red!45,   draw=red!75!black,   thick},
      barBcut/.style ={fill=red!12,   draw=red!40!black,   thick},
      cutline/.style ={dashed, thick, red!75!black},
      annot/.style   ={font=\fontsize{8pt}{8pt}\selectfont},
      title/.style   ={font=\bfseries\fontsize{8pt}{13.2pt}\selectfont},
    ]

    \node[title] at (1.7,4.2) {Computational $Z$ basis};
    \node[annot] at (1.7,3.8) {$\PR_Z \gg k$\;\;(delocalized)};

    \draw[->,thick] (0,0) -- (3.5,0);
    \node[annot] at (1.75,-0.4) {basis state $x$};
    \draw[->,thick] (0,0) -- (0,3.0) node[above,annot] {$|\alpha_x|^2$};

    \foreach \x/\h in {%
        0.20/0.45, 0.40/0.55, 0.60/0.40, 0.80/0.65, 1.00/0.50,
        1.20/0.70, 1.40/0.60, 1.60/0.55, 1.80/0.75, 2.00/0.50,
        2.20/0.45, 2.40/0.60, 2.60/0.40, 2.80/0.55, 3.00/0.50,
        3.20/0.35}{%
      \pgfmathparse{\h > 0.62 ? 1 : 0}\ifnum\pgfmathresult=1
        \draw[barZ]    (\x-0.05,0) rectangle (\x+0.05,\h);
      \else
        \draw[barZcut] (\x-0.05,0) rectangle (\x+0.05,\h);
      \fi
    }
    \draw[cutline] (-0.05,0.62) -- (3.45,0.62);
    \node[annot,red!75!black] at (3.45,0.62) [above left] {top-$k$ cut};

    \draw[decorate,decoration={brace,amplitude=4pt,raise=1pt},thick,
          blue!60!black]
      (0.15,0.90) -- (3.05,0.90)
      node[midway,above=3pt,annot,align=center,blue!60!black]
      {only $\sim k/\PR_Z$ of\\the mass retained};

    \begin{scope}[shift={(5.35,3.6)},scale=0.65]
      \draw[thick] (0,0) circle (1);
      \draw[dashed] (-1,0) arc (180:360:1 and 0.3);
      \draw[thick]  (-1,0) arc (180:0:1 and 0.3);
      \draw[->,thick] (0,-1.05) -- (0,1.10);
      \node[annot,scale=1] at (0,1.45) {$\ket{0}$};
      \node[annot,scale=1] at (0,-1.45) {$\ket{1}$};
      \draw[->,thick,red!75!black] (0,0) -- (0.57,0.82);
      \node[annot,red!75!black,scale=1.0] at (1.1,1.1)
        {$\vec{v}_j^{\,\max}$};
    \end{scope}

    \draw[->,line width=1.4pt,blue!60!black]
      (3.85,1.6) -- (6.85,1.6);
    \node[title,blue!60!black] at (5.35,2.0)
      {$U \;=\; \bigotimes_{j=1}^{N} U_j$};
    \node[annot,align=center] at (5.35,0.9)
      {$U_j$ diagonalizes\\[0.5ex] $\rho_j=\Tr_{\!\neq j}\ket{\psi}\!\bra{\psi}$};

    \begin{scope}[shift={(0.7,0)}]
      \node[title] at (8.3,4.2) {RDM eigenbasis (BASS)};
      \node[annot] at (8.3,3.8) {$\PR \sim \mathcal{O}(1)$\;\;(concentrated)};

      \draw[->,thick] (6.55,0) -- (10.1,0);
      \node[annot] at (8.325,-0.4) {basis state $x$};
      \draw[->,thick] (6.55,0) -- (6.55,3.0) node[above,annot] {$|\alpha_x|^2$};

      \foreach \x/\h in {6.75/2.05, 6.95/1.45, 7.15/1.05, 7.35/0.75, 7.55/0.55}{%
        \draw[barB] (\x-0.05,0) rectangle (\x+0.05,\h);
      }
      \foreach \x/\h in {%
          7.85/0.18, 8.05/0.14, 8.25/0.10, 8.45/0.08,
          8.65/0.07, 8.85/0.06, 9.05/0.05, 9.25/0.05,
          9.45/0.04, 9.65/0.04, 9.85/0.03}{%
        \draw[barBcut] (\x-0.05,0) rectangle (\x+0.05,\h);
      }
      \draw[cutline] (6.50,0.45) -- (10.05,0.45);
      \node[annot,red!75!black] at (10.05,0.45) [above left] {top-$k$ cut};

      \draw[decorate,decoration={brace,amplitude=4pt,raise=1pt},thick,
            red!70!black]
        (6.70,2.15) -- (7.60,2.15);
      \node[annot,align=left,red!70!black,anchor=south west] at (6.65,2.30)
        {$\sim 1-\varepsilon$ of the\\mass retained};
    \end{scope}

  \end{tikzpicture}
  }
  \caption{\textbf{Schematic of basis-adaptive sparse-state
  simulation.}
  \textbf{Left:} in the computational basis the state $\ket{\psi}$ is
  delocalized, with participation ratio $\PR_Z \gg k$; a sparse simulator
  retains only $\sim k/\PR_Z$ of the total probability (states above the
  dashed top-$k$ cut, dark blue).
  \textbf{Middle:} BASS rotates the representation by a product unitary
  $U=\bigotimes_j U_j$, where each $U_j$ diagonalizes the single-qubit
  reduced density matrix $\rho_j$. The Bloch-sphere inset shows the
  dominant eigenvector $\vec{v}_j^{\,\max}$ along which probability is
  re-concentrated.
  \textbf{Right:} in the rotated frame probability collapses onto
  $\PR\sim O(1)$ basis states, so the same top-$k$ truncation now
  captures $1-\varepsilon$ of the mass at identical memory cost.}
  \label{fig:bass_schematic}
\end{figure*}
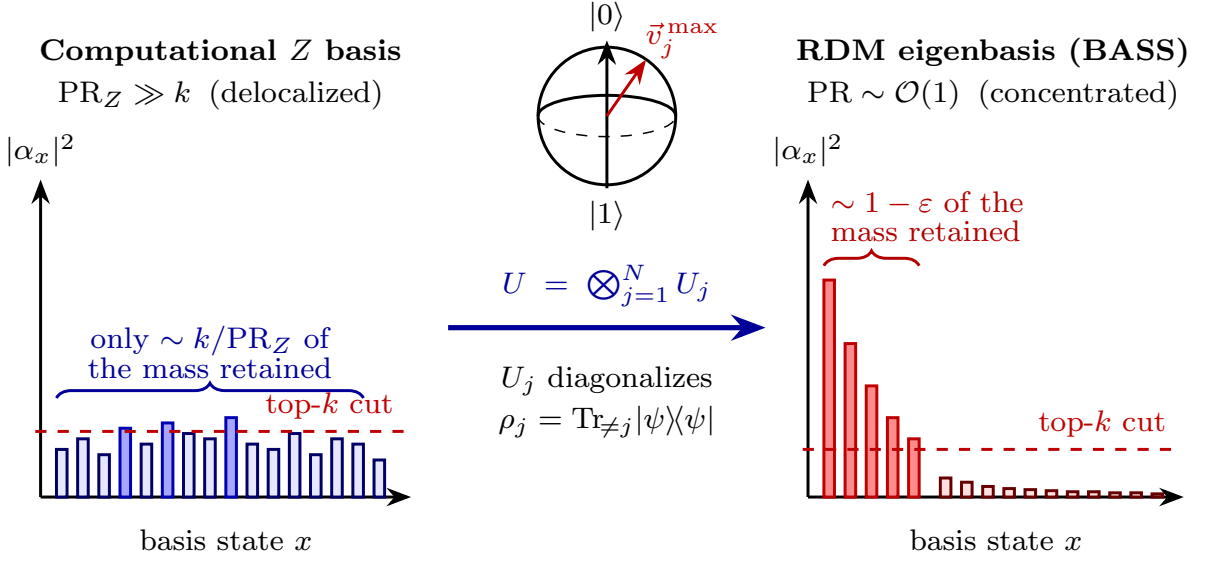

\indent \textit{(ii) A practical algorithmic primitive for classical simulation.}
We introduce BASS (Basis-Adaptive Sparse-State simulation). Qubits are rotated into their local reduced-density-matrix (RDM) eigenbasis before truncation, which tightens top-$k$ retention in the rotated frame. Open-addressing hash tables give $\mathcal{O}(N\cdot \mathrm{nnz})$ RDM updates and $\mathcal{O}(\mathrm{nnz})$ gate applications without sort-heavy paths used in some sparse simulators. Deferred truncation, a diagonal-gate fast path, and multi-pass coordinate descent with a participation-ratio (PR) do-no-harm guard keep overhead bounded; rejected rotations are rare in our benchmarks (Table~\ref{tab:doNH}).

\indent \textit{(iii) A computable crossover diagnostic validated by exponential fidelity gains.}

The $Z$-basis PR $\PR_Z$ is a single dimensionless diagnostic: in our benchmarks BASS wins for $k < \PR_Z$ and matches fixed-basis behavior above it. On brickwork circuits with $N=18$ and $k=5000$, for example, BASS reaches $F \approx 1.3 \times 10^{-2}$ versus $F \approx 4.3 \times 10^{-4}$ for the fixed basis, at about $6\text{--}13\times$ wall-clock cost in the memory-limited regime.

BASS sits alongside tensor-network, stabilizer, and fixed-basis sparse simulators.
Orbital-optimization ideas from quantum chemistry may transfer here as well~\cite{roos1980,tubman2016,tubman2020,ratini2024}.
Adaptive natural-orbital-like updates for sparse circuit simulation appear
largely unexplored in prior sparse simulators~\cite{fasano2022, buijsman2018, ratini2024}.
The $\PR_Z$ crossover rule, the scoped theorems, and the large fidelity gaps on
stress-test circuits are the main quantitative messages of the paper.

The remainder of this paper is organized as follows.
Section~\ref{sec:background} introduces the sparse-state framework, key
quantities, and comparison methods. Section~\ref{sec:theory} presents our three
main theorems with complete proofs.
Section~\ref{sec:schmidt_weighted} provides experimental validation that
Schmidt-weighted truncation is suboptimal. Section~\ref{sec:algorithm} describes the
BASS algorithm and its efficient implementation. Section~\ref{sec:results} presents
numerical benchmarks. Section~\ref{sec:mps_comparison} compares BASS with
tensor-network methods and sketch the sparsity-entanglement phase
diagram. Section~\ref{sec:discussion} covers the physical mechanism, limitations, links to natural-orbital theory, and implications for the quantum advantage boundary. Section~\ref{sec:conclusions} concludes.

\section{\label{sec:background}Background and Preliminaries}
Figure~\ref{fig:bass_schematic} sketches the BASS workflow. For any circuit whose amplitudes are spread across many computational-$Z$ basis states (left), fixed top-$k$ truncation retains only a fraction of the weight ($\sim k/\PR_Z$). A product of single-qubit rotations (middle), $U=\bigotimes_j U_j$, built from each $\rho_j$, reorients the local axes (Bloch-sphere inset). After rotation (right), $\PR$ is much smaller in the adapted frame, so the same budget~$k$ captures a larger share of the probability while memory and per-gate cost remain at the sparse-simulator scale.
We next fix notation for sparse states, PR, and the simulators
we compare against.

\subsection{Sparse-state representation}

An $N$-qubit quantum state $\ket{\psi} \in (\mathbb{C}^2)^{\otimes N}$ requires all $2^N$ computational-basis amplitudes for an exact representation. Any sparse approximation retains only the $k$ most significant of these
amplitudes:
\begin{equation}
\ket{\phi} = \frac{1}{\gamma} \sum_{i=1}^{k} \alpha_i \ket{x_i},
\quad
\gamma^2 = \sum_{i=1}^{k}|\alpha_i|^2,
\label{eq:sparse}
\end{equation}
where $\mathcal{S}=\{x_1,\ldots,x_k\}\subset\{0,1\}^N$ is the support set,
$\alpha_i = \braket{x_i}{\psi}$ are the retained amplitudes, and
$\gamma \le 1$ is the normalization factor ensuring $\braket{\phi}{\phi}=1$.
The \textit{retained probability} is $p_\mathcal{S} = \gamma^2 =
\sum_{x \in \mathcal{S}} |\braket{x}{\psi}|^2$, which equals the
fidelity $F = |\braket{\psi}{\phi}|^2$ for a single truncation step~\cite{miller2026}.

The cumulative retention probability $\Gamma_L = \prod_{t=1}^{L} \gamma_t^2$ grows multiplicatively in a multi-step simulation with $L$ truncation events. With Pearson correlation $\rho(\Gamma_L, F) \ge 0.97$ across all methodologies and circuit families investigated in this work, $\Gamma_L$ is an experimentally reliable predictor, even though it generally does not match the true fidelity.

\subsection{Participation Ratio}

The Participation Ratio (PR) quantifies the effective number of basis states that substantially contribute to the quantum state:
\begin{equation}
\PR = \frac{\bigl(\sum_{i=1}^{D}|\alpha_i|^2\bigr)^2}
{\sum_{i=1}^{D}|\alpha_i|^4}
= \frac{1}{\sum_{i=1}^{D}|\alpha_i|^4},
\label{eq:pr}
\end{equation}
where the second equality is true for normalized states and $D = 2^N$. The PR satisfies $1 \le \PR \le D$: $\PR = D$ denotes the uniform superposition, and $\PR = 1$ denotes a single basis state. A state with low PR in a given basis is highly concentrated---and therefore well suited for sparse-state simulation---while a state with $\PR \gg k$ spreads its weight across many more basis states than the simulator can retain.

The $Z$-basis PR $\PR_Z$ (computed in the computational
basis) plays a central role throughout this work: across the benchmarked
families it predicts the crossover between the regime where BASS provides
advantage ($k < \PR_Z$) and where standard fixed-basis truncation suffices
($k > \PR_Z$).

\subsection{Fixed-basis sparse simulation}
\label{subsec:fixed_basis}

The fixed-basis sparse simulator used throughout this work operates entirely
in the computational $Z$ basis, following the same top-$k$ sparse-update
principle as Ref.~\cite{miller2026}. For each two-qubit gate: (1)~apply the
$4 \times 4$ unitary to every stored basis state, expanding the support to at
most $4k$ entries; (2)~merge duplicate basis states by summing amplitudes;
(3)~sort by $|\alpha|^2$; (4)~retain the top-$k$ entries; (5)~update
$\gamma$. The cost per gate is $\mathcal{O}( \log k)$ (dominated by the sort); total
memory is $\mathcal{O}(k)$. We refer to this comparison method simply as the
\textit{fixed-basis} sparse simulator in the remainder of the paper.

\subsection{Matrix product states and Schmidt rank}
\label{subsec:mps_intro}

An MPS with bond dimension $\chi$ represents the state as a product of $N$
tensors, each of dimension at most $2 \chi^2$, requiring
$\sim 2N\chi^2$ parameters. The cost of simulating an $M$-gate circuit is
$\mathcal{O}(MN\chi^3)$. MPS methods are efficient when the bipartite entanglement
entropy satisfies an area law, so that $\chi$ remains bounded; for
one-dimensional critical systems, $\chi$ grows polynomially in $N$, while for
two-dimensional or Haar-random circuits, it grows
exponentially~\cite{schollwock2011,orus2014}.

\section{\label{sec:theory}Theoretical Results}
The theoretical contribution of this work is a scoped characterization of
what can and cannot be improved by support selection within the sparse-state
simulation paradigm.
The three theorems presented below are logically sequential and progressively
tighten the constraints: Lemma~\ref{lem:topk} closes the support-selection
problem within any fixed basis; Lemma~\ref{lem:diversity} explains the
probability cost of diversity-preserving alternatives; and
Theorem~\ref{thm:stationarity} identifies adaptive product bases as a
principled next degree of freedom. Together, they reduce the practical
question studied here to a sharply posed one: how to choose a useful local
representation basis at a fixed sparse budget.

\subsection{Top-\texorpdfstring{$k$}{k} optimality in any fixed basis}

The first result states the optimality of the core fixed-basis
sparse-update operation.
Retaining the $k$ states with the biggest squared amplitudes maximizes the retention probability $\gamma^2$ and is unique when no amplitudes tie at the truncation boundary. No reweighting, diversity penalty, or entanglement-aware selection can do better without first altering the basis.

\begin{lemma}[Top-$k$ optimality]
\label{lem:topk}
Let $\ket{\chi}=\sum_x \beta_x\ket{x}$ be a normalized state expressed in any
fixed orthonormal basis $\{\ket{x}\}$. Among all index sets $\mathcal{S}$ with
$|\mathcal{S}|=k$, any set $\mathcal{S}^*$ consisting of the $k$ indices with
the largest $|\beta_x|^2$ maximizes the retained probability
$p_{\mathcal{S}} = \sum_{x\in\mathcal{S}}|\beta_x|^2$. If the smallest
retained squared amplitude is strictly larger than the largest excluded
squared amplitude, this maximizer is unique.
\end{lemma}

\begin{proof}
Let $|\beta_{\sigma(1)}|^2 \ge |\beta_{\sigma(2)}|^2 \ge \cdots$ be the
squared amplitudes in sorted order, so that
$\mathcal{S}^* = \{\sigma(1),\ldots,\sigma(k)\}$. For any alternative set
$\mathcal{S}'$ with $|\mathcal{S}'|=k$, we decompose the probability
difference as
\begin{equation}
p_{\mathcal{S}^*} - p_{\mathcal{S}'}
= \sum_{x\in\mathcal{S}^*\setminus\mathcal{S}'}|\beta_x|^2
- \sum_{x\in\mathcal{S}'\setminus\mathcal{S}^*}|\beta_x|^2.
\label{eq:topk_diff}
\end{equation}
Both sums range over the same number of indices,
$|\mathcal{S}^*\setminus\mathcal{S}'| = |\mathcal{S}'\setminus\mathcal{S}^*|$.
Every index $x \in \mathcal{S}^*\setminus\mathcal{S}'$ ranks at most $k$-th
in the sorted order (since it belongs to $\mathcal{S}^*$), while every
$x \in \mathcal{S}'\setminus\mathcal{S}^*$ ranks strictly below $k$-th (since
it does not belong to $\mathcal{S}^*$). Pairing the elements of the two sets
in sorted order, each term in the first sum is at least as large as the
corresponding term in the second, yielding
$p_{\mathcal{S}^*} - p_{\mathcal{S}'} \ge 0$.

Equality in Eq.~\eqref{eq:topk_diff} requires every swapped pair to have
identical squared amplitude, which demands a tie at exactly the $k$-th
largest value. In the generic case of distinct $|\beta_x|^2$, the inequality
is strict and $\mathcal{S}^*$ is the unique maximizer.
\end{proof}

\begin{remark}[Operational scope]
Lemma~\ref{lem:topk} applies to any input state, including previously
truncated states: at each gate application, top-$k$ maximizes the retained
probability of the state that the simulator \textit{currently holds}, regardless
of the prior truncation history.
\end{remark}
\begin{remark}[Greedy vs.\ global optimality]
Top-$k$ maximizes $p_t$ at each step, thereby maximizing
$\Gamma_L = \prod_t p_t$ greedily. Whether a non-greedy strategy that trades lower $p_t$ for greater $p_{t+1}$ would yield higher cumulative retention remains an open question. Empirically, no
alternative strategy outperformed top-$k$ at fixed budget comparisions. (see Sec.~\ref{sec:schmidt_weighted}).
\end{remark}

\subsection{Probability-diversity tradeoff}

Lemma~\ref{lem:topk} shows that for any fixed basis, the top-$k$ truncation technique is optimal. One may intuitively assume that a more diverse support set, encompassing a greater variety of Schmidt sectors, would better preserve entanglement structure and improve fidelity propagation across successive circuit layers. We formally dispute this perception with Lemma~\ref{lem:diversity}, which establishes that for a fixed budget $k$, any support set that diversifies across more Schmidt sectors at the expense of coefficient magnitude always delivers a weakly lower retention probability.
Diversity is not free.
\begin{definition}[Schmidt sector count]
For a bipartition $c$
dividing the $N$ qubits into two groups, the Schmidt sector
count of a support set $S$ is
\begin{equation}
D_c(S)=\left|\left\{\left\lfloor x/2^c \right\rfloor : x\in S\right\}\right|.
\end{equation}
Here, a Schmidt sector corresponds to a distinct computational
basis configuration on one side of the bipartition. Two basis
states belong to the same Schmidt sector if they differ only on
the opposite side of the cut. Equivalently, the Schmidt sectors
group basis states according to the left-half bitstring across the
partition. Thus, $D_c(S)$ counts how many distinct left-side configurations
are represented among the retained basis states in $S$. A larger
Schmidt sector count, therefore, means that the support spans a
more diverse set of entanglement sectors across the cut.
\end{definition}

\begin{lemma}[Probability-diversity tradeoff]
\label{lem:diversity}
Let $|\mathcal{S}'| = |\mathcal{S}^*| = k$. If $D_c(\mathcal{S}') > D_c(\mathcal{S}^*)$ for some bipartition~$c$, then $p_{\mathcal{S}'} \le p_{\mathcal{S}^*}$, with equality only when all indices in $\mathcal{S}^*\setminus\mathcal{S}'$ have the same squared amplitude as those in $\mathcal{S}'\setminus\mathcal{S}^*$.
\end{lemma}

\begin{proof}
The condition $D_c(\mathcal{S}') > D_c(\mathcal{S}^*)$ implies $\mathcal{S}' \neq \mathcal{S}^*$. The result then follows directly from Lemma~\ref{lem:topk}, which establishes $p_{\mathcal{S}'} \le p_{\mathcal{S}^*}$ for any $\mathcal{S}' \neq \mathcal{S}^*$ of the same cardinality, with equality only under tied boundary amplitudes.
\end{proof}

\begin{remark}[Scope and limitations]
Three important caveats apply.
(1)~\textit{Fixed $k$ only}: adaptive-$k$ strategies that increase the budget
at high-entanglement steps are not excluded by this theorem.
(2)~\textit{Single-step guarantee}: whether sacrificing $p_t$ for diversity at
step $t$ might improve $p_{t+1}$ at the next step is an open question; we did
not observe this effect across the trials reported below.
(3)~\textit{Fixed basis}: in the Schmidt basis, the largest amplitudes are
by construction the most entanglement-rich, so no diversity tradeoff
exists-this is precisely the mechanism by which MPS methods succeed.
\end{remark}

\begin{corollary}
Schmidt-weighted truncation-deliberately retaining representatives from
each Schmidt sector at the expense of amplitude-is provably suboptimal at
fixed~$k$. This prediction is confirmed numerically in
Sec.~\ref{sec:schmidt_weighted}.
\end{corollary}

The combined implication of Lemmas~\ref{lem:topk} and~\ref{lem:diversity}
is intentionally narrower but useful: for a fixed basis, fixed budget, and
single truncation step, top-$k$ support selection is optimal. Alternative
selection rules can preserve different structural features, but they cannot
increase the same retained probability without either changing the basis,
changing the budget, or accepting a non-greedy multi-step objective.
The distinction makes the representation basis an ideal target for improving
sparse simulation with fixed memory.

\subsection{BASS stationarity with controlled error}

We now turn to the central question of optimal basis selection. Motivated by
the analogy with natural orbitals in quantum
chemistry~\cite{lowdin1955,ratini2024}, we consider rotating each qubit
independently into the eigenbasis of its single-qubit reduced density matrix.
Theorem~\ref{thm:stationarity} shows that this choice is a stationary point
of the participation-ratio objective, with the gradient residual controlled by
the entanglement of each qubit with the rest of the system.

\begin{theorem}[BASS stationarity]
\label{thm:stationarity}
Let $U_j$ diagonalize the single-qubit RDM
$\rho_j = \Tr_{\ne j} \ket{\psi}\!\bra{\psi}$ and let
$\ket{\phi} = U^\dagger\ket{\psi}$, where
$U = \bigotimes_j U_j$, and define the inverse PR
$\mathcal{I}=\sum_x|\phi_x|^4=1/\PR$. The IPR gradient residual
\begin{equation}
\widetilde{R}_j =
\max_{\|G_j\|_2=1}\left|\frac{d\,\mathcal{I}}{d\theta}\right|_{\theta=0}
\end{equation}
under rotations $U_j(\theta) = e^{-i\theta G_j}U_j$ satisfies the bound
\begin{equation}
\widetilde{R}_j \le
4\sqrt{D \cdot \Var_{\bar{x}_j}\!\bigl(|\phi(\bar{x}_j,0)|^2\bigr)}
\cdot \sqrt{\sum_{\bar{x}_j}|\phi(\bar{x}_j,0)|^2|\phi(\bar{x}_j,1)|^2},
\label{eq:Rj}
\end{equation}
where $D = 2^{N-1}$ and $\bar{x}_j$ index the $N-1$ qubits other than $j$.
The corresponding participation-ratio residual
$R^{(\PR)}_j=\max_{\|G_j\|_2=1}|d\PR/d\theta|_{\theta=0}$ obeys
$R^{(\PR)}_j \le \widetilde{R}_j/\mathcal{I}^2$. A sufficient condition for
both residuals to vanish is that qubit $j$ be unentangled from the rest of the
system; both residuals vanish continuously as the single-qubit entanglement
entropy $S_j \to 0$ for fixed finite $N$.
\end{theorem}

\begin{proof}
The proof proceeds in five steps.

\textit{Step~1 (Gradient expression).}
Under the parameterized rotation $U_j(\theta)=e^{-i\theta G_j}U_j$, the
derivative of the inverse PR at $\theta=0$ evaluates to
\begin{equation}
\left.\frac{d\,\mathcal{I}}{d\theta}\right|_{\theta=0}
= -4\,\mathrm{Re}\!\left(i\,\Tr(G_j \cdot \Lambda_j)\right),
\label{eq:grad_pr}
\end{equation}
where $\Lambda_j$ is the $2\times 2$ matrix with elements
$(\Lambda_j)_{ba} = \sum_{\bar{x}_j} |\phi(\bar{x}_j,a)|^2\,
\phi^*_{(\bar{x}_j,a)}\,\phi_{(\bar{x}_j,b)}$.
The IPR gradient vanishes for all generators $G_j$ if and only if the
off-diagonal element $(\Lambda_j)_{01} = 0$. Since
$d\PR/d\theta=-\mathcal{I}^{-2}d\mathcal{I}/d\theta$, the same stationarity
condition applies to the PR, with the pre-factor included in
$R^{(\PR)}_j$ above.

\textit{Step~2 (Covariance decomposition).}
Define the weight function $w_{\bar{x}} = |\phi_{(\bar{x},0)}|^2$ and the
coherence $c_{\bar{x}} = \phi^*_{(\bar{x},0)}\phi_{(\bar{x},1)}$.
The condition that $U_j$ diagonalizes $\rho_j$ implies
$\sum_{\bar{x}} c_{\bar{x}} = 0$, i.e., the coherence has zero mean.
Setting $\bar{w} = D^{-1}\sum_{\bar{x}} w_{\bar{x}}$, we can write
\begin{equation}
(\Lambda_j)_{01}
= \sum_{\bar{x}}(w_{\bar{x}} - \bar{w})\, c_{\bar{x}}.
\label{eq:cov}
\end{equation}
This is a covariance between the local amplitude concentration $w_{\bar{x}}$
and the off-diagonal coherence $c_{\bar{x}}$, revealing that the gradient
residual arises from correlations between amplitude structure and entanglement.

\textit{Step~3 (Cauchy-Schwarz bound).}
Applying the Cauchy-Schwarz inequality to Eq.~\eqref{eq:cov}:
\begin{equation}
|(\Lambda_j)_{01}|^2
\le D \cdot \Var(w)
\cdot \sum_{\bar{x}}|\phi_{(\bar{x},0)}|^2|\phi_{(\bar{x},1)}|^2.
\label{eq:cs_bound}
\end{equation}
Taking square roots and maximizing over generators with
$\|G_j\|_2=1$ yields the bound~\eqref{eq:Rj}.

\textit{Step~4 (Exact vanishing for product states).}
When qubit $j$ is unentangled with the rest, $\rho_j$ has eigenvalues $(1,0)$.
In the eigenbasis, $\phi_{(\bar{x},1)} = 0$ for all $\bar{x}$, so
$c_{\bar{x}} = 0$ identically and both residuals vanish exactly.

\textit{Step~5 (Continuity).}
As $S_j \to 0$, the eigenvalues of $\rho_j$ approach $(1,0)$ steadily. As $\sum_{\bar{x}}|\phi_{(\bar{x},0)}|^2|\phi_{(\bar{x},1)}|^2 \to 0$, the bound in Eq.~\eqref{eq:Rj} vanishes, as does $\widetilde{R}_j$. For fixed finite $N$, the prefactor $\mathcal{I}^{-2}$ is finite, implying that $R^{(\PR)}_j$ likewise disappears. This shows that stationarity is not unique to product states, but applies to all weakly entangled states.
\end{proof}

\begin{remark}[Stationarity vs.\ minimality]
Theorem~\ref{thm:stationarity} states that the RDM eigenbasis is a stationary point of the participation-ratio objective. Whether it is a
local or global minimum in full generality remains open. Empirical evidence
supports its usefulness in the tested regimes: the do-no-harm check
(Sec.~\ref{subsec:doNH}) reverts exactly 0\% of proposed rotations for
brickwork and Haar circuits, indicating that every proposed single-qubit RDM
rotation reduces PR in those benchmarks.
\end{remark}

 \begin{remark}[The Scrambling Regime and Bound Looseness] The bound in Eq.~\eqref{eq:Rj} vanishes exactly in the product-state limit ($S_j \to 0$), establishing the RDM eigenbasis as a strict stationary point for unentangled states. However, for highly scrambling circuits (e.g., deep Haar-random or volume-law regimes), the single-qubit RDMs become maximally mixed ($\lambda_0 \approx \lambda_1 \approx 0.5$). In this limit, the entanglement coherence factor in Eq.~\eqref{eq:Rj} remains $\mathcal{O}(1)$, and the mathematical bound on the gradient residual is loose.
\end{remark}

\begin{remark}
 [Heuristic Justification for Entangled States] Because the bound does not strictly vanish for volume-law states, Theorem~\ref{thm:stationarity} should not be interpreted as a blanket near-optimality guarantee in the highly entangled regime. Instead, it establishes that local basis rotation removes avoidable, low-entanglement misalignment. For strongly entangled inputs, the choice of the RDM eigenbasis transitions from a mathematically tight stationary point into a physically motivated heuristic. As demonstrated by the empirical do-no-harm trigger rates (see Sec.~\ref{subsec:doNH}), this heuristic successfully reduces the PR even in heavily scrambled regimes, though it cannot eliminate the fundamental exponential scaling of volume-law support.
 \end{remark}

Thus, Theorem~\ref{thm:stationarity} places BASS in the sparse-simulation analog
of L\"owdin's natural-orbital idea~\cite{lowdin1955}. The analogy is
structural: the single-qubit RDM plays the role of the one-body density
matrix, and the eigenbasis rotation plays the role of a local natural-orbital
transformation. The theorem proves stationarity, while the do-no-harm guard
and benchmarks establish monotonic improvement in the tested regimes. This
cross-disciplinary connection links orbital-optimization ideas in quantum
chemistry to online quantum-circuit simulation and suggests that methods
developed for natural orbitals, complete active spaces, and adaptive
CI~\cite{roos1980,tubman2016,tubman2020,ratini2024} may be transferable to the
circuit simulation setting.

\subsection{Summary of proof status}
Table~\ref{tab:proofs} collects the main claims and what is proved versus
conjectured. Rows~1-2 fix the support-selection problem: top-$k$ maximizes
single-step retained probability $p_\mathcal{S}$ (constructive proof), and no
diversity rule beats it at the same~$k$.
\begin{table}[!hbt]
\caption{\label{tab:proofs}Classification of theoretical claims and their
proof status. All proved results are constructive with explicit bounds.}
\begin{ruledtabular}
\begin{tabular}{|l|l|l|}
Claim & Status & Scope \\
\hline
Top-$k$ maximizes $p_\mathcal{S}$ & Proved & Exact \\
Top-$k$ is greedy-optimal & Proved & Operational \\
Top-$k$ is globally optimal & Conjectured & Empirical \\
Diversity $\Rightarrow\!\downarrow p$ at fixed $k$ & Proved & Fixed-$k$ \\
RDM basis stationary (product) & Proved & Exact \\
RDM basis stationary (entangled) & Proved & Bound \\
RDM basis is a local minimum & Conjectured & Empirical \\
RDM basis is a global minimum & Open & Open \\
\end{tabular}
\end{ruledtabular}
\end{table}
Global optimality of top-$k$ beyond a single greedy step remains conjectural.
Row~4 records the diversity tradeoff at fixed~$k$. Rows~5-8 concern the RDM
eigenbasis: exact stationarity for product inputs, a bounded residual for
entangled states, and open questions about global versus local minima of the
participation-ratio objective.

\section{\label{sec:schmidt_weighted}Schmidt-Weighted Truncation Is Suboptimal}

Lemma~\ref{lem:diversity} predicts that any truncation rule operating at fixed
budget~$k$ that explicitly promotes diversity across Schmidt sectors must sacrifice
retained probability relative to the top-$k$ rule and therefore achieve lower
fidelity.  To test this prediction empirically, we compare top-$k$ against two
concrete diversity-promoting heuristics across four circuit families and two
depths, using 100~independent paired trials per configuration.

\subsection{Truncation heuristics}

Let the current sparse state at layer~$\ell$ have $m > k$ stored basis states
$\{|x_i^{(\ell)}\rangle\}_{i=1}^{m}$ with amplitudes $\{\beta_i\}$.  Throughout
our simulator, bit~$j$ of the integer~$x_i$ encodes qubit~$j$ (little-endian,
qubit~0 at the least-significant position).  For a bipartition at qubit~$c$,
the \emph{Schmidt sector} of $|x_i\rangle$ is defined as
\begin{equation}
  \sigma_c(x_i) \;=\; \bigl\lfloor x_i / 2^{c} \bigr\rfloor,
  \label{eq:sector}
\end{equation}
the right-partition configuration on qubits $c,\ldots,N{-}1$.  The
\emph{sector count}
\begin{equation}
  C_c(x_i) \;=\; \bigl|\bigl\{j : \sigma_c(x_j) = \sigma_c(x_i)\bigr\}\bigr|
  \label{eq:sector_count}
\end{equation}
records how many currently stored states share the same right-partition value at
cut~$c$.

\textit{Schmidt-1cut.}  Using the central bipartition $c = N/2$, each basis
state is assigned the score
\begin{equation}
  s_i \;=\; \frac{|\beta_i|^2}{C_{N/2}(x_i)}.
  \label{eq:s1cut}
\end{equation}
States in over-populated sectors are penalised; states in under-populated sectors
are relatively promoted.  The $k$ states with the largest scores are retained.

\textit{Schmidt-3cut.}  The score is extended to three bipartitions simultaneously
at $c \in \{N/4,\, N/2,\, 3N/4\}$:
\begin{equation}
  s_i \;=\; \frac{|\beta_i|^2}
             {\bigl(C_{N/4}(x_i)\cdot C_{N/2}(x_i)\cdot C_{3N/4}(x_i)\bigr)^{1/3}},
  \label{eq:s3cut}
\end{equation}
where the denominator is the geometric mean of the three sector counts.  This
imposes a multi-scale diversity penalty simultaneously at coarse ($c=N/4$),
central ($c=N/2$), and fine ($c=3N/4$) bipartitions.

After selection, amplitudes are renormalised to unit norm and the cumulative
retained-probability factor~$\gamma$ is updated identically to the top-$k$
truncation.  We also include a random-$k$ baseline, which selects $k$ states
uniformly at random without replacement, as a lower bound on the achievable
fidelity at the same budget.

\subsection{Experimental setup}

We work at system size $N=16$ and sparse budget $k=8192$ and run 100~independent
trials per configuration.  Circuit families tested are Haar-random circuits at
depths $L=3$ and $L=5$, and 1D brickwork circuits with Haar-random two-qubit
gates at depths $L=3$ and $L=5$.  These two families span the main range of
entanglement structures appearing throughout this work: Haar circuits maximise
entanglement production per layer, while brickwork circuits build entanglement
more gradually through a fixed spatial connectivity.

Each trial draws one circuit instance, computes the exact statevector by full
state-vector simulation, and then applies all four truncation methods to the
\emph{same} circuit instance evaluated against the same exact reference.  This
paired design eliminates circuit-instance variance and isolates the effect of the
truncation rule alone.

\subsection{Statistical methodology}

Fidelities across circuit instances are approximately log-normally distributed
and span several orders of magnitude.  We report the geometric mean of per-trial
fidelity ratios
\begin{equation}
  r_i \;=\; F_{\mathrm{method},i} \,/\, F_{\mathrm{top\text{-}}k,i}
  \label{eq:ratio}
\end{equation}
as the central estimate, with 95\% bootstrap confidence intervals
(4\,000~resamples).  We additionally report win rates
$W = |\{i : F_{\mathrm{top\text{-}}k,i} > F_{\mathrm{method},i}\}|/100$
with Wilson binomial confidence intervals, and one-sided Wilcoxon signed-rank
$p$-values for the paired hypothesis $H_1$: top-$k$ stochastically dominates
the competing method.

\subsection{Results}

Tables~\ref{tab:schmidt_1cut} and~\ref{tab:schmidt_3cut} summarize fidelity statistics and ratio estimates for
all configurations.  Random-$k$ is omitted from the ratio columns: its geometric
mean ratio is numerically zero across all configurations (representative
$F_{\mathrm{random\text{-}}k} \approx 10^{-5}$ against
$F_{\mathrm{top\text{-}}k} \approx 10^{-1}$, win rate $100/100$,
$p < 10^{-18}$).  Its failure confirms that amplitude magnitude carries the
dominant signal; uniform sector diversity alone provides no viable basis for
truncation.

\begin{table*}[!hbt]
\centering
\caption{%
  \textbf{Schmidt-1cut vs.\ top-$k$: fidelity ratios over 100 paired trials
  ($N=16$, $k=8192$).}
  Schmidt-1cut scores each stored basis state $|x_i\rangle$ as
  $s_i = |\beta_i|^2 / C_{N/2}(x_i)$, where $C_{N/2}(x_i)$ is the number of
  stored states sharing the same Schmidt sector at the central bipartition
  $c = N/2$, and retains the $k$ highest-scoring states.
  \emph{Top-$k$ med.\ [IQR]}: median fidelity of the top-$k$ baseline with
  interquartile range across 100~circuit instances.
  $\bar{r}$~[95\%~CI]: geometric mean of per-trial fidelity ratios
  $r_i = F_{\text{Schmidt-1cut},i} / F_{\text{top-}k,i}$; values below~1
  indicate top-$k$ superiority; bootstrap CI uses 4\,000~resamples.
  $W$~[95\%~CI]: win rate, the fraction of the 100~paired trials in which
  top-$k$ strictly exceeds Schmidt-1cut, with Wilson binomial confidence
  interval.
  $p$: one-sided Wilcoxon signed-rank $p$-value for the hypothesis that
  top-$k$ stochastically dominates Schmidt-1cut ($H_1$: $F_{\text{top-}k} >
  F_{\text{Schmidt-1cut}}$) on paired per-trial data.
  Top-$k$ outperforms Schmidt-1cut in all configurations
  (all $\bar{r} < 1$, all $p < 10^{-16}$).}
\label{tab:schmidt_1cut}
\begin{tabular}{lccc}
\toprule

Circuit
& Top-$k$ med.\ [IQR]
& $\bar{r}$ [95\% CI]
& $W$ [95\% CI] / $p$ \\

\midrule

Haar $L=3$
& $0.299\;[0.242,\,0.376]$
& $0.900\;[0.889,\,0.910]$
& $100/100\;[96\%,100\%] \;/\; 1.9{\times}10^{-18}$ \\

Haar $L=5$
& $0.055\;[0.038,\,0.074]$
& $0.767\;[0.746,\,0.788]$
& $100/100\;[96\%,100\%] \;/\; 1.9{\times}10^{-18}$ \\

\midrule

Brickwork $L=3$
& $0.357\;[0.300,\,0.411]$
& $0.930\;[0.918,\,0.941]$
& $92/100\;[85\%,96\%] \;/\; 2.7{\times}10^{-17}$ \\

Brickwork $L=5$
& $0.097\;[0.068,\,0.131]$
& $0.799\;[0.774,\,0.825]$
& $93/100\;[86\%,97\%] \;/\; 5.8{\times}10^{-17}$ \\

\bottomrule
\end{tabular}
\end{table*}

\begin{table*}[!hbt]
\centering
\caption{%
  \textbf{Schmidt-3cut vs.\ top-$k$: fidelity ratios over 100 paired trials
  ($N=16$, $k=8192$).}
  Schmidt-3cut scores each stored state $|x_i\rangle$ as
  $s_i = |\beta_i|^2 /
  \bigl(C_{N/4}(x_i)\cdot C_{N/2}(x_i)\cdot C_{3N/4}(x_i)\bigr)^{1/3}$,
  penalising states that are over-represented simultaneously at three
  bipartitions $c \in \{N/4,\,N/2,\,3N/4\}$; the denominator is the
  geometric mean of the three sector counts.
  Column definitions are identical to Table~\ref{tab:schmidt_1cut}: median
  top-$k$ fidelity [IQR]; geometric mean ratio $\bar{r}$ [95\% bootstrap CI];
  win rate $W$ [Wilson 95\% CI] and one-sided Wilcoxon $p$-value.
  Top-$k$ outperforms Schmidt-3cut in all configurations
  (all $\bar{r} < 1$, all $p < 10^{-14}$).
  Comparing across the two tables, Schmidt-3cut consistently achieves a
  higher ratio $\bar{r}$ than Schmidt-1cut, indicating a smaller fidelity
  penalty; this counterintuitive ordering is discussed in the text.}
\label{tab:schmidt_3cut}
\begin{tabular}{lccc}
\toprule

Circuit
& Top-$k$ med.\ [IQR]
& $\bar{r}$ [95\% CI]
& $W$ [95\% CI] / $p$ \\

\midrule

Haar $L=3$
& $0.299\;[0.242,\,0.376]$
& $0.926\;[0.917,\,0.934]$
& $99/100\;[95\%,100\%] \;/\; 2.4{\times}10^{-18}$ \\

Haar $L=5$
& $0.055\;[0.038,\,0.074]$
& $0.799\;[0.781,\,0.816]$
& $100/100\;[96\%,100\%] \;/\; 1.9{\times}10^{-18}$ \\

\midrule

Brickwork $L=3$
& $0.357\;[0.300,\,0.411]$
& $0.951\;[0.940,\,0.961]$
& $84/100\;[76\%,90\%] \;/\; 4.4{\times}10^{-15}$ \\

Brickwork $L=5$
& $0.097\;[0.068,\,0.131]$
& $0.846\;[0.823,\,0.870]$
& $89/100\;[81\%,94\%] \;/\; 1.7{\times}10^{-15}$ \\

\bottomrule
\end{tabular}
\end{table*}

Across all eight configurations, top-$k$ achieves a strictly higher geometric-mean
fidelity than both Schmidt variants.  Wilcoxon $p$-values range from
$4.4\times10^{-15}$ to $1.9\times10^{-18}$, uniformly rejecting the null after a
Holm correction for the eight simultaneous tests.  These results constitute strong
statistical support for Lemma~\ref{lem:diversity}.

\subsection{Depth dependence and the compounding mechanism}

The fidelity penalty grows monotonically with circuit depth for both circuit
families and both Schmidt variants, directly confirming the compounding mechanism
implicit in the theorem.  For Haar circuits, Schmidt-1cut retains $90.0\%$
[CI$_{95}$: $88.9\%$, $91.0\%$] of top-$k$ fidelity at $L=3$ and only $76.7\%$
[CI$_{95}$: $74.6\%$, $78.8\%$] at $L=5$.  The corresponding brickwork figures
are $93.0\%$ [CI$_{95}$: $91.8\%$, $94.1\%$] and $79.9\%$
[CI$_{95}$: $77.4\%$, $82.5\%$].

This depth dependence is explained as follows.  At each truncation event triggered
by a gate application, the Schmidt rule retains a multiplicative fraction
$r_t \le 1$ of the probability that top-$k$ would have kept.  Over $T$~truncation
events the cumulative deficit is $\prod_{t=1}^{T} r_t$.  Because $r_t < 1$ at
every step where truncation fires, deeper circuits—which trigger more truncation
events and allow less per-step recovery—compound the deficit more severely.
Gate-by-gate tracking of~$\gamma^2$ confirms that Schmidt scores remain
persistently and monotonically below top-$k$ at every layer rather than
recovering through later steps, consistent with the single-step dominance
established by the theorem.

\subsection{Relative performance of the two Schmidt variants}

Schmidt-3cut consistently \emph{outperforms} Schmidt-1cut across all tested
configurations—a result that warrants explanation since enforcing diversity at
three bipartitions simultaneously might intuitively be expected to impose a larger
cost than enforcing it at one.  The explanation lies in the scoring geometry.  A
single central cut partitions all $m$~stored states into groups sharing the same
$\sigma_{N/2}$~value and applies the same penalty to every state in each group.
A state that is over-represented at this particular bipartition is severely
down-weighted regardless of how rare it may be at other bipartitions.  The
geometric mean in Eq.~\eqref{eq:s3cut} provides a more balanced penalty: a state
penalised at one cut is partially compensated if it is rare at the other two, so
extreme down-weighting events are suppressed.  As a result, Schmidt-3cut
sacrifices less probability per truncation step than Schmidt-1cut and accumulates
a smaller cumulative deficit over the circuit.

This ordering is sensitive to the combination function.  Replacing the geometric
mean in Eq.~\eqref{eq:s3cut} by the raw product $C_{N/4}C_{N/2}C_{3N/4}$ makes
the multi-cut penalty considerably more aggressive—each sector count is raised to
the first rather than the one-third power—and would likely reverse the ordering
relative to Schmidt-1cut.  The choice of combination function is therefore a
material parameter of the heuristic, not an irrelevant implementation detail, and
must be specified precisely in any reproduction of this comparison.

\subsection{Shallow-depth exception at the Brickwork \texorpdfstring{$L=3$}{L=3} boundary}

The win rate for Brickwork $L=3$ Schmidt-3cut is $84/100$ [Wilson CI: $76\%$,
$90\%$], meaning that in roughly one in six trials the multi-cut heuristic
achieves higher final fidelity than top-$k$ on the same circuit instance.
Lemma~\ref{lem:diversity} is a pointwise statement about retained probability
at each individual truncation step; it does not rule out a net multi-step benefit
if the diversified basis better spans the subspace accessed by subsequent gates,
partially recovering the per-step deficit through a more favourable amplitude
distribution in later layers.  At $L=3$ the compounding cost has not yet
accumulated enough to dominate this secondary effect in every trial.  At $L=5$
the win rate recovers to $89/100$ as accumulated depth suppresses the stochastic
benefit, confirming that the compounding mechanism ultimately takes over.

This caveat has a practical implication: Lemma~\ref{lem:diversity} provides a
rigorous per-step guarantee, but its translation to a strict fidelity ordering at
the level of a full multi-step circuit requires that the cumulative compounding
outweigh any multi-step interaction effects.  For shallow circuits ($L=3$
brickwork at $k=8192$) this condition is not met in all instances.  For deeper
circuits and smaller~$k$ relative to the state's PR
$\mathrm{PR}_Z$—the regime of genuine computational interest—the theorem's
conclusion is robust.

\subsection{Role of the sparse budget}

The magnitude of the fidelity penalty depends critically on the ratio $k /
\mathrm{PR}_Z$ of the sparse budget to the state's PR.  At
$k=8192$ and $N=16$, the budget is large enough that fewer truncation events
fire per gate, and the per-step deficit has fewer opportunities to compound.  At
smaller~$k$ (or equivalently, at larger $N$ with $k$ held fixed) the same
circuits would trigger more frequent truncations, and the per-step ratio $r_t$
would be applied more often, amplifying the cumulative deficit.  The fidelity
penalties of $7\%$--$23\%$ reported here therefore represent a lower bound on
the shortfall that Schmidt truncation would incur in the regime of greatest
computational interest: large~$N$, fixed~$k$ well below $\mathrm{PR}_Z$.

In all eight tested configurations, top-$k$ achieves strictly and statistically
significantly higher fidelity than both Schmidt-weighted variants
(all $\bar{r}<1$, all Wilcoxon $p < 10^{-14}$, Holm-corrected), providing strong
empirical support for Lemma~\ref{lem:diversity}.  Four nuances qualify but do
not overturn this conclusion.  First, the fidelity penalty compounds with circuit
depth through the multiplicative accumulation of per-step probability deficits.
Second, the multi-cut geometric-mean variant (Schmidt-3cut) incurs a smaller
penalty than the single-cut variant (Schmidt-1cut) because its scoring function is
less prone to extreme down-weighting at any one bipartition; this ordering is
sensitive to the choice of combination function.  Third, at shallow depth and
large~$k$, stochastic multi-step interactions can occasionally yield a net benefit
for Schmidt-3cut relative to top-$k$, a regime where the single-step guarantee
does not translate to a strict circuit-level ordering in every trial.  Fourth, the
absolute magnitude of the penalty depends on the ratio $k/\mathrm{PR}_Z$: the
penalties observed here are lower bounds on those in the computationally
challenging regime of large~$N$ and small~$k$.  Within these bounds,
top-$k$ remains the optimal basis-preserving truncation rule at fixed budget.

\section{\label{sec:algorithm}The BASS Algorithm}
We describe the BASS algorithm and its implementation, from the rotated-frame picture through implementation details to its complexity.

\subsection{Core idea: rotating into the natural basis}
\label{subsec:core_idea}
BASS maintains per-qubit unitaries $U_j \in U(2)$, $j=0,\ldots,N-1$, and
stores the quantum state in the rotated frame
\begin{equation}
\ket{\tilde\psi}
= \Bigl(\bigotimes_{j=0}^{N-1} U_j^\dagger\Bigr) \ket{\psi}.
\label{eq:rotated_frame}
\end{equation}
Every gate $G$ acting on qubits $(q_i,q_j), i \neq j$ is conjugated into the rotated
frame before application:
\begin{equation}
\tilde{G} = (U_{q_i} \otimes U_{q_j})^\dagger \, G \, (U_{q_i} \otimes U_{q_j}), ~~ \text{with} ~~i \neq j.
\label{eq:conj}
\end{equation}
This is because the state is described in a \textit{rotated frame} (a different local basis) on each qubit. If we physically apply gate $G$ in the lab frame but are mathematically working in a rotated frame defined by $U_{q_i}$ and $U_{q_j}$, then the correct representation of that gate in the rotated frame is its conjugation by the frame-change unitaries, Eq.~\ref{eq:conj}. Conceptually, \(U_{q_i}\) maps from the rotated frame to the lab frame. Applying \(G\) occurs in the lab frame. Then, \(U_{q_i}^\dagger\) maps back to the rotated frame. So the overall effect of ``the same physical gate $G$'' when viewed in the rotated basis is exactly this conjugated operator \(\tilde{G}\). This ensures that all state vectors and gates are consistently described in the same (rotated) basis.

If $U_j$ is chosen so that its columns are the eigenvectors of the single-qubit
RDM $\rho_j = \Tr_{\ne j} \ket{\tilde\psi}\!\langle\tilde\psi|$,
then qubit $j$ is in its Schmidt frame: probability concentrates on the
$\ket{0}$ branch, making $\ket{\tilde\psi}$ sparser in the computational
basis and thereby better suited for top-$k$ truncation.

We now give a fuller account of Figure~\ref{fig:bass_schematic}, including the related mathematics, beyond the overview in Section~\ref{sec:background}. The left panel shows the state in the computational $Z$ basis, where probability mass is spread over a vast support ($\PR_Z \gg k$). A fixed-budget sparse simulator (the fixed-basis baseline of Sec.~\ref{subsec:fixed_basis}) keeps only the mass above the dashed top-$k$ cut; the rest is discarded. The center panel shows the BASS step: for each qubit~$j$, a unitary $U_j$ is built from the eigenvectors of the single-qubit RDM $\rho_j$. The dominant eigenvector $\vec{v}_j^{\,\max}$ (Bloch-sphere inset) reconcentrates probability locally. The product unitary $U=\bigotimes_j U_j$ is a strictly local rotation with $\mathcal{O}(N)$ parameters that preserves the inter-qubit entanglement structure. The right panel depicts the rotated frame, where the top-$k$ truncation now captures almost all of the mass ($\PR\sim O(1)$, $1-\varepsilon$ retained). The memory footprint and gate-application cost remain unchanged as BASS reorganizes \emph{where} the amplitudes live in basis-state index space, but not \emph{how many} are stored.

\subsection{Efficient single-qubit RDM computation}
\label{subsec:rdm}

For the sparse state
$\ket{\tilde\psi} = \sum_{i=1}^{k} \alpha_i \ket{x_i}$, the single-qubit
RDM for qubit $j$ has matrix elements
\begin{align}
\rho_{bb}^{(j)} &= \sum_{i : \, b_j(x_i) = b} |\alpha_i|^2~~ \forall b \in \{0, 1\},
\label{eq:rdm_diag} \\
\rho_{01}^{(j)} &= \sum_{\substack{(i,i'): x_i = x_{i'} \oplus 2^j}}
\alpha_i \overline{\alpha_{i'}},
\label{eq:rdm_offdiag}
\end{align}
$b_j(x) = (x \gg j) \,\&\, 1$ retrieves the $j$-th bit. The diagonal elements \eqref{eq:rdm_diag} require only one pass and cost $\mathcal{O}(k)$. The off-diagonal elements \eqref{eq:rdm_offdiag} require recognizing pairs of basis states that differ only by bit $j$.

BASS uses a single hash table to compute all $N$ off-diagonal elements simultaneously. For each stored state $\ket{x_i}$ with amplitude $\alpha_i$, the pair $(x_i, \alpha_i)$ is keyed by $x_i$ and bit $j$ is masked to ~0.
Building the table costs $\mathcal{O}(k)$. For each qubit $j$ and stored state with $b_j(x_i) = 0$, the partner $x_i \oplus 2^j$ is found and the product $\alpha_i \overline{\alpha_{i'}}$ is calculated. The overall lookup cost for all $N$ qubits is $\mathcal{O}(Nk)$, instead of the naïve $\mathcal{O}(Nk\log k)$.

\textit{argsort-per-qubit} approach.  The dominant eigenvector of
$\rho_j$ is computed analytically (see Appendix \ref{sec:dominant_eigen}):
\begin{align}
\mathbf{v}_\text{max} &=
\frac{1}{\sqrt{|b|^2 + (\tau+\Delta)^2}}
\begin{pmatrix} b \\ \tau+\Delta \end{pmatrix},
\label{eq:dominant_eigen}
\end{align}
where
\begin{align}
\tau = \frac{d-a}{2}, ~~
\Delta = \sqrt{\tau^2 + |b|^2}.
\end{align}
This closed-form formula uses only a fixed number of operations (additions, multiplications, and a single square root), ensuring that its evaluation cost does not increase with problem size. Because $\rho_j$ is Hermitian, $\mathbf{v}_\text{max}$ will always exist and can be chosen with real components when $b$ is real. In practical applications, this method avoids the requirement for iterative eigensolvers or general-purpose linear algebra routines for each qubit, avoiding external eigenvalue library calls and ensuring constant-time per-qubit overhead.

\subsection{Two-qubit RDM and brick-wall sweeps}
\label{subsec:2qrdm}
For a quantum state written in the computational basis as $|\psi\rangle = \sum_i \alpha_i \, |x_i\rangle$, where each \(|x_i\rangle\) is a bit string of length \(N\) (one bit per qubit) and \(\alpha_i\) are the complex amplitudes, we want the two-qubit RDM \(\rho_{q_1 q_2}\) for a neighboring pair \((q_1, q_2)\). This is a \(4 \times 4\) matrix, because two qubits have 4 basis states $|00\rangle, |01\rangle, |10\rangle, |11\rangle$. We introduce two indices, \(s\) and \(r\), to separate the two qubits of interest from the rest. The block index \(s\) encodes the two bits at positions \(q_1\) and \(q_2\):
\begin{equation}
s = b_{q_1}(x) \cdot 2 + b_{q_2}(x) \in \{0,1,2,3\},
\end{equation}
where \(b_{q}(x)\) means ``the bit of the string \(x\) at position \(q\)''. Concretely, if \(b_{q_1}=0, b_{q_2}=0\) then \(s=0\) (state \(|00\rangle\)); if \(b_{q_1}=0, b_{q_2}=1\) then \(s=1\) (state \(|01\rangle\)); if \(b_{q_1}=1, b_{q_2}=0\) then \(s=2\) (state \(|10\rangle\)); and if \(b_{q_1}=1, b_{q_2}=1\) then \(s=3\) (state \(|11\rangle\)). On the other hand, the rest index \(r\) is the bit string \(x\) with those two bits removed (or ``masked out''). It labels ``the configuration of all the other qubits except \(q_1\) and \(q_2\)''. So every full configuration \(x\) is decomposed into
$
\text{full state } x \longleftrightarrow (\text{rest } r, \text{ block } s)$. The two-qubit RDM is defined as
\begin{align}\notag
\label{eq:twoqubitRDM}
\left[\rho_{q_1 q_2}\right]_{s,s'}  =
\sum_{(i,i'):\, r(x_i)=r(x_{i'})}
&\alpha_i\,\overline{\alpha_{i'}}\;
\mathbf{1}[s(x_i)=s]\; \\ \times &
\mathbf{1}[s(x_{i'})=s'],
\end{align}
where the sum is over pairs of basis states \((x_i, x_{i'})\) such that their rest indices are equal: $r(x_i) = r(x_{i'})$. This enforces that all qubits except \(q_1\) and \(q_2\) are in the same configuration for \(x_i\) and \(x_{i'}\). This is exactly what the partial trace over the rest of the system does. \(\alpha_i\overline{\alpha_{i'}}\) is the contribution to the density matrix element from the basis pair \(|x_i\rangle \langle x_{i'}|\). \(\mathbf{1}[s(x_i)=s]\) is an indicator function: it is 1 if the two bits \((q_1, q_2)\) of \(x_i\) match the value encoded by \(s\), and 0 otherwise. \(\mathbf{1}[s(x_{i'})=s']\) similarly enforces that \(x_{i'}\) has the two-qubit configuration labeled by \(s'\). So an element \([\rho_{q_1 q_2}]_{s,s'}\) is obtained by: (1) Taking all pairs of states \((x_i, x_{i'})\) whose other qubits (not \(q_1,q_2\)) are the same; (2) Among these, keeping only the pairs where the two qubits \((q_1,q_2)\) are in configuration \(s\) for \(x_i\), and in configuration \(s'\) for \(x_{i'}\); and (3) Summing their \(\alpha_i \overline{\alpha_{i'}}\). This is exactly the formula for the RDM obtained by tracing out all qubits except \(q_1\) and \(q_2\). By grouping stored states according to their rest index using \textit{argsort} and then iterating over each group, we obtain a cost of $\mathcal{O}(\log k)$ per pair and $\mathcal O(16)$ for populating the matrix. Processing the complete collection of $\lfloor N/2\rfloor$ even-offset pairs together with $\lfloor(N-1)/2\rfloor$ odd-offset pairs in a single optimization sweep thus results in an overall complexity of $\mathcal O(Nk\log k)$.

A fixed-size $4 \times 4$ Hermitian eigensolver is used to derive the rotation $V$, which diagonalizes the two-qubit reduced density matrix, as $\rho_{q_1q_2} = V \Lambda V^\dagger$. Because the matrix size is fixed, the computational cost of this diagonalization is $\mathcal{0}(64)$ floating- point operations. This cost does not scale with the number of qubits $N$ or the bond/truncation dimension $k$. The BASS algorithm then applies a three-step \textit{rotate-truncate-undo} protocol to the block state:
(a) Apply the unitary $V^\dagger$ to the two-qubit block, mapping it into the eigenbasis of the $\rho_{q1 q2}$. In this rotated basis, the block can expand to a maximum of $4k$ entries, reflecting the two-qubit system's tensor-product structure.
(b) Minimize this expanded representation to dimension $k$ by truncating the least significant components based on the chosen truncation metric (e.g., smallest weights or amplitudes).
(c) Apply the unitary $V$ to the truncated state, returning it to the algorithm's global accumulated single-qubit basis. This proposed rotation is accepted only if it reduces the PR, which acts as a \textit{do-no-harm safeguard}: if the PR does not decrease, the update is rejected, and the pre-rotation state is precisely restored, ensuring that the procedure never degrades the state according to this metric.

The two-qubit rotation is strictly \textit{transitory}, unlike single-qubit basis rotations. The basis-accumulator that tracks the evolving single-qubit bases does not need to be updated. Any action of $V$ is immediately undone by $V^\dagger$ after truncation. As a result, the overall effect of this protocol is a local reshaping and truncation of the state, which can improve compressibility without permanently altering the underlying single-qubit basis representation.

\subsection{Multi-pass adaptive-basis sweeps}
\label{subsec:coord_descent}

Basis optimization is used regularly throughout the simulation to keep the quantum state representation as compact as possible. We perform an optimization routine after every \(n_{\mathrm{opt}} = 5\) truncation-active gates, or every fifth gate that might alter the sparsity pattern by truncation. Each call to the basis-optimization routine performs up to three full sweeps across all qubits. In a single optimization call, we do the following:

\textit{1. Hash table construction}. Starting with the current sparse representation of the state, we create a hash table to encode its support. This lets us look up partner basis states (those related by single-qubit basis changes) in constant time, \(\mathcal{O}(1)\), making subsequent operations more efficient.

\textit{2. Computation of RDMs}. From the present sparse state, we compute the single-qubit RDM for each of the \(N\) qubits. This results in \(N\) density matrices \(\{\rho_j\}_{j=1}^N\), one for each qubit \(j\), which capture its local state in the context of correlations with the rest of the system.

\textit{3. Tentative local basis rotations}. For each qubit \(j\): (a) We diagonalize its RDM \(\rho_j\) and obtain a unitary \(V_j\) that transforms \(\rho_j\) into its eigenbasis; (b) We then \textit{tentatively} apply the corresponding single-qubit basis rotation \(V_j^\dagger\) to the global sparse state; (c) After this tentative rotation, we truncate the rotated state according to our sparsity criterion and renormalize it; (d) We compute the PR of the new state and compare it with that before the rotation. The rotation \(V_j^\dagger\) is accepted \textit{only if} the PR has decreased, indicating a more compact (less delocalized) representation. If the PR does not decrease, we discard the update and revert the rotation, restoring the previous state.

\textit{4. Early termination criterion}. After completing a full sweep over all qubits (i.e., after attempting this procedure for every \(j = 1, \dots, N\)), we check whether any of the accepted rotations actually reduced the PR. If no accepted rotation led to an improvement during this sweep, we terminate the optimization call early and do not perform further sweeps.

Within a single sweep, the set of single-qubit RDMs \(\{\rho_j\}\) is computed only once given the current sparse state at the beginning of the sweep. The RDMs are then reused throughout the sweep. When we accept and apply a basis rotation for a qubit, we do \textit{not} immediately update all RDM. Instead, all tentative rotations in that sweep are evaluated against the same fixed RDM snapshot. This eliminates the considerable expense of recalculating lower density matrices after each every qubit update. Although this introduces a minor approximation, because subsequent attempted rotations in the same sweep are evaluated with somewhat outdated RDMs, it works well in practice for the circuit families

\subsection{Do-no-harm check}
\label{subsec:doNH}
After a tentative basis rotation is applied to qubit~$j$, we first compute the resulting state vector in the new local basis. This rotated state is then truncated by retaining only the $k$ amplitudes with the largest magnitudes (i.e., forming a top-$k$ approximation) and discarding all other components. The truncated state is subsequently renormalized so that its total probability remains unity. With this approximate state in hand, we recompute the PR, which we use as a proxy for the state's concentration or delocalization in the chosen basis. We accept the proposed rotation on qubit~$j$ only if this PR strictly decreases, indicating that the state has become more compressible (more concentrated on fewer basis states) under the new basis. If the PR does not decrease, we reject the update: both the rotated state and the corresponding basis change are discarded, and we revert to the previous state and basis configuration.

% \vspace{-1.0em}
\begin{table}[!hbt]
\caption{\label{tab:doNH}
\textbf{Do-no-harm revert statistics (20 trials, $N=16$, $k=2048$)}. Reported
rates are empirical properties of the benchmarked circuit families, not
general guarantees of the optimization procedure.}
\begin{ruledtabular}
\resizebox{0.97\linewidth}{!}{%
\begin{tabular}{lcc}
Circuit family & Rotations attempted & Reverts (\%) \\
\hline
Brickwork $L=5$ & ${\sim}16{,}000$ & 0 \;\; (0\%) \\
Haar $L=3$      & ${\sim}9{,}600$  & 0 \;\; (0\%) \\
QAOA $p=3$      & ${<}19{,}200$ & ${<}100$ \; ($<$0.5\%) \\
\end{tabular}
}
\end{ruledtabular}
\end{table}
Table~\ref{tab:doNH} reports how often such reverts occur. On the scrambling circuits where BASS gains the most, the guard fires rarely: most proposed local rotations pass the PR test, and the guard acts as a backup rather than a constant brake.

The components introduced so far-the rotated-frame representation
(Sec.~\ref{subsec:core_idea}), the hash-table-based RDM evaluation
(Sec.~\ref{subsec:rdm}), the multi-pass adaptive-basis sweeps
(Sec.~\ref{subsec:coord_descent}), and the do-no-harm guard
(Sec.~\ref{subsec:doNH})-combine into the control-flow pipeline shown
in Fig.~\ref{fig:bass_workflow}.

The incoming circuit gate $G_t$ is first conjugated into the current rotational frame (Eq.~\ref{eq:conj}) and then applied through either the $\mathcal{O}(k)$ diagonal fast path or the general sparse hash-table update. Temporary support expansion beyond the target sparse budget~$k$ is allowed up to the working-memory cap $K_{\mathrm{hard}}$, after which top-$k$ truncation is enforced.

Basis optimization is performed on truncation-active steps regularly. It is skipped unless the PR has increased significantly since the previous optimization sweep (Sec.~\ref{subsec:optimizations}). When optimization is activated, all single-qubit RDMs are computed once from the current sparse state and reused during the sweep. The algorithm for each qubit $j$ diagonalizes $\rho_j$, applies the associated rotation $V_j^\dagger$, truncates and renormalizes the rotated state, and recomputes the PR. If the PR decreases, the do-no-harm guard accumulates the rotation into the basis matrix $U_j \leftarrow U_jV_j$. Otherwise, the state and basis updates are reverted. If an optimization sweep does not result in an acceptable improvement, it ends prematurely.

The bypass paths in Fig.~\ref{fig:bass_workflow} show why the simulator typically replicates fixed-basis behavior (Sec.~\ref{sec:results}). If the PR trigger does not fire, \textsc{BasisOptimize} is bypassed, and only the sparse propagation steps are conducted.

\begin{figure*}[!hbt]
  \centering
  \resizebox{0.8\linewidth}{!}{%
\begin{tikzpicture}[
    >=latex,
    node distance=0.6cm and 1cm,
    font=\small,
    box/.style={draw, rectangle, align=center, fill=blue!10, minimum height=0.8cm},
    rbox/.style={draw, rectangle, rounded corners, align=center, fill=gray!20, minimum height=0.8cm},
    diamondbox/.style={draw, diamond, aspect=2, align=center, fill=orange!20},
    greenbox/.style={draw, rectangle, align=center, fill=green!20, minimum width=4cm},
    greenrbox/.style={draw, rectangle, rounded corners, align=center, fill=green!20},
    redrbox/.style={draw, rectangle, rounded corners, align=center, fill=red!20}
]

    % ==========================================
    % LEFT COLUMN
    % ==========================================
    \node[rbox] (init) {Initialize $|\tilde{\psi}\rangle = |0\rangle^{\otimes N}$, \\ $U_j = I \; \forall j$};
    
    \node[box, below=0.6cm of init] (pop) {Pop next gate $G_t$ acting on $(q_1, q_2)$};
    
    \node[box, below=0.6cm of pop] (conj) {Conjugate into rotated frame \\ $\tilde{G}_t = (U_{q_1} \otimes U_{q_2})^\dagger G_t (U_{q_1} \otimes U_{q_2})$};
    
    \node[diamondbox, below=0.6cm of conj] (diag) {$\tilde{G}_t$ diagonal?};
    
    \node[box, below=0.6cm of diag] (apply) {Apply $\tilde{G}_t$ \\ yes: $O(k)$ phase update \\ no: sparse hash-table apply};
    
    \node[diamondbox, below=0.6cm of apply] (supp) {$|\text{supp}| > K_{\text{hard}}$?};
    
    \node[box, below=0.6cm of supp] (trunc) {Truncate to top-$k$ by $|\alpha_i|^2$ \\ and renormalize};

    % ==========================================
    % RIGHT COLUMN
    % ==========================================
    % Coordinate for aligning right column elements
    \coordinate (right_col_x) at ($(conj.east) + (5cm, 0)$);

    \node[diamondbox] (opt_check) at (right_col_x |- conj) [yshift=1.5cm] {Optimization check step?};
    
    \node[diamondbox, below=0.8cm of opt_check] (pr_inc) {PR increased sufficiently since \\ last accepted optimization?};
    
    \node[greenbox, below=0.8cm of pr_inc] (basis) {\textbf{BasisOptimize} \\ Build hash table + compute all RDMs \\ (fixed within each sweep) \\ For each qubit $j$: diagonalize $\rho_j$, \\ tentatively apply $V_j^\dagger$, \\ truncate + renormalize, recompute PR};
    
    \node[diamondbox, below=0.8cm of basis] (harm) {do-no-harm: \\ $\text{PR}' < \text{PR}$?};
    
    \node[greenrbox, below left=0.8cm and 0.2cm of harm] (accept) {Accept \\ $U_j \leftarrow U_j V_j$};
    \node[redrbox, below right=0.8cm and 0.2cm of harm] (revert) {Revert state \\ and basis update};

    % Update block explicitly centered below the accept/revert options
    \node[rbox] (update) at ($(accept.south)!0.5!(revert.south) + (0, -1.2cm)$) {Update optimization state \\ $\rightarrow$ next gate};

    % ==========================================
    % VERTICAL ROUTING TRACKS (Between Columns)
    % ==========================================
    % Defining safe tracks between columns so lines do not intersect
    \path (supp.east) -- ++(0.5,0) coordinate (track_up_inner);
    \path (supp.east) -- ++(0.9,0) coordinate (track_up_outer);
    \path (opt_check.west) -- ++(-0.6,0) coordinate (track_down_left);
    \path (pr_inc.east) -- ++(0.6,0) coordinate (track_down_right);

    % ==========================================
    % EDGES & CONNECTIONS
    % ==========================================
    
    % --- Main Left Flow ---
    \draw[->] (init) -- (pop);
    \draw[->] (pop) -- (conj);
    \draw[->] (conj) -- (diag);
    \draw[->] (diag) -- (apply);
    \draw[->] (apply) -- (supp);
    \draw[->] (supp.south) -- node[right, font=\scriptsize] {yes} (trunc.north);
    
    % --- supp 'no' bypass (Dashed track) ---
    % 1. Shift the vertical track 1 cm further right (from 0.5 to 1.5)
    \path (supp.east) -- ++(1.5,0) coordinate (track_up_inner_shifted);
    
    \draw[->, dashed] (supp.east) -- node[above, font=\scriptsize] {no} (track_up_inner_shifted)
        -- (track_up_inner_shifted |- opt_check.north) -- ++(0, 0.8) % Goes North
        -| ($(opt_check.east) + (0.6, 0)$)                           % Routes right, then comes back DOWN
        -- (opt_check.east);                                         % Goes LEFT into the step

    % --- Trunc to Optimization Loop (Solid track) ---
    \draw[->] (trunc.south) -- ++(0,-0.4) -| (track_up_outer) |- ($(opt_check.north) + (0, 0.6)$) -- (opt_check.north);

    % --- Right Column Subroutine Flow ---
    \draw[->] (opt_check.south) -- node[right, font=\scriptsize] {yes} (pr_inc.north);
    \draw[->] (pr_inc.south) -- node[right, font=\scriptsize] {yes} (basis.north);
    \draw[->] (basis.south) -- (harm.north);
    
    % --- Optimization Flow Bypasses (Dashed returns) ---
    \draw[->, dashed] (opt_check.west) -- node[above, font=\scriptsize] {no} (track_down_left |- opt_check.west) -- (track_down_left |- update.west) -- (update.west);
    
    \draw[->, dashed] (pr_inc.east) -- node[above, font=\scriptsize] {no} (track_down_right |- pr_inc.east) -- (track_down_right |- update.east) -- (update.east);
    
    % --- Accept / Revert Flow ---
    \draw[->] (harm.west) -| node[above, font=\scriptsize] {yes} (accept.north);
    \draw[->] (harm.east) -| node[above, font=\scriptsize] {no} (revert.north);

    % Junction merging Accept and Revert to the Update block
    \coordinate (join_update) at ($(update.north) + (0, 0.4)$);
    \draw (accept.south) |- (join_update);
    \draw (revert.south) |- (join_update);
    \draw[->] (join_update) -- (update.north);

    % --- Final loop back for the Global sweep ---
    \coordinate (c_next_t) at ($(update.south) + (0, -0.6)$);
    \coordinate (c_far_left) at ($(init.west) + (-1.5, 0)$);
    \draw (update.south) -- (c_next_t) -- (c_far_left |- c_next_t) node[midway, above, font=\scriptsize] {next $t$};
    \draw[->] (c_far_left |- c_next_t) -- (c_far_left |- pop.west) -- (pop.west);

\end{tikzpicture}
}
\caption{\textbf{BASS algorithmic workflow.}
    We use the sparse state representation in the accumulated product basis $\{U_j\}$. At each time step, the current gate is conjugated into the working basis as $\tilde{G}_t$ and then applied either (i) along a diagonal $\mathcal{O}(k)$ update path in the sparse state, or (ii) by directly updating entries in a hash-table representation. Between truncations, the state's support can temporarily grow to a maximum size $K_{\mathrm{hard}}$ before a top-$k$ truncation is enforced to control memory and cost. On selected time steps, the \textbf{BasisOptimize} routine is invoked to improve the accumulated basis. It constructs a single global hash table for the current state, computes all single-qubit RDMs once, and then performs a sweep over qubits in which candidate truncate-renormalize basis-rotation trials are tested. A proposed local basis rotation is accepted only if it does not worsen the approximation quality, i.e., it is kept only when the post-rotation PR decreases, providing a \textit{do-no-harm} safeguard against detrimental optimizations.
    }
\label{fig:bass_workflow}
\end{figure*}
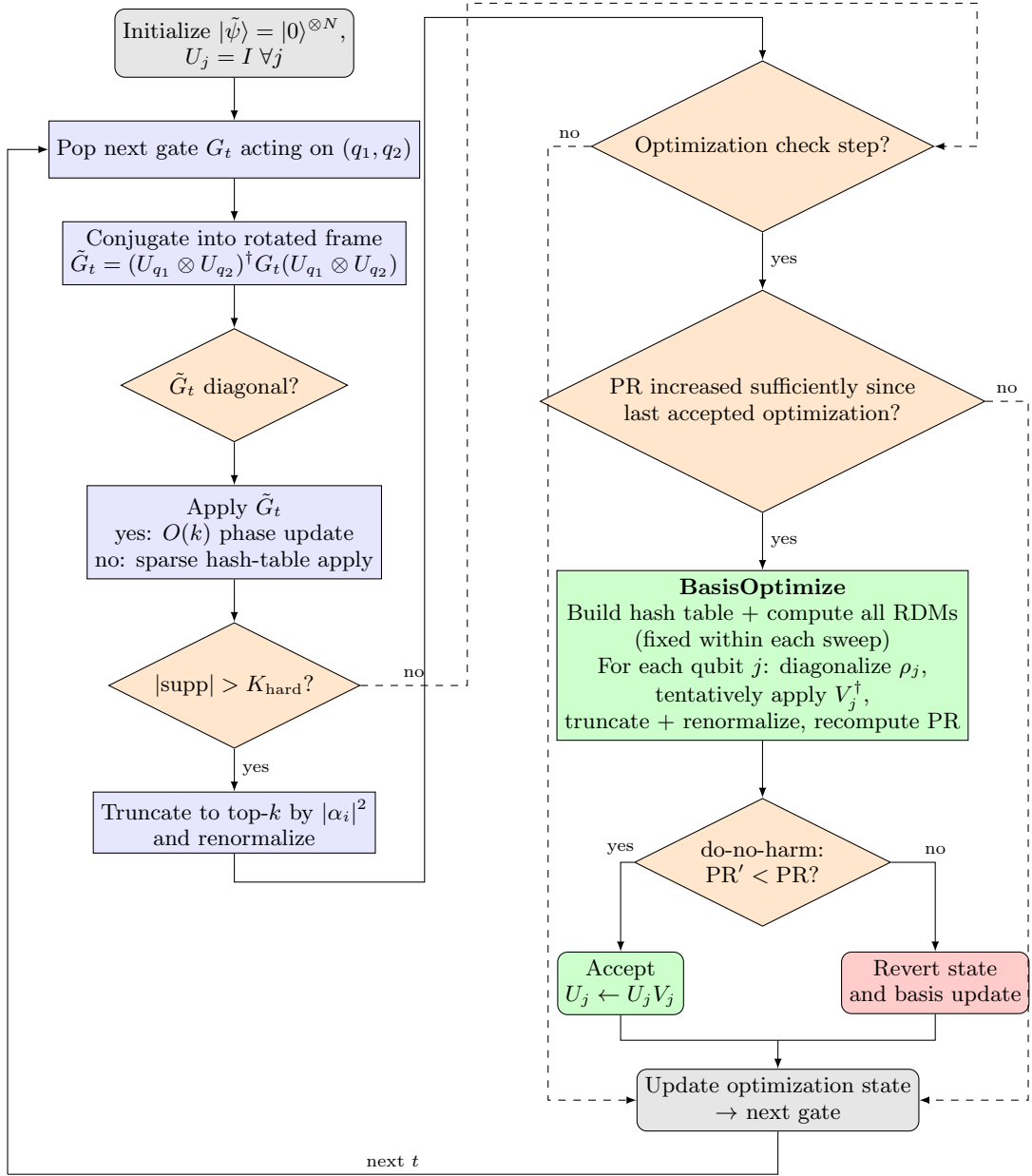

% \end{document}

\subsection{Computational optimizations}
\label{subsec:optimizations}

Several implementation-level optimizations reduce the overhead associated
with adaptive basis updates while preserving the sparse budget~$k$.

\textit{Hash-table gate application.}
For generic non-diagonal two-qubit gates, output amplitudes are accumulated
through open-addressing hash-table insertion rather than global
sort-and-merge operations. Since each retained basis state contributes at
most four output configurations, the gate-application cost scales
approximately linearly with the active support size. A single global hash table is constructed for the active support and reused across all single-qubit reduced-density-matrix evaluations. This reduces the time complexity of computing all $N$ single-qubit RDMs to an expected $\mathcal{O}(Nk)$ operations, avoiding the $\mathcal{O}(N k \log k)$ cost of naïve per-qubit sorting.

\textit{Diagonal-gate fast path.}
Two-qubit gates that remain diagonal in the current working basis are applied
as direct phase updates,
\[
\alpha_i \leftarrow
\alpha_i \,\tilde{G}_{b_1 b_2,b_1 b_2},
\]
without support expansion or temporary memory allocation. Here
\[
\tilde{G}
=
(U_{q_1}\otimes U_{q_2})^\dagger
\,G\,
(U_{q_1}\otimes U_{q_2})
\]
denotes the gate represented in the accumulated local basis frame. Once the
local basis rotations become sufficiently nontrivial, $\tilde{G}$ generally
ceases to remain diagonal, at which point the implementation automatically
reverts to the generic sparse hash-table update path. Consequently, the
largest speedup from the diagonal fast path occurs in weakly rotated regimes
where the adaptive basis remains close to the computational basis.

\textit{Deferred truncation and temporary expansion.}
The simulator permits transient support growth beyond the target sparse
budget~$k$ through a working-memory cap
\[
K_{\mathrm{hard}} = \min(8k,2^N).
\]
Intermediate gate applications may therefore temporarily retain more than
$k$ amplitudes before truncation is enforced. This deferred-truncation
strategy improves amplitude selection by allowing multiple nearby updates to
accumulate before support pruning, while leaving the reported sparse budget
unchanged. Fixed-basis and adaptive-basis simulations use the same temporary
working-memory cap at matched~$k$.

\textit{Adaptive optimization trigger.}
Basis optimization is not invoked after every gate application. Instead,
optimization calls are attempted periodically on truncation-active steps and
are skipped unless the PR has increased sufficiently since the previous accepted optimization sweep. In the implementation, this criterion
is controlled by a relative threshold parameter
$r_{\mathrm{opt}}$, with the default value corresponding to an approximate
$10\%$ increase in PR before re-optimization is triggered. This adaptive
trigger substantially reduces optimization overhead on slowly evolving
circuits where the locally optimal basis changes little between successive
gate applications.

\subsection{Complexity analysis}
We identify two primary contributions to the runtime for an \(M\)-gate circuit: (a) the sequential application of \(M\) gates to a state truncated to a maximum of \(k\) basis components, and (b) the periodic basis-optimization sweeps, performed every \(n_{\text{opt}}\) gates, which evaluate and update the local basis for all \(N\) qubits.

\paragraph{Gate Application Cost:}
Applying a one- or two-qubit gate to a sparse state of size \(k\) expands the temporary support of the wavefunction. BASS resolves basis collisions using a linear-probing hash table, reducing the expected combination cost to \(\mathcal{O}(k)\). However, reducing the expanded state back to the budget \(k\) requires a top-$k$ selection or sorting algorithm. Using a standard sort, the upper bound for a single gate application is \(\mathcal{O}(k \log k)\). Assuming an optimized truncation kernel (e.g., radix sort or quickselect), the average gate cost is near \(\mathcal{O}(k)\) with a large constant factor. Over \(M\) gates, the runtime contribution is therefore
\(
T_{\text{gates}} = \mathcal{O}(M k \log k)
\)
in the worst case, or practically $\mathcal{O}(M k)$ when using $\mathcal{O}(k)$ selection.

\paragraph{Basis-Optimization Cost:}
Basis optimization runs periodically to cap overhead, about $M/n_{\text{opt}}$ times per circuit. Each sweep computes all $N$ single-qubit RDMs at $\mathcal{O}(k)$ per qubit via hash-table pair lookups. The $\mathcal{O}(N)$ local rotations are applied one qubit at a time, and each step can expand the support. To enforce the strict memory budget \(k\), the state must undergo intermediate truncation \textit{after each qubit rotation}. This sequential expansion and sorting incurs a cost of $\mathcal{O}(k \log k)$ per qubit. Consequently, the total cost for a full $N$-qubit optimization sweep is bounded by $\mathcal{O}(N k \log k)$.

The total asymptotic runtime is the sum of these contributions:
\begin{align}
\label{eq:TBASScomplex}
T_{\text{BASS}}  
= \mathcal{O}\!\left(Mk \log k + \frac{M}{n_{\text{opt}}} Nk \log k\right).
\end{align}
While specialized $\mathcal{O}(k)$ selection algorithms can drop the $\log k$ term in theory, the massive memory-bandwidth requirements of reallocating and moving $N$ temporary arrays of size $\mathcal{O}(k)$ during the optimization sweep dominate the practical wall-clock time, making the $\mathcal{O}(N k \log k)$ bound more representative of the algorithm's actual hardware performance.

\paragraph{Memory Scaling:}
While both Fixed-basis and BASS methods exhibit strict asymptotic $\mathcal{O}(k)$ memory scaling, the practical constant factors for hardware allocation are significant. To prevent catastrophic linear-probing degradation, the sparse hash tables require power-of-two capacities with maximum load factors strictly below $0.5$. Furthermore, the double-buffering required during the truncation sweeps necessitates intermediate allocations of up to $4k$ elements. Thus, the hardware RAM footprint scales as $\mathcal{O}(C \cdot k)$, where $C \sim 10-20$ relative to a simple dense vector array.

We emphasize that this constant-factor overhead consists strictly of empty hash-table buckets and unpopulated transient arrays required to maintain $\mathcal{O}(1)$ average algorithmic throughput. It does \textit{not} provide additional physical degrees of freedom. In all benchmarks, both the fixed-basis baseline and BASS utilize the exact same transient working-memory cap ($K_{\text{hard}}$) and strictly truncate the state back to exactly $k$ non-zero amplitudes before proceeding. The reported fidelity advantages stem entirely from the rotational basis choice, not from an expanded state vector.

Table~\ref{tab:complexity} summarizes the asymptotic bounds of BASS compared to standard exact and tensor-network baselines.

\begin{table}[!hbt]
\caption{\label{tab:complexity}\textbf{Worst-case asymptotic time and memory costs for simulating an $M$-gate, $N$-qubit circuit.} $k$: reported sparse budget; $\chi$: MPS bond dimension; $n_\text{opt}$: optimization interval for BASS. Adaptive basis simulations incur an $\mathcal{O}(Nk \log k)$ penalty per sweep due to the intermediate sorting/truncation required after rotating each of the $N$ qubits.}
\begin{ruledtabular}
\begin{tabular}{lll}
Method   & Time                                                  & Memory      \\
\hline
Exact    & $\mathcal{O}(M \cdot 2^N)$                                  & $\mathcal{O}(2^N)$    \\
MPS      & $\mathcal{O}(MN\chi^3)$                                     & $\mathcal{O}(N\chi^2)$\\
Fixed basis & $\mathcal{O}(Mk\log k)$                                  & $\mathcal{O}(k)$\footnote{Constants reflect minimal array allocations.} \\
BASS   & $\mathcal{O}\!\left(Mk\log k + \frac{M}{n_\text{opt}}Nk\log k\right)$  & $\mathcal{O}(k)$\footnote{Constants reflect hash-table overhead and double-buffering allocations ($C \sim 10-20 \times k$).} \\
\end{tabular}
\end{ruledtabular}
\end{table}
\section{\label{sec:results}Numerical Results}
We now benchmark the adaptive-basis BASS against the fixed-basis sparse simulator across four circuit families, multiple system sizes, and a range of sparse budgets. All fidelity values are computed by direct overlap with the exact state vector, not estimated from $\gamma^2$.

\subsection{Experimental setup}
\label{subsec:setup}

Unless otherwise stated, benchmarks use $k=2048$ retained basis states, matched
transient gate-application buffers for fixed-basis and adaptive-basis runs, and
the do-no-harm check active.  Code-path equivalence is verified: with basis
optimization disabled, BASS reproduces the fixed-basis fidelity to machine
precision ($\max|\Delta F| < 5\times10^{-20}$; 10 trials, QAOA $p=3$, $N=16$).
Fidelities are aggregated as geometric means over independent circuit instances;
all reported confidence intervals are 95\% bootstrap intervals with $4\,000$
resamples.

We benchmark on five circuit families spanning the entanglement spectrum from
near-product to maximally scrambled.

\textit{(a) Brickwork ($L=5$):}
alternating even/odd layers of Haar-random two-qubit
gates on a 1D chain~\cite{nahum2017,vonkeyserlingk2018,chan2018}.
Five layers produce near-maximal $Z$-basis delocalization and represent the most
demanding regime for fixed-basis sparse simulation.

\textit{(b) Haar-random ($L=3$):}
in each of $L=3$ layers, qubits are randomly paired and an independent
Haar-random two-qubit gate is applied to each pair~\cite{nahum2017}.
The absence of geometric structure leads to rapid mean-field-like scrambling
without the spatial locality of brickwork.

\textit{(c) QAOA ($p=3$):}
MaxCut circuit on a single random 3-regular graph generated once per circuit
instance via the pairing model~\cite{bollobas1980}.  The $p=3$-round circuit
alternates cost unitaries
$$U_C(\gamma_k) = \prod_{(i,j)\in E} e^{+i\gamma_k Z_iZ_j/2}$$
and mixer unitaries
$$U_B(\beta_k) = \prod_i e^{-i\beta_k X_i},$$ with angles
\begin{align}
  \gamma_k &= \frac{\pi}{4}\sin\!\bigl(\frac{(2k-1)\pi}{4p}\bigr),\\
  \beta_k  &= \frac{\pi}{4}\cos\!\bigl(\frac{(2k-1)\pi}{4p}\bigr),
  \label{eq:qaoa_angles}
\end{align}
the one-term Fourier parameterisation near-optimal for 3-regular
MaxCut~\cite{farhi2014,zhou2020,crooks2018}, with a small per-instance
perturbation $\varepsilon\sim\mathcal{U}(-0.05,+0.05)$ to sample the landscape
around the optimum.  The same graph is used across all $p$ rounds.

\textit{(d) Quantum RFIM ($L=5$, $W=2$):}
first-order Trotterized time evolution under the one-dimensional quantum
Random-Field Ising Hamiltonian
\begin{equation}
  H = -J\!\sum_i Z_iZ_{i+1} - h_0\!\sum_i X_i - \sum_i \Delta_i Z_i,
  \label{eq:rfim_H}
\end{equation}
with $J=h_0=1$, uniform transverse field, and longitudinal disorder
$\Delta_i\sim\mathcal{U}[-W,W]$, $W=2$, $dt=0.2$.
Disorder draws are independent across sites and circuit instances.
At $W=2$, the system sits near the ergodic/many-body-localized crossover of the
1D chain~\cite{palhuse2010,oganesyanhuse2007,luitz2015}, producing intermediate
$\mathrm{PR}_Z$ values for which basis adaptivity provides a moderate advantage.
Each Trotter step comprises three sublayers:
$U_{ZZ}=\prod_i e^{+iJdt\,Z_iZ_{i+1}}$,
$U_X=\prod_i e^{+ih_0 dt\,X_i}$, and
$U_\Delta=\prod_i e^{+i\Delta_i dt\,Z_i}$,
implemented respectively by \texttt{RZZGate}, \texttt{RXGate}, and
\texttt{RZGate}.

\textit{(e) UCCSD ($L=1$) [i.e UCC-1]:}
first-order Trotterized Jordan-Wigner unitary coupled-cluster singles and
doubles~\cite{whitfield2011,anand2022}, starting from the half-filling
Hartree-Fock reference $|1^{N_e}0^{N-N_e}\rangle$.
Single-excitation amplitudes $\theta_{ia}\sim\mathcal{U}[-0.3,0.3]$ and
double-excitation amplitudes $\theta_{ijab}\sim\mathcal{U}[-0.05,0.05]$
are drawn independently per instance at scales typical of weakly correlated
molecules~\cite{anand2022}.
Each excitation operator is decomposed into Pauli strings via the
Jordan-Wigner mapping and exponentiated by a CNOT-cascade circuit~\cite{whitfield2011};
intermediate qubits contribute $Z$ parity strings automatically through the
cascade.
The resulting states are substantially more concentrated in the $Z$ basis than
brickwork or Haar circuits, probing the regime where basis adaptivity provides
little benefit.

For each configuration, multiple independent random instances are generated
(counts reported in the corresponding figure or table caption) to ensure reliable
geometric-mean estimates.  System size $N$ and sparse budget $k$ are varied as
specified in each experiment.

\subsection{Fidelity and PR vs.\ sparse budget}
\label{subsec:fid_vs_k}

We compare BASS to fixed-basis sparse simulation on three stochastic families: 1D brickwork, $4\times5$ 2D brickwork, and the random-field Ising model (RFIM). Each point is the geometric mean over 100 independent $N=20$ trials, with 95\% bootstrap confidence intervals, as we sweep the retained state count~$k$ and measure fidelity $F$ against an exact reference.

Figure~\ref{fig:fid_pr} and Tables~\ref{tab:fidelity_vs_k}–\ref{tab:pr_vs_k} show a consistent gap, largest when the state is spread across many $Z$ strings. The $\PR$ curves track the same story: at matched $F$, BASS needs far fewer effective basis states than the fixed basis.

\textit{Random-field Ising model.} RFIM states are partly localized, and the relative gain is largest at small $k$. At $k=100$, BASS reaches $F \approx 0.530$, $3.78\times$ the fixed-basis value. The margin shrinks as $k$ grows but does not cross over. For $k \le 1000$, BASS runs at $\PR \approx 1.7$–$5.2$ while the fixed basis needs $\PR \approx 22$–$301$ for comparable fidelity.

\textit{Brickwork circuits.} These are the hard case: 1D brickwork has $\PR_Z \simeq 1.81\times 10^5$. BASS stays ahead across the sweep, with a $427.9\times$ ratio at $k=500$ and still $4.1\times$ at $k=50{,}000$. On $4\times5$ 2D brickwork the pattern repeats—order-of-magnitude gains at small $k$ and $\PR \simeq 72$ versus $\PR \simeq 951$ at $k=2000$.

\begin{figure}[!hbt]
\centering
\makebox[\columnwidth][l]{\textbf{(a)}}\\[-0.5ex]
\includegraphics[width=0.95\columnwidth]{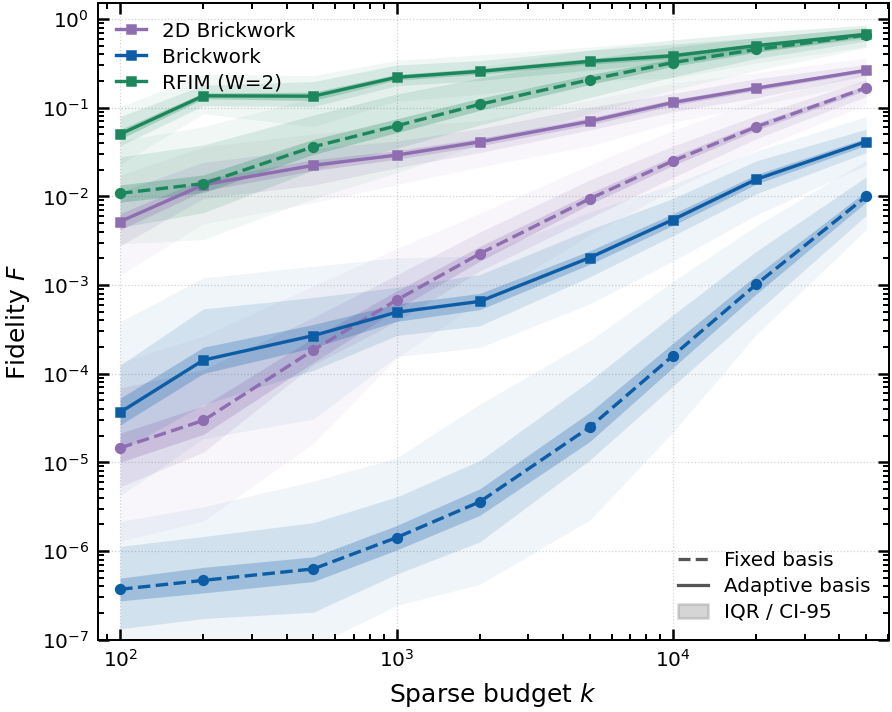}
\makebox[\columnwidth][l]{\textbf{(b)}}\\[-0.5ex]
\includegraphics[width=0.95\columnwidth]{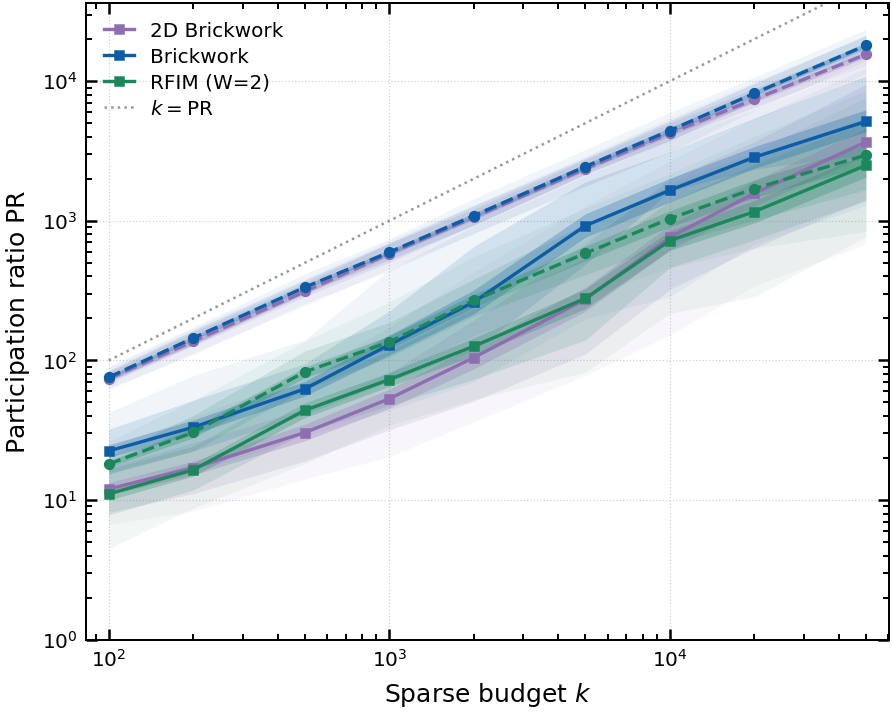}
\caption{\textbf{Fidelity and PR vs. sparse budget.} \textbf{(a)} Fidelity $F$ as a function of the sparse budget $k$. \textbf{(b)} PR $\PR$ as a function of $k$. Curves represent the geometric mean of per-trial results. Shaded regions illustrate the distribution and statistical uncertainty of the data, with bands from lightest to darkest representing the 10th–90th percentile range, the interquartile range (IQR, 25th–75th percentile), and the 95\% bootstrap confidence interval for the geometric mean, respectively. Solid markers denote results obtained with the adaptive basis method, while open markers denote the fixed-basis method. Circuit families considered: 1D random brickwork circuits (depth 6); random-field Ising model (RFIM) circuits ($\sigma_h/h = 0.1$, 5 Trotter layers); and $4\times5$ 2D brickwork circuits (depth 4).}
\label{fig:fid_pr}
\end{figure}

\begin{table*}[!hbt]
\centering
\caption{\textbf{Fidelity ($F$) vs. Sparse Budget ($k$).} Statistics for $N=20$ circuits. Reported as Median [IQR]. Ratio is the geometric mean of $F_{\text{BASS}}/F_{\text{fixed}}$.}
\label{tab:fidelity_vs_k}
\begin{tabular}{l  l  r  r  l}
\toprule
Family & $k$ & \multicolumn{1}{c}{$F_{\text{fixed}}$} & \multicolumn{1}{c}{$F_{\text{BASS}}$} & \multicolumn{1}{c}{Ratio (GM [95\% CI])} \\
\midrule
2D Brickwork & 100 & $1.67 \, [0.5, 6.9] \times 10^{-5}$ & $6.50 \, [2.8, 11.2] \times 10^{-3}$ & $356.9\times \, [236, 535]$ \\
             & 200 & $3.41 \, [1.3, 11.1] \times 10^{-5}$ & $1.38 \, [0.9, 2.4] \times 10^{-2}$ & $454.6\times \, [306, 661]$ \\
             & 500 & $2.34 \, [1.3, 4.3] \times 10^{-4}$ & $2.43 \, [1.3, 3.5] \times 10^{-2}$ & $121.3\times \, [88, 167]$ \\
             & 2k  & $2.60 \, [1.4, 4.0] \times 10^{-3}$ & $4.42 \, [3.1, 5.8] \times 10^{-2}$ & $18.2\times \, [15, 22]$ \\
             & 5k  & $9.92 \, [5.7, 15.3] \times 10^{-3}$ & $6.77 \, [5.6, 10.1] \times 10^{-2}$ & $7.5\times \, [6.5, 8.5]$ \\
             & 10k & $2.60 \, [1.6, 3.7] \times 10^{-2}$ & $1.15 \, [0.9, 1.45] \times 10^{-1}$ & $4.6\times \, [4.1, 5.1]$ \\
             & 20k & $6.22 \, [4.5, 8.1] \times 10^{-2}$ & $1.64 \, [1.3, 2.12] \times 10^{-1}$ & $2.7\times \, [2.5, 3.0]$ \\
             & 50k & $1.66 \, [1.3, 2.1] \times 10^{-1}$ & $2.67 \, [2.3, 3.1] \times 10^{-1}$ & $1.6\times \, [1.5, 1.7]$ \\
\midrule
1D Brickwork & 100 & $4.25 \, [0.1, 1.1] \times 10^{-7}$ & $3.77 \, [1.4, 12.5] \times 10^{-5}$ & $99.1\times \, [62, 156]$ \\
             & 200 & $7.00 \, [1.7, 14.6] \times 10^{-7}$ & $1.55 \, [0.4, 5.4] \times 10^{-4}$ & $304.5\times \, [198, 454]$ \\
             & 500 & $6.36 \, [0.2, 2.1] \times 10^{-7}$ & $3.54 \, [1.1, 7.2] \times 10^{-4}$ & $427.9\times \, [290, 631]$ \\
             & 2k  & $3.57 \, [1.3, 10.5] \times 10^{-6}$ & $7.33 \, [3.5, 13.1] \times 10^{-4}$ & $182.1\times \, [124, 270]$ \\
             & 5k  & $3.13 \, [1.1, 8.3] \times 10^{-5}$ & $1.96 \, [1.2, 4.2] \times 10^{-3}$ & $81.5\times \, [56, 119]$ \\
             & 10k & $1.92 \, [0.7, 4.6] \times 10^{-4}$ & $5.60 \, [3.6, 9.5] \times 10^{-3}$ & $34.1\times \, [25, 46]$ \\
             & 20k & $1.12 \, [0.5, 2.4] \times 10^{-3}$ & $1.61 \, [1.1, 2.5] \times 10^{-2}$ & $15.3\times \, [12, 19]$ \\
             & 50k & $1.02 \, [0.6, 1.7] \times 10^{-2}$ & $3.95 \, [3.1, 5.7] \times 10^{-2}$ & $4.1\times \, [3.6, 4.8]$ \\
\midrule
RFIM ($W=2$) & 100 & $1.18 \, [0.4, 2.8] \times 10^{-2}$ & $5.23 \, [3.7, 7.9] \times 10^{-2}$ & $4.7\times \, [3.5, 6.0]$ \\
             & 200 & $1.37 \, [0.7, 3.8] \times 10^{-2}$ & $1.31 \, [1.0, 1.8] \times 10^{-1}$ & $9.8\times \, [8.0, 11.9]$ \\
             & 500 & $3.77 \, [2.0, 8.2] \times 10^{-2}$ & $1.55 \, [1.0, 2.0] \times 10^{-1}$ & $3.7\times \, [2.9, 4.7]$ \\
             & 2k  & $1.19 \, [0.6, 2.2] \times 10^{-1}$ & $2.88 \, [2.0, 3.4] \times 10^{-1}$ & $2.4\times \, [1.9, 2.8]$ \\
             & 5k  & $2.25 \, [1.3, 3.6] \times 10^{-1}$ & $3.79 \, [2.7, 4.5] \times 10^{-1}$ & $1.6\times \, [1.4, 1.9]$ \\
             & 10k & $3.50 \, [2.2, 4.8] \times 10^{-1}$ & $4.14 \, [3.2, 5.2] \times 10^{-1}$ & $1.2\times \, [1.1, 1.3]$ \\
             & 20k & $4.94 \, [3.6, 6.1] \times 10^{-1}$ & $5.16 \, [4.3, 6.1] \times 10^{-1}$ & $1.1\times \, [1.0, 1.2]$ \\
             & 50k & $7.00 \, [5.7, 8.0] \times 10^{-1}$ & $7.00 \, [6.0, 8.0] \times 10^{-1}$ & $1.0\times \, [1.0, 1.1]$ \\
\bottomrule
\end{tabular}
\end{table*}

\begin{table*}[!hbt]
\centering
\caption{\textbf{PR ($\PR$) vs. Sparse Budget ($k$).} Statistics for $N=20$ circuits. Reported as Median [IQR]. Compression Factor refers to the geometric mean of the structural localization ratio ($\PR_{\text{fixed}} / \PR_{\text{BASS}}$).}
\label{tab:pr_vs_k}
\begin{tabular}{llrr r}
\toprule
Family & $k$ & $\PR_{\text{fixed}}$ (Med [IQR]) & $\PR_{\text{BASS}}$ (Med [IQR]) & \multicolumn{1}{c}{Compression Factor} \\
\midrule
2D Brickwork & 100 & 76.6 [69.0, 82.7] & 11.5 [8.1, 17.3] & $6.2\times$ \\
             & 200 & 142.0 [127.0, 159.0] & 16.7 [11.1, 25.7] & $8.1\times$ \\
             & 500 & 317.0 [287.0, 359.0] & 29.4 [19.0, 45.1] & $10.2\times$ \\
             & 1k  & 606.0 [504.0, 690.0] & 51.5 [31.8, 81.2] & $10.9\times$ \\
             & 2k  & 1060.0 [936.0, 1210.0] & 91.6 [51.0, 195.0] & $10.1\times$ \\
             & 5k  & 2420.0 [2130.0, 2740.0] & 259.0 [111.0, 680.0] & $8.6\times$ \\
             & 10k & 4350.0 [3790.0, 5000.0] & 953.0 [322.0, 2000.0] & $5.6\times$ \\
             & 20k & 7680.0 [6500.0, 8880.0] & 2280.0 [637.0, 3670.0] & $4.7\times$ \\
             & 50k & 16100.0 [13000.0, 18500.0] & 5770.0 [1420.0, 9470.0] & $4.2\times$ \\
\midrule
1D Brickwork & 100 & 76.3 [70.7, 84.4] & 20.8 [15.6, 32.2] & $3.4\times$ \\
             & 200 & 149.0 [130.0, 165.0] & 30.5 [22.4, 51.9] & $4.4\times$ \\
             & 500 & 344.0 [303.0, 376.0] & 58.1 [41.0, 94.9] & $5.4\times$ \\
             & 1k  & 638.0 [527.0, 695.0] & 117.0 [70.7, 226.0] & $4.6\times$ \\
             & 2k  & 1100.0 [925.0, 1280.0] & 231.0 [118.0, 656.0] & $4.2\times$ \\
             & 5k  & 2480.0 [2170.0, 2820.0] & 1420.0 [488.0, 1900.0] & $2.7\times$ \\
             & 10k & 4660.0 [3740.0, 5210.0] & 2480.0 [1630.0, 3120.0] & $2.7\times$ \\
             & 20k & 8290.0 [7160.0, 9620.0] & 4000.0 [2340.0, 5430.0] & $2.9\times$ \\
             & 50k & 18500.0 [16000.0, 21300.0] & 6870.0 [2980.0, 10900.0] & $3.5\times$ \\
\midrule
RFIM ($W=2$) & 100 & 17.7 [15.3, 20.8] & 12.8 [7.8, 16.8] & $1.6\times$ \\
             & 200 & 31.2 [22.5, 40.4] & 16.6 [11.8, 24.3] & $1.9\times$ \\
             & 500 & 83.1 [62.0, 117.0] & 52.1 [30.6, 68.0] & $1.9\times$ \\
             & 1k  & 133.0 [105.0, 187.0] & 76.7 [47.5, 117.0] & $1.9\times$ \\
             & 2k  & 260.0 [208.0, 382.0] & 143.0 [73.2, 234.0] & $2.1\times$ \\
             & 5k  & 601.0 [412.0, 953.0] & 333.0 [140.0, 534.0] & $2.1\times$ \\
             & 10k & 1140.0 [679.0, 1760.0] & 881.0 [462.0, 1370.0] & $1.4\times$ \\
             & 20k & 1960.0 [1010.0, 2930.0] & 1510.0 [731.0, 2530.0] & $1.5\times$ \\
             & 50k & 3250.0 [1670.0, 5480.0] & 3000.0 [1400.0, 5090.0] & $1.2\times$ \\
\bottomrule
\end{tabular}
\end{table*}

\textit{Table Notes:} Notation: $k$ denotes the sparse budget; $N=20$ is the circuit size. Reported values for Fidelity ($F$) and PR ($\PR$) are Medians with Interquartile Ranges [IQR] over 100 trials. Ratio metrics utilize the Geometric Mean (GM) with 95\% Bootstrap Confidence Intervals [CI]. Compression Factor is defined as $\PR_{\text{fixed}} / \PR_{\text{BASS}}$. Values are reported as $\alpha \, [\beta, \gamma] \times 10^{-n}$, where $\alpha$ is the median, and $[\beta, \gamma]$ represents the interquartile range (IQR) [25th–75th percentile]. The scaling factor $10^{-n}$ applies to the median and the interval bounds.

\subsection{System-size scaling}
\label{subsec:scaling}

We now distinguish between two physically distinct scaling regimes:
(i) exponential-budget scaling, where the retained subspace size grows with
system size, and (ii) fixed-budget scaling, where the truncation budget is
held constant as the Hilbert space dimension grows exponentially. These two experiments probe complementary aspects of adaptive-basis
truncation. Exponential-budget scaling tests whether adaptive basis
selection remains useful once the truncation pressure weakens with increasing system size, while fixed-budget scaling probes the true asymptotic compression regime where the retained fraction of Hilbert space decreases exponentially with~$N$.

\begin{figure}[!hbt]
\centering
\makebox[\linewidth][l]{\textbf{(a)}}\\
\includegraphics[width=\linewidth]{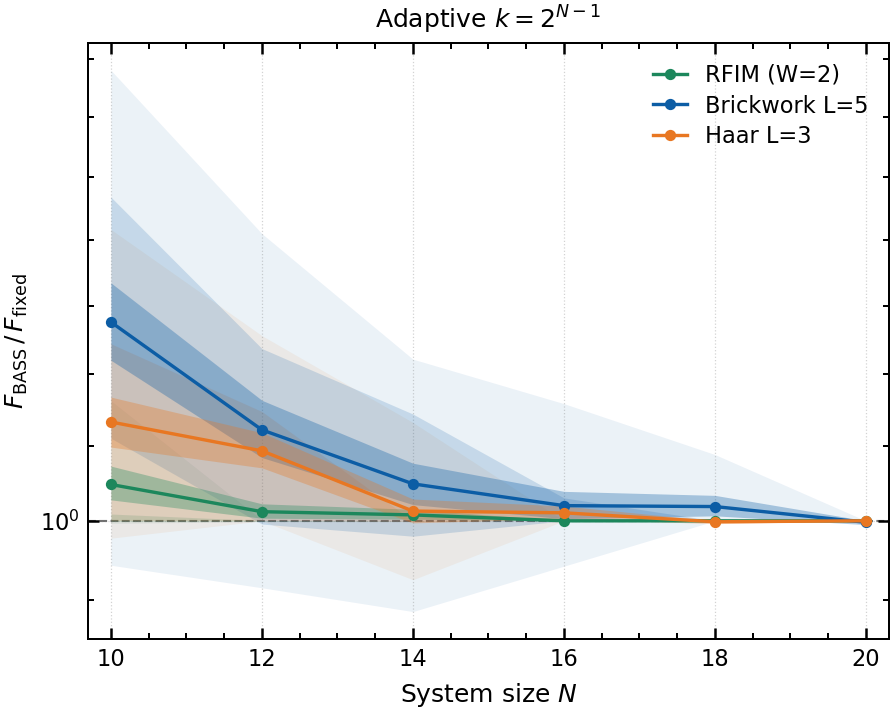}
\makebox[\linewidth][l]{\textbf{(b)}}\\
\includegraphics[width=\linewidth]{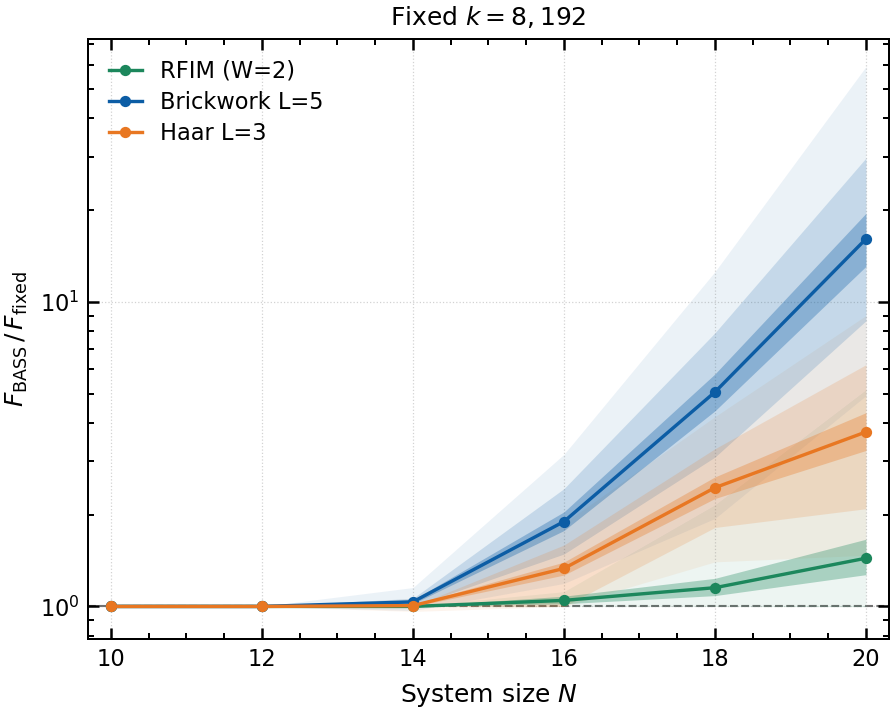}
\caption{\textbf{System-size scaling of adaptive-basis advantage.}
\textbf{(a)} Exponential-budget scaling: adaptive/fixed fidelity ratio as
a function of system size under the scaling trajectory
$k = 2^{N-1}$. Dashed horizontal line indicates ratio $=1$
(no adaptive advantage). Brickwork and Haar-random circuits exhibit
modest enhancement at intermediate system sizes before converging back
toward unity as truncation pressure weakens, while RFIM circuits remain
effectively identical to the fixed-basis method across all tested sizes.
\textbf{(b)} Fixed-budget scaling: geometric mean adaptive/fixed fidelity
ratio versus system size under fixed truncation budget $k=8192$. The vertical axis
is logarithmic. Under a fixed compression budget, the adaptive advantage
grows rapidly for chaotic circuit families while remaining near unity for
RFIM circuits. In both panels, shaded regions denote interquartile uncertainty
bands over 100 random-circuit trials.}
\label{fig:scaling}
\end{figure}

\paragraph*{Exponential-budget scaling.}
Figure~\ref{fig:scaling}(a) shows the adaptive/fixed fidelity ratio under
the scaling trajectory $k = 2^{N-1}$ for shallow circuit ensembles. In this regime, the retained subspace grows
exponentially with system size, so truncation becomes progressively less
severe at larger~$N$.

The resulting behavior is quantitatively weak across all tested families.
For RFIM circuits, the geometric-mean fidelity ratio remains essentially
unity across all system sizes, varying only from
$1.02\times$ at $N=10$
to $1.00\times$ by $N=20$.
The median fidelities of the fixed and adaptive methods are statistically
indistinguishable throughout the entire scaling trajectory, confirming that
the dominant wavefunction support already remains strongly aligned with the
computational basis.

Brickwork circuits exhibit the strongest transient enhancement at small
system size, reaching a geometric-mean improvement of
$1.14\times~[1.11,1.17]$
at $N=10$, decreasing to
$1.06\times~[1.04,1.08]$
at $N=12$, and converging toward unity for larger systems:
$1.02\times$ at $N=14$,
$1.01\times$ at $N=16$ and $N=18$,
and effectively $1.00\times$ by $N=20$.

Haar-random circuits display similar behavior with weaker enhancement.
The adaptive method achieves
$1.07\times~[1.05,1.08]$
at $N=10$ and
$1.05\times~[1.04,1.06]$
at $N=12$, after which the ratio rapidly saturates toward unity:
$1.01\times$ at $N=14$ and $N=16$,
and exactly $1.00\times$ within statistical uncertainty for
$N \ge 18$.

This convergence toward
$F_{\mathrm{BASS}}/F_{\mathrm{Fixed}} \rightarrow 1$
is expected. As $k$ grows exponentially with system size,
the retained subspace eventually captures nearly all significant amplitudes
even without basis adaptation. Consequently, the computational basis itself
becomes sufficient to represent the dominant wavefunction support,
eliminating the need for adaptive rotations.

\paragraph*{Fixed-budget scaling.}
The more informative scaling regime is shown in
Fig.~\ref{fig:scaling}(b), where the truncation budget is fixed at
$k=8192$ while the system size increases. Here, the retained fraction of
Hilbert space decreases exponentially with~$N$, producing increasingly
severe compression pressure.

Under fixed-budget scaling, the adaptive advantage becomes dramatic for
chaotic circuit families.

\textit{Brickwork circuits.}
Brickwork circuits display the strongest scaling enhancement.
The geometric-mean fidelity ratio remains near unity for small systems,
with
$1.04\times~[1.02,1.05]$
at $N=14$,
but then increases rapidly to
$1.90\times~[1.77,2.04]$
at $N=16$,
$5.06\times~[4.39,5.80]$
at $N=18$,
and
$16.09\times~[13.08,19.53]$
at $N=20$.

At the same time, the absolute fixed-basis fidelities collapse rapidly.
The median fixed-basis fidelity decreases from
$6.36\times10^{-1}$ at $N=14$
to
$9.86\times10^{-2}$ at $N=16$,
$8.68\times10^{-3}$ at $N=18$,
and
$8.10\times10^{-4}$ at $N=20$.
In contrast, the adaptive basis preserves substantially larger wavefunction
weight, maintaining median fidelities of
$1.85\times10^{-1}$,
$4.63\times10^{-2}$,
and
$1.18\times10^{-2}$
at the same respective system sizes.

This behavior reflects the exponentially increasing mismatch between the
computational basis and the dominant wavefunction support under strong
scrambling dynamics. As the effective support spreads broadly across
Hilbert space, fixed-basis truncation becomes increasingly inefficient,
whereas adaptive basis rotations continue to concentrate probability
weight into the retained subspace.

\textit{Haar-random circuits.}
Haar-random circuits exhibit the same qualitative trend, although with
weaker scaling than brickwork circuits.
The geometric-mean fidelity ratio increases from
$1.01\times$ at $N=14$
to
$1.33\times~[1.27,1.40]$
at $N=16$,
$2.45\times~[2.26,2.66]$
at $N=18$,
and
$3.74\times~[3.25,4.33]$
at $N=20$.

The weaker scaling relative to brickwork circuits likely reflects the less
structured entanglement geometry of Haar-random gate placement, which can
occasionally leave substantial support partially aligned with the
computational basis even at large system sizes.

\textit{Random-field Ising model.}
RFIM circuits remain relatively stable across all tested system sizes.
The geometric-mean fidelity ratio remains exactly unity up to $N=14$,
before increasing only modestly under stronger compression pressure:
$1.05\times~[1.02,1.08]$
at $N=16$,
$1.15\times~[1.09,1.24]$
at $N=18$,
and
$1.44\times~[1.27,1.66]$
at $N=20$.
Even at the largest system size, the progress remains substantially
smaller than that observed for chaotic circuit families.

These results confirm that adaptive rotations are primarily beneficial
when the dominant wavefunction support becomes strongly misaligned with
the computational basis. In the RFIM regime, the computational basis
already provides a relatively efficient representation, so BASS naturally
reduces to behavior that is nearly identical to fixed-basis truncation.

Taken together, these results establish the key scaling conclusion of
this work:
\textit{Adaptive basis selection becomes increasingly important as
compression pressure grows and the effective wavefunction support spreads
away from the computational basis.}

The advantage is hence strongly regime dependent. When the retained
subspace is large enough to capture nearly all relevant amplitudes,
adaptive rotations become unnecessary and the method smoothly converges to
the fixed-basis limit. Conversely, under strong compression pressure in
highly scrambling circuit families, adaptive basis selection yields
multiplicative fidelity improvements that grow rapidly with system size.

\subsection{The \texorpdfstring{$\PR_Z$}{PR\_Z} crossover criterion}
\label{subsec:crossover}

Whether BASS beats fixed-basis truncation on a given state is largely determined by
$k/\PR_Z$.
Figure~\ref{fig:crossover} shows the collapse of fidelity ratios onto this axis.

\begin{figure}[!hbt]
\centering
\includegraphics[width=\linewidth]{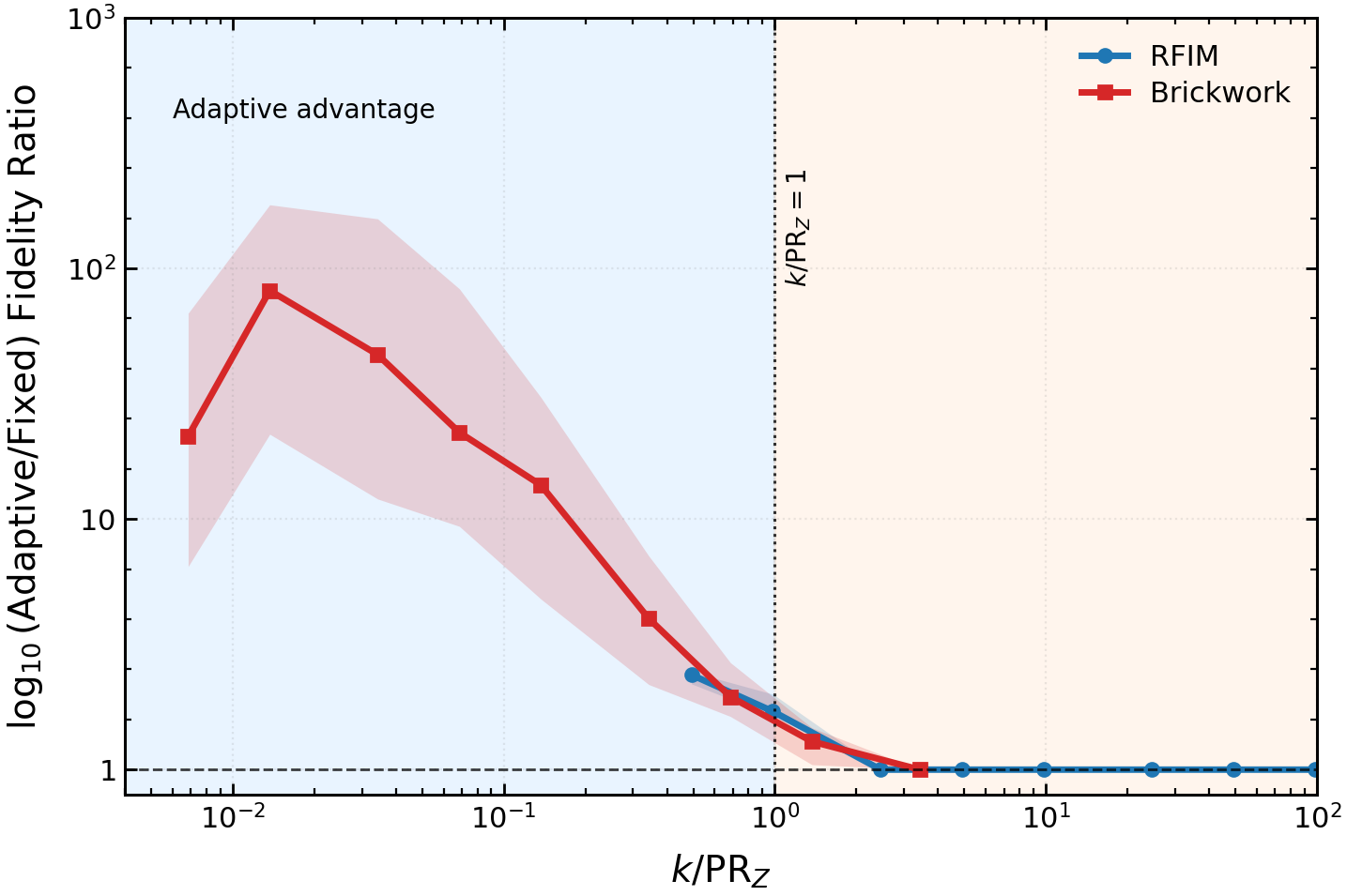}
\caption{\textbf{Universal crossover controlled by the normalized truncation scale $k/\PR_Z$.} Adaptive-basis/fixed-basis fidelity ratio as a function of the normalised truncation budget $k/\PR_Z$ for RFIM and brickwork circuit ensembles at $N=16$. The vertical dashed line indicates $k/\PR_Z = 1$, the truncation size is comparable to the state's effective participation support. For $k/\PR_Z < 1$, adaptive truncation performs better than fixed-basis truncation, with the highest improvements obtained for chaotic brickwork circuits ($\sim10$-$10^2\times$). RFIM states exhibit less improvement ($\sim1$-$3\times$) due to their more organized amplitude distributions. As $k$ approaches $\PR_Z$, both methods converge and the fidelity ratio reaches unity. Increasing system size $N$, circuit depth, or moving to more chaotic and critical regimes often increases the PR $\PR_Z$, driving the dynamics farther into the undersampled domain ($k/\PR_Z \ll 1$). In contrast, localized or weakly entangling regimes provide less $\PR_Z$, allowing for convergence with smaller truncation budget.}
\label{fig:crossover}
\end{figure}

When $k < \PR_Z$, the $k$ states retained by the fixed-basis simulator in the
$Z$ basis covers only a small fraction of the total probability; BASS
concentrates that weight into far fewer rotated-basis states. When
When $k > \PR_Z$, fixed-basis truncation captures the dominant weight, and BASS doesn't provide extra compression.
$k/\PR_Z$ is a reliable predictor of whether to use basis adaptation on tested circuit families due to its sharp crossover.

Operationally, the simulator can compute the PR of its
current retained support at $\mathcal{O}(k)$ cost,
\begin{equation}
\widehat{\PR}_\mathrm{supp} =
\frac{\left(\sum_{i=1}^{k}|\alpha_i|^2\right)^2}
{\sum_{i=1}^{k}|\alpha_i|^4}.
\label{eq:prz_sparse}
\end{equation}
This quantity is not an unbiased estimator of the exact-state $\PR_Z$, nor is it by itself a rigorous upper or lower bound after repeated truncation,
because the discarded amplitudes are no longer represented. Its value is
nevertheless useful as a headroom diagnostic: when
$\widehat{\PR}_\mathrm{supp}\ll k$, the retained state is already concentrated
and basis adaptation is usually unnecessary, when
$\widehat{\PR}_\mathrm{supp}$ saturates near~$k$, the current basis has filled
The available sparse support and BASS should be engaged. In practice this
The saturation signal has identified every $Z$-delocalized configuration in our
benchmarks. If a certified trigger is required, a conservative variant can
Combine Eq.~\eqref{eq:prz_sparse} with a low-budget pilot evolution or an exact
small-$N$ calibration for the circuit family.

For brickwork circuits at depth $L=5$, empirical fitting of the exact states shows the PR grows exponentially as $\mathrm{PR}_{Z} \sim 2^{0.70N}$ (see Appendix~\ref{subsec:pr_scaling} for scaling fits), ensuring $k < \mathrm{PR}_{Z}$ for any polynomial $k$ at sufficiently large $N$.

\subsection{Runtime overhead}
\label{subsec:runtime}

Table~\ref{tab:runtime} summarizes the wall-clock above adaptive
basis selection across brickwork circuits for
$N \in \{12,14,16,18,20\}$ and sparse budgets
$k \in \{500,10^3,2\times10^3,5\times10^3,10^4,2\times10^4\}$,
using 100 independent random-circuit realizations per configuration.

Two distinct computational regimes originating from the benchmark.

\paragraph{Undersampled regime ($k < \mathrm{PR}_Z$).}
When the sparse budget is less than the effective support size of the wavefunction, adaptive basis optimization is activated throughout the simulation. In this regime, the geometric-mean overhead rises to around $10.4\times$ due to the added cost of RDM creation and basis optimization.
However, this regime also results to get highest fidelity.
For instance, at $N=20$ and $k=500$, the fixed-basis fidelity decreases to
$\sim10^{-6}$ while BASS retains
$\sim3.5\times10^{-4}$ overlap of the wavefuction.
Similarly, at $N=18$ and $k=5000$, BASS improves the geometric-mean
fidelity from
$4.35\times10^{-4}$ to
$1.32\times10^{-2}$ while maintaining an overhead below
$9\times$.

\paragraph{Oversampled regime ($k \ge \mathrm{PR}_Z$).}
When the sparse budget exceeds the wavefunction's effective PR, the adaptive trigger activates less frequently, leading to a significant drop in runtime overhead.
In this domain, both fixed-basis and adaptive simulations reach great fidelity, which reduces the adaptive advantage.
At $N=12$ and $k\ge10^4$, both techniques achieve unit fidelity while reducing overhead to around $1.2\times$.

The runtime scaling shows that adaptive basis selection has just a minor computational overhead compared to the significant fidelity improvements achieved in the severely compressed domain.
As system size and compression pressure increase, fidelity improves significantly. However, the runtime overhead remains within one order of magnitude across all evaluated setups.

\begin{table*}[!hbt]
\centering
\caption{\textbf{Runtime overhead statistics for adaptive-basis truncation.}
Brickwork-circuit benchmark results ($L=6$, 100 random-circuit trials per
configuration). Overhead denotes the runtime ratio
$t_{\mathrm{BASS}}/t_{\mathrm{fixed}}$.
Reported overhead values correspond to geometric means with 95\%
bootstrap confidence intervals.
Median overheads are reported with interquartile ranges [25th, 75th
percentiles].
$\mathrm{GM}(F_{\mathrm{fixed}})$ and
$\mathrm{GM}(F_{\mathrm{BASS}})$ denote geometric-mean fidelities over
circuit realizations.
Regimes are classified according to the crossover criterion
$k \lessgtr \mathrm{PR}_Z \sim 2^{0.7N}$.}
\label{tab:runtime}
\begin{tabular}{llrrrrrrrl}
\toprule
$N$ & $k$
& $\mathrm{GM}(t_\mathrm{BASS}/t_\mathrm{fixed})$
& 95\% CI
& Median Overhead [IQR]
& $\mathrm{GM}(F_\mathrm{fixed})$
& $\mathrm{GM}(F_\mathrm{BASS})$
& $\mathrm{GM}(F_\mathrm{BASS}/F_\mathrm{fixed})$
& Regime \\
\midrule

12 & 500
& $12.92\times$
& [11.73, 14.09]
& $13.65\times$ [11.30, 14.57]
& $3.24\times10^{-2}$
& $1.02\times10^{-1}$
& $3.2\times$
& over \\

12 & 1\,000
& $12.54\times$
& [10.93, 13.94]
& $13.07\times$ [12.79, 13.72]
& $1.43\times10^{-1}$
& $2.44\times10^{-1}$
& $1.7\times$
& over \\

12 & 2\,000
& $12.02\times$
& [10.22, 14.30]
& $11.82\times$ [9.49, 14.11]
& $4.69\times10^{-1}$
& $5.15\times10^{-1}$
& $1.1\times$
& over \\

12 & 5\,000
& $1.29\times$
& [1.23, 1.36]
& $1.27\times$ [1.21, 1.36]
& $1.00$
& $1.00$
& $1.0\times$
& over \\

12 & 10\,000
& $1.22\times$
& [1.20, 1.24]
& $1.23\times$ [1.20, 1.24]
& $1.00$
& $1.00$
& $1.0\times$
& over \\

12 & 20\,000
& $1.22\times$
& [1.18, 1.27]
& $1.21\times$ [1.19, 1.23]
& $1.00$
& $1.00$
& $1.0\times$
& over \\

\midrule

14 & 500
& $12.33\times$
& [11.43, 13.13]
& $12.77\times$ [11.95, 13.16]
& $1.20\times10^{-3}$
& $2.21\times10^{-2}$
& $18.4\times$
& under \\

14 & 1\,000
& $12.43\times$
& [11.45, 13.47]
& $12.11\times$ [10.92, 14.12]
& $6.29\times10^{-3}$
& $4.51\times10^{-2}$
& $7.2\times$
& over \\

14 & 2\,000
& $10.48\times$
& [9.36, 11.68]
& $10.63\times$ [8.74, 12.24]
& $4.26\times10^{-2}$
& $1.06\times10^{-1}$
& $2.5\times$
& over \\

14 & 5\,000
& $9.65\times$
& [7.65, 11.90]
& $10.47\times$ [7.47, 12.44]
& $2.63\times10^{-1}$
& $3.21\times10^{-1}$
& $1.2\times$
& over \\

14 & 10\,000
& $6.45\times$
& [4.08, 9.44]
& $7.03\times$ [5.08, 9.20]
& $6.96\times10^{-1}$
& $6.91\times10^{-1}$
& $1.0\times$
& over \\

14 & 20\,000
& $1.15\times$
& [1.13, 1.17]
& $1.15\times$ [1.14, 1.17]
& $1.00$
& $1.00$
& $1.0\times$
& over \\

\midrule

16 & 500
& $12.88\times$
& [12.37, 13.36]
& $12.87\times$ [12.19, 13.77]
& $9.30\times10^{-5}$
& $3.56\times10^{-3}$
& $38.3\times$
& under \\

16 & 1\,000
& $12.22\times$
& [11.14, 13.40]
& $11.72\times$ [10.76, 14.24]
& $5.01\times10^{-4}$
& $5.93\times10^{-3}$
& $11.8\times$
& under \\

16 & 2\,000
& $11.14\times$
& [10.27, 11.97]
& $11.17\times$ [10.90, 12.06]
& $2.65\times10^{-3}$
& $1.68\times10^{-2}$
& $6.3\times$
& under \\

16 & 5\,000
& $8.16\times$
& [7.34, 9.06]
& $8.48\times$ [6.90, 9.02]
& $2.10\times10^{-2}$
& $5.65\times10^{-2}$
& $2.7\times$
& over \\

16 & 10\,000
& $8.03\times$
& [6.43, 9.69]
& $8.78\times$ [6.80, 9.05]
& $8.51\times10^{-2}$
& $1.26\times10^{-1}$
& $1.5\times$
& over \\

16 & 20\,000
& $6.89\times$
& [5.47, 8.80]
& $7.42\times$ [4.99, 7.60]
& $2.77\times10^{-1}$
& $2.99\times10^{-1}$
& $1.1\times$
& over \\

\midrule

18 & 500
& $11.01\times$
& [9.91, 12.10]
& $11.52\times$ [9.93, 11.89]
& $1.49\times10^{-5}$
& $1.00\times10^{-3}$
& $67.5\times$
& under \\

18 & 1\,000
& $11.51\times$
& [10.76, 12.36]
& $11.49\times$ [10.79, 12.36]
& $2.97\times10^{-5}$
& $8.96\times10^{-4}$
& $30.2\times$
& under \\

18 & 2\,000
& $11.15\times$
& [10.09, 12.19]
& $11.50\times$ [10.01, 12.90]
& $7.84\times10^{-5}$
& $3.63\times10^{-3}$
& $46.3\times$
& under \\

18 & 5\,000
& $8.72\times$
& [8.12, 9.40]
& $8.85\times$ [8.15, 9.23]
& $4.35\times10^{-4}$
& $1.32\times10^{-2}$
& $30.5\times$
& under \\

18 & 10\,000
& $8.35\times$
& [7.76, 8.90]
& $8.17\times$ [8.02, 9.12]
& $3.75\times10^{-3}$
& $2.62\times10^{-2}$
& $7.0\times$
& over \\

18 & 20\,000
& $6.20\times$
& [5.22, 7.30]
& $6.16\times$ [4.79, 7.78]
& $2.06\times10^{-2}$
& $5.48\times10^{-2}$
& $2.7\times$
& over \\

\midrule

20 & 500
& $10.02\times$
& [9.01, 11.08]
& $10.38\times$ [8.39, 11.49]
& $1.02\times10^{-6}$
& $3.55\times10^{-4}$
& $349\times$
& under \\

20 & 1\,000
& $9.14\times$
& [7.78, 10.39]
& $9.75\times$ [8.70, 10.78]
& $2.79\times10^{-6}$
& $3.29\times10^{-4}$
& $118\times$
& under \\

20 & 2\,000
& $9.18\times$
& [8.03, 10.43]
& $9.24\times$ [7.95, 10.85]
& $2.50\times10^{-6}$
& $7.57\times10^{-4}$
& $303\times$
& under \\

20 & 5\,000
& $8.68\times$
& [8.21, 9.26]
& $8.65\times$ [7.96, 9.27]
& $2.70\times10^{-5}$
& $2.58\times10^{-3}$
& $95.4\times$
& under \\

20 & 10\,000
& $8.86\times$
& [7.85, 10.46]
& $8.52\times$ [7.59, 9.26]
& $1.79\times10^{-4}$
& $8.20\times10^{-3}$
& $45.8\times$
& under \\

20 & 20\,000
& $7.25\times$
& [6.71, 7.71]
& $7.43\times$ [7.20, 7.67]
& $1.23\times10^{-3}$
& $1.77\times10^{-2}$
& $14.5\times$
& over \\

\bottomrule
\end{tabular}
\end{table*}

\textit{Table Notes:}
Reported runtime overheads correspond to the ratio
$t_{\mathrm{BASS}}/t_{\mathrm{fixed}}$
computed independently for each random-circuit realization.
Configuration-wise overhead statistics are aggregated using geometric
means because runtime ratios are multiplicative and strongly right-skewed.
Reported confidence intervals correspond to 95\% bootstrap percentile
intervals of the geometric mean.
Median overheads are accompanied by interquartile ranges [25th, 75th
percentiles], which provide a robust estimate of configuration-to-configuration
variability.
The undersampled and oversampled classifications are determined using the
empirical crossover criterion
$\mathrm{PR}_Z \sim 2^{0.7N}$,
which estimates the effective participation-ratio scale separating strong-
and weak-compression regimes.

\subsection{Comparison with Tensor Network Methods}
\label{sec:mps_comparison}
Sparse-state simulation and tensor-network approaches exploit fundamentally distinct structural characteristics of quantum states to produce classical compression. Sparse-state approaches leverage computational-basis concentration: where only a tiny percentage of the Hilbert space has significant probability mass, the state can be efficiently represented by storing and updating only non-zero amplitudes. Tensor-network techniques, such as Matrix Product States (MPS), employ compression by using low bipartite entanglement, expressing the global state as network of low-rank tensors with bond dimensions reflecting subsystem correlations.

The optimal choice between the two paradigms is determined by the physical resource, most limiting for the target circuit. Fixed-basis sparse simulation and BASS use computational sparsity, which is measured globally by the PR $\text{PR}_Z$. MPS approaches use low entanglement, as measured by the bipartite Von-Neumann entropy $S$.

These resources are orthogonal. A product state in the $X$ basis has a maximal PR $\text{PR}_Z = 2^N$ but zero entanglement entropy $S = 0$. This makes MPS representation trivial, but significantly worse for sparse simulations. A highly localized cat state $(\ket{0\cdots0} + \ket{1\cdots1})/\sqrt{2}$ has $\text{PR}_Z = 2$ but $S = 1$~bit. It is easy for sparse methods but requires a nontrivial MPS bond dimension. BASS closes the gap by systematically pushing states towards lower $\text{PR}_Z$ via local basis adaptation without modifying $S$, widening the domain of effective sparse simulation.

We can formalise this complementary relationship by counting analytical parameters. For an MPS with bond dimension $\chi$, the state representation needs around $2N\chi^2$ parameters. BASS, with a truncation budget of $k$, takes around $2k$ parameters and an extra $\mathcal{O}(N)$ overhead for the local basis rotation matrices. As a result, sparse simulators naturally provide a more economical parametric representation when the required bond dimension exceeds the threshold:

\begin{equation}
\chi > \sqrt{\frac{k}{N}}.
\label{eq:crossover_mps}
\end{equation}

Rather than one paradigm completely dominating the other, traditional MPS techniques and BASS perform best under discrete regimes. By directly obtaining Schmidt spectra from the actual state evolution of our target ensembles, we can analytically constrain the MPS bond dimension. For shallow one-dimensional circuits and strictly area-law dynamics, $\chi$ remains tiny, making MPS maximally efficient. However, for deeper two-dimensional geometries or structurally sophisticated analyses where the required bond dimension in one-dimensional tensor-network representations becomes prohibitively large, BASS provides a viable alternative, assuming the state exhibits latent computational-basis sparsity.

\begin{figure*}[!ht]
\centering
\includegraphics[width=0.75\linewidth]{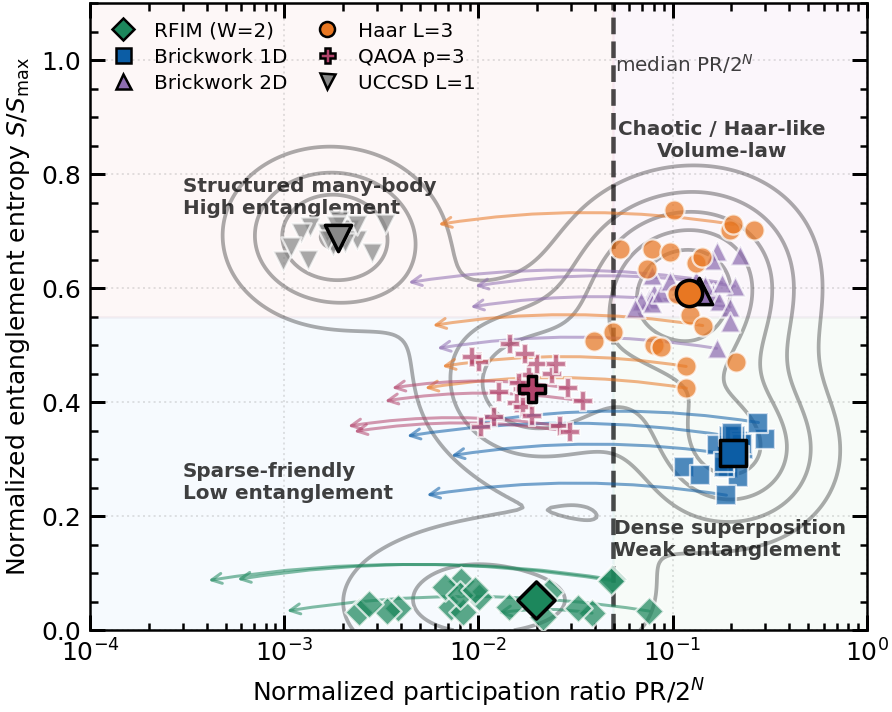}
\caption{
\textbf{Sparsity-entanglement phase diagram for representative quantum circuit families.}
Results shown for $N=16$ qubits with truncation budget $k=4096$.
The horizontal axis measures normalized PR
$\text{PR}/2^N$, quantifying Hilbert-space delocalization, while the
vertical axis shows normalized bipartite entanglement entropy
$S/S_{\max}$. Distinct circuit families populate different regions of the diagram:
RFIM ($W=2$) circuits occupy the sparse-friendly, low-entanglement regime;
UCCSD ($L=1$) circuits sit in the highly entangled but heavily localized upper-left quadrant;
QAOA ($p=3$) exhibits intermediate structural growth in both resources;
while Haar-random ($L=3$) and Brickwork (1D and 2D) circuits rapidly approach the chaotic volume-law regime in the
upper-right. Smooth contour lines indicate the empirical density distribution of
sampled states. The vertical dashed line marks the median $\text{PR}/2^N$ of the exact states.
Arrows show how BASS shifts states leftward (lower effective $\text{PR}$) without changing the global entanglement entropy.}
\label{fig:phase_diagram}
\end{figure*}

Figure~\ref{fig:phase_diagram} places several circuit families on a sparsity--entanglement plane: normalized PR $\text{PR}/2^N$ (delocalization in the $Z$ basis) versus normalized half-chain entropy $S/S_{\max}$.

TFIM and RFIM sit low on the plot—weak entanglement and fairly localized amplitudes—so both sparse paradigms do fine. UCCSD is different: $S/S_{\max} \approx 0.69$ but $\text{PR}/2^N \approx 1.87 \times 10^{-3}$, so it lands in the upper-left corner where entropy is high yet the state still occupies a thin slice of Hilbert space. QAOA sits between these extremes. Deep Haar-random and brickwork (1D/2D) circuits drift to the upper right, where entanglement and $Z$-basis spread are both large—the usual worst case for sparse truncation.

The vertical line $\text{PR}/2^N \approx k/2^N$ separates oversampled states (left) from undersampled ones (right). Past that line the effective support exceeds budget~$k$ and fixed-basis fidelity collapses.

BASS shifts points leftward by rotating each qubit into its local RDM eigenbasis before truncation, lowering effective $\text{PR}_Z$ without changing $S$. That moves some upper-right states back into the tractable wedge at the same~$k$.

\section{\label{sec:discussion}Discussion}

\paragraph*{Physical mechanism: why basis adaptation works.} The theoretical and numerical results together explain the adaptive-basis
mechanism. Lemma~\ref{lem:topk} establishes that top-$k$ selection is
optimal for a single truncation step within any fixed basis.
Lemma~\ref{lem:diversity} shows that entanglement-aware diversity rules pay
for their diversity by lowering the retained probability at the same budget.
These results separate support selection from representation choice: once the
single-step fixed-basis update is specified, adaptive basis selection becomes
the natural route to higher fidelity with fixed memory. Theorem~\ref{thm:stationarity}
then identifies the single-qubit RDM eigenbasis as a stationary local choice,
in direct analogy to natural orbitals.

The physical intuition is straightforward: local basis rotations can concentrate
probability mass into fewer effective amplitudes. For each qubit, the dominant
eigenvector of its one-qubit RDM defines a locally preferred axis. If this axis
is aligned with the computational $Z$ basis, the fixed basis is already a good
local representation. If it is tilted away from $Z$, an adaptive product basis
can reduce the number of computational-basis strings needed to retain the same
probability weight. This concentration is quantified globally by the
PR, and the largest gains occur when $\text{PR}_Z \gg k$.

\begin{figure}[!hbt]
\centering
\includegraphics[width=\columnwidth]{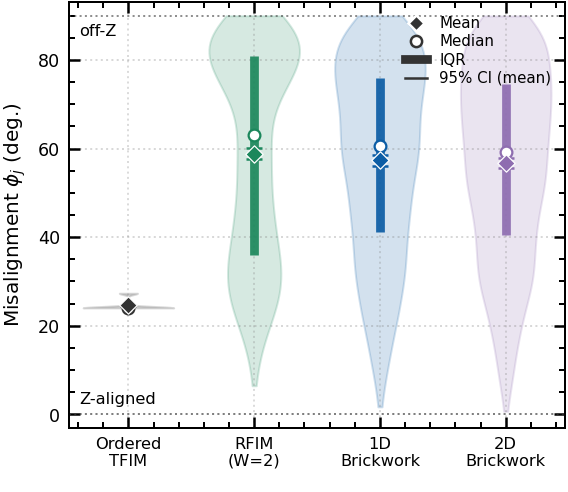}
\caption{\textbf{Per-qubit basis-misalignment diagnostic.}
\textbf{Left:} Distribution of single-qubit RDM basis-misalignment angles $\phi_j$ for
$N=12$ circuits. Diamonds mark means, open circles mark medians, and thick
vertical bars show the interquartile range (IQR). The thin vertical lines represent the 95\% confidence interval of the mean.
\textbf{Right:} Scatter plot of the computational-basis PR $\text{PR}_Z$ versus the mean misalignment angle $\bar{\phi}$ for individual circuit realizations (small markers), with large circles denoting the ensemble averages. Ordered TFIM is
closest to the computational basis, with mean $\bar{\phi}\approx 25^\circ$. RFIM ($W=2$) is significantly more tilted
($\bar{\phi}\approx 59^\circ$), while 1D and 2D brickwork circuits are broadly
off-$Z$, with $\bar{\phi}\approx 58^\circ$ and $\bar{\phi}\approx 57^\circ$, respectively.}
\label{fig:rotation_angles}
\end{figure}

We diagnose the local mechanism using the per-qubit basis-misalignment angle
\(
\phi_j = \min(\theta_j, 180^\circ-\theta_j),
\)
where $\theta_j$ is the polar angle between the dominant RDM eigenvector and
$\ket{0}$. This definition treats eigenvectors near either $\ket{0}$ or
$\ket{1}$ as $Z$-aligned. Thus $\phi_j=0^\circ$ means the locally preferred
basis is computational, while $\phi_j=90^\circ$ means it lies in the equatorial
plane and is maximally off-$Z$.

Figure~\ref{fig:rotation_angles} separates locally structured and delocalized regimes. Ordered TFIM has a narrow distribution near
$\bar{\phi}\simeq25^\circ$ and a very small computational-basis PR,
$\text{PR}_Z\simeq3.0$. RFIM exhibits a significantly larger local tilt, with a mean $\bar{\phi}\simeq59^\circ$ and median $\phi\simeq63^\circ$, corresponding to an increased $\text{PR}_Z\simeq1.8\times10^2$. The brickwork ensembles also display broad, highly
off-$Z$ misalignment distributions with means near $57^\circ\text{--}58^\circ$, but are accompanied by much larger
computational-basis PRs:
$\text{PR}_Z\simeq1.16\times10^3$ for 1D brickwork and
$\text{PR}_Z\simeq7.73\times10^2$ for 2D brickwork.

Theorem~\ref{thm:stationarity} guarantees stationarity, not global minimality.
For entangled states, product-unitary rotations cannot remove inter-qubit
correlations. The diagnostic should therefore be read as a local mechanism
test: adaptive rotations remove avoidable one-qubit basis misalignment before
top-$k$ support selection is applied. The broad brickwork distributions show
why adaptive basis selection can substantially improve retained probability at
fixed sparse budget, while the remaining large PRs show why
single-qubit rotations alone cannot fully compress highly entangled states.

BASS provides exponentially growing advantage for circuits that generate
substantial $Z$-basis delocalization (brickwork, Haar-random, 2D brickwork),
modest but practically useful improvement for intermediate delocalization
(RFIM), and correctly reverts to fixed-basis behavior for circuits that are
naturally $Z$-sparse (TFIM, UCCSD). The $\PR_Z$ crossover criterion is a
reliable, computable predictor of which regime applies.

\paragraph*{Product-unitary limitations and extensions.}
Single-qubit rotations cannot remove inter-qubit entanglement. For a Bell
state $\ket{\Phi^+}$, every qubit's RDM is maximally mixed
($\lambda = 1/2$), so no product-unitary rotation can reduce the
PR. This is a fundamental limitation of the single-qubit
BASS approach.

Block-unitary extensions that perform Schmidt decompositions over contiguous
qubit blocks of size~$b$ could capture multi-body correlations at cost
$\mathcal{O}(\log k + 2^{2b}k)$ per block. The $b=2$ case is already partially
realized in the brick-wall sweep of Sec.~\ref{subsec:2qrdm}. Systematic
investigation of block-BASS strategies, including the optimal tradeoff
between block size and computational overhead, is an important direction for
future work.

\paragraph*{Connection to natural-orbital theory. }
The connection between BASS and L\"owdin's natural
orbitals~\cite{lowdin1955} is a useful structural analogy. In quantum
chemistry, natural orbitals are eigenvectors of the one-body reduced density
matrix, and the resulting CI expansion is often much more compact. BASS applies
the analogous local idea to quantum circuit simulation: the single-qubit RDM
eigenbasis supplies a stationary product-basis update for the sparse
representation, and the do-no-harm guard accepts the update only when it
reduces the PR. The single-qubit RDM is therefore the
circuit-simulation analog of the one-body density matrix, while the local
eigenbasis rotation is the analog of a natural-orbital transformation.

Ratini et al.~\cite{ratini2024} recently showed that the natural-orbital
basis also improves mutual-information sparsity in quantum chemistry,
providing an independent information-theoretic motivation. BASS adapts this
principle from static ground-state calculations to dynamic circuit simulation,
applying local RDM rotations online at each optimization step with the same
asymptotic memory cost.

Orbital-optimization ideas from quantum chemistry—complete active spaces~\cite{roos1980}, natural-orbital truncations~\cite{tubman2016,tubman2020}, adaptive CI—may transfer to circuit simulation with little change in the one-body RDM object.

\paragraph*{Implications for the quantum advantage boundary.} The classical simulation boundary—where existing methods stop matching circuit outputs—is central to reading quantum-advantage experiments~\cite{preskill2018,arute2019,kim2023}. Better classical tools tighten benchmarks and narrow the hard regimes. BASS works on computational-basis concentration, which is largely independent of the entanglement structure MPS exploits and the Clifford structure stabilizer methods exploit.

For circuits where $\PR_Z$ grows exponentially with $N$ (brickwork, 2D
brickwork, Haar-random)---precisely those where random-circuit
anticoncentration~\cite{dalzell2022} and entanglement growth under random
unitaries~\cite{nahum2017,chan2018} drive the state away from any fixed
basis---BASS provides exponentially growing advantage over
fixed-basis methods: the adaptive-basis/fixed-basis fidelity ratio scales as
$\sim 2^{0.7N}$ for brickwork circuits (Table~\ref{tab:fixedk_scaling}).

The $\PR_Z$ crossover rule is easy to compute from the sparse support and
predicts when basis adaptation pays off. On our ensembles the fidelity ratio
sharpens near $k/\PR_Z \approx 1$ (Fig.~\ref{fig:crossover}) and lines up across
circuit families once $\PR_Z$ is estimated from the retained amplitudes.

A hybrid simulator could call BASS when $\PR_Z$ is large and MPS when
entanglement is low, using Fig.~\ref{fig:phase_diagram} as a coarse selector.
The natural-orbital link and block extensions are natural next steps toward mixing sparse
and tensor-network representations.

\section{\label{sec:conclusions}Conclusions and Outlook}
Three theoretical results characterize the part of sparse-state simulation most
relevant to this work. Lemma~\ref{lem:topk} closes the single-step support-selection problem:
top-$k$ by amplitude is optimal in any fixed basis, up to boundary ties.
Lemma~\ref{lem:diversity} eliminates the most natural class of
alternatives: diversity-preserving strategies are provably suboptimal at
a fixed budget. Theorem~\ref{thm:stationarity} identifies a principled adaptive
lever: the single-qubit RDM eigenbasis is a stationary point of the
participation-ratio objective, with controlled error that vanishes for weakly
entangled qubits. Fixed bases fix the greedy support rule; adaptive local bases are the remaining knob for concentrating weight at fixed~$k$.

BASS implements the
natural-orbital principle imported from quantum chemistry's 70-year-old
theory of L\"owdin natural orbitals~\cite{lowdin1955}-for quantum circuit
simulation in the same spirit. The
efficient $\mathcal{O}(Nk)$ RDM computation per optimization call, hash-table-based
updates, coordinate descent with do-no-harm guard, and diagonal-gate fast path
make BASS practical at the same memory cost as the fixed-basis method.

The ratio $\mathrm{PR}_Z$ predicts when BASS beats fixed-basis truncation in our tests up to $N=20$. At fixed $k=8192$, the brickwork fidelity ratio grows with $N$ ($\approx 5.1\times$ at $N=18$, $\approx 16.1\times$ at $N=20$). Under heavy compression ($k=100$, RFIM, $N=20$), BASS reaches $F \approx 5.2 \times 10^{-2}$ versus $F \approx 1.2 \times 10^{-2}$ for the fixed basis. On circuits that are already $Z$-sparse, the adaptive trigger switches off and fidelity matches the fixed-basis baseline. Wall-clock cost stays within roughly one order of magnitude ($\sim 6$–$13\times$ in the undersampled regime). Compared to MPS, BASS uses fewer parameters when $\chi > \sqrt{k/N}$, a regime common in 2D or nonlocal circuits (Sec.~\ref{sec:mps_comparison}).

Natural next steps include block rotations (cost $\sim 2^{2b}$ per block of size~$b$), hybrids that call BASS when $\PR_Z$ is large and MPS when entanglement is low, and block sizes chosen from local entropy estimates. Methods from quantum chemistry—complete active spaces, adaptive CI, natural-orbital hierarchies~\cite{lowdin1955,ratini2024,roos1980}—may port cleanly because the same one-body RDM object appears in both settings.

BASS adds adaptive local bases to the sparse-simulation toolbox alongside tensor networks, stabilizer methods, and fixed-basis truncation. It targets computational-basis concentration, a resource orthogonal to entanglement-area laws and Clifford structure, and the $\PR_Z$ rule gives a simple check for when that resource is worth spending simulation time on.

\section*{Code and data availability}

A reference implementation of BASS, the fixed-basis sparse comparison method, and the
analysis scripts used to generate the figures and tables will be made
available at \url{https://github.com/Skill0issue/BASS-Basis-Adapative-Sparse-Simulations-}. The benchmark data needed
to reproduce the numerical claims will be released with the same repository or
through a persistent data repository before publication.

\section{Acknowledgment}
AK acknowledges support from the INSPIRE scholarship provided by the Department of Science and Technology (DST), Government of India. SB acknowledges seed grant funding from IIT Hyderabad. BS acknowledges a fellowship granted by the Infosys Foundation.

\newpage

\appendix

\makeatletter
\@addtoreset{figure}{section}
\def\thefigure{\thesection.\arabic{figure}}
\makeatother

\section{\label{app:proofs}Extended Proofs}

\subsection{Detailed derivation of the gradient expression (Theorem~\ref{thm:stationarity}, Step~1)}

We provide the full derivation of Eq.~\eqref{eq:grad_pr}. Write the
inverse PR as
\begin{equation}
\mathcal{I} = \sum_{x}|\phi_x|^4,
\end{equation}
where $\ket{\phi} = (\bigotimes_j U_j^\dagger)\ket{\psi}$. Under a
small rotation $U_j(\theta) = e^{-i\theta G_j} U_j$, the rotated amplitudes
become
\begin{equation}
\phi_x(\theta) = \phi_x - i\theta \sum_{a} (G_j)_{b_j(x),a}\,
\phi_{x \oplus (a \oplus b_j(x))\cdot 2^j} + O(\theta^2),
\end{equation}
where $b_j(x)$ is the $j$-th bit of $x$. Differentiating the inverse
PR $\text{IPR} = \sum_x |\phi_x|^4$:
\begin{align}
&\frac{d\,\text{IPR}}{d\theta}\bigg|_{\theta=0}
= 4\,\text{Re}\!\left(
\sum_x |\phi_x|^2 \phi_x^* \frac{d\phi_x}{d\theta}\bigg|_0
\right) \nonumber \\
&= -4\,\text{Re}\!\left(
i \sum_{b,a} (G_j)_{ba} \sum_{\bar{x}_j}
|\phi_{(\bar{x}_j,b)}|^2\,
\phi^*_{(\bar{x}_j,b)}\,\phi_{(\bar{x}_j,a)}
\right).
\end{align}
Identifying the $2 \times 2$ matrix
$(\Lambda_j)_{ba} = \sum_{\bar{x}_j} |\phi_{(\bar{x}_j,a)}|^2\,
\phi^*_{(\bar{x}_j,a)}\,\phi_{(\bar{x}_j,b)}$
and using $d\PR/d\theta = -(1/\text{IPR}^2)\, d\text{IPR}/d\theta$,
we recover Eq.~\eqref{eq:grad_pr} for the IPR gradient and the corresponding
participation-ratio gradient with the positive prefactor
$1/\text{IPR}^2$. This prefactor does not affect the stationarity condition.
The gradient vanishes for all traceless Hermitian generators $G_j$ if and only
if $(\Lambda_j)_{01} = 0$.

\subsection{Covariance decomposition and the role of entanglement coherence}

We elaborate on Step~2 of the proof. The off-diagonal element is
\begin{equation}
(\Lambda_j)_{01} = \sum_{\bar{x}_j} |\phi_{(\bar{x}_j,0)}|^2\,
\phi^*_{(\bar{x}_j,0)}\,\phi_{(\bar{x}_j,1)}.
\end{equation}
Define the weight $w_{\bar{x}} = |\phi_{(\bar{x},0)}|^2$ and the coherence
$c_{\bar{x}} = \phi^*_{(\bar{x},0)}\,\phi_{(\bar{x},1)}$. Since $U_j$
diagonalizes $\rho_j$:
\begin{equation}
(\rho_j)_{01} = \sum_{\bar{x}_j}
\phi^*_{(\bar{x}_j,0)}\,\phi_{(\bar{x}_j,1)} = \sum_{\bar{x}} c_{\bar{x}} = 0.
\end{equation}
Hence the coherences have zero mean. This allows us to write
\begin{equation}
(\Lambda_j)_{01} = \sum_{\bar{x}} w_{\bar{x}}\, c_{\bar{x}}
= \sum_{\bar{x}}(w_{\bar{x}} - \bar{w})\, c_{\bar{x}},
\end{equation}
where $\bar{w} = D^{-1}\sum_{\bar{x}} w_{\bar{x}}$ and the equality uses
$\sum_{\bar{x}} c_{\bar{x}} = 0$. This covariance structure reveals
that the gradient vanishes when:
\begin{enumerate}
\item The coherences $c_{\bar{x}}$ are identically zero (product state), or
\item The weight distribution $w_{\bar{x}}$ is uniform
($w_{\bar{x}} = \bar{w}$ for all $\bar{x}$), or
\item The correlation between weight fluctuations and coherences
vanishes by cancellation.
\end{enumerate}
Case~(1) is the product-state limit. Case~(2) corresponds to a maximally
entangled state with flat amplitude distribution. Case~(3) can occur for
specially structured entangled states. In all three cases, the RDM
eigenbasis is exactly stationary.

\subsection{Tightness of the Cauchy-Schwarz bound}

The bound in Eq.~\eqref{eq:Rj} is not generically tight. By the
Cauchy-Schwarz inequality:
\begin{equation}
|\text{Cov}(w, c)|^2 \le \text{Var}(w) \cdot \sum_{\bar{x}} |c_{\bar{x}}|^2.
\end{equation}
Equality requires $w_{\bar{x}} - \bar{w} = \mu\, c_{\bar{x}}^*$ for some
constant $\mu$, i.e., the weight fluctuations must be proportional to
the complex conjugate of the coherences. For generic entangled states this
condition is not satisfied, and the bound overestimates the true gradient
residual. We further bound
$\sum_{\bar{x}} |c_{\bar{x}}|^2 \le \sum_{\bar{x}}
|\phi_{(\bar{x},0)}|^2 |\phi_{(\bar{x},1)}|^2$
by the Cauchy-Schwarz inequality on each term, yielding
Eq.~\eqref{eq:Rj}. For weakly entangled states
($\lambda_{\min}(\rho_j) \ll 1$), the factor
$\sum_{\bar{x}} |\phi_{(\bar{x},0)}|^2 |\phi_{(\bar{x},1)}|^2
\le \lambda_{\min}(\rho_j)$
scales linearly with the minor eigenvalue, ensuring rapid decay of the
bound as $S_j \to 0$.

\subsection{Global optimality: reduction to a known open problem}
Whether the RDM eigenbasis is a \emph{global} minimum of the participation
ratio over all product unitaries remains open. We show that this question
reduces to a known problem in multilinear algebra. The PR
under a product unitary $U = \bigotimes_j U_j$ is
\begin{equation}
\PR(U) = \frac{1}{\sum_x |\langle x| U^\dagger |\psi\rangle|^4}.
\end{equation}
Minimizing $\PR(U)$ is equivalent to maximizing the 4-norm of the
coefficient vector $\{\langle x| U^\dagger |\psi\rangle\}$ over
product unitaries. This is a special case of the problem of maximizing
tensor $p$-norms under product-unitary transformations, which is known
to be NP-hard in general~\cite{hillar2013} but may admit efficient
solutions for the structured case $p=4$ with single-qubit unitaries.
Our empirical evidence-zero reverts for brickwork/Haar circuits-suggests
that the local optimum found by coordinate descent coincides with the
global optimum in practice.

\subsection{Proof of Eq. \ref{eq:dominant_eigen}}
\label{sec:dominant_eigen}
We seek to find the dominant eigenvector of a $2 \times 2$ Hermitian matrix (e.g., a single-qubit density matrix) analytically. By avoiding iterative numerical eigenvalue routines, we guarantee a predictable, constant computational cost $O(1)$ per qubit during simulation loops.

Consider a general $2 \times 2$ Hermitian matrix representing the single-qubit state:
\begin{equation}
\rho_j = \begin{pmatrix} a & b \\ b^* & d \end{pmatrix}
\end{equation}
where $a, d \in \mathbb{R}$ and $b \in \mathbb{C}$. To isolate the spectrum, we construct the characteristic polynomial $\det(\rho_j - \lambda I) = 0$, which yields the quadratic equation:
\begin{equation}
\lambda^2 - (a+d)\lambda + (ad - |b|^2) = 0
\end{equation}
We define the matrix trace center (mean value) as $m = \frac{a+d}{2}$ and the half-difference as $\tau = \frac{d-a}{2}$. The eigenvalues evaluate exactly to $\lambda_\pm = m \pm \Delta$, where the scale parameter $\Delta$ is defined as:
\begin{equation}
\Delta = \sqrt{\tau^2 + |b|^2}
\end{equation}
Because $\Delta \ge 0$, the dominant eigenvalue corresponding to the larger spectral weight is uniquely identified as:
\begin{equation}
\lambda_\text{max} = m + \Delta
\end{equation}
To find the corresponding dominant eigenvector $\mathbf{v}_\text{max} = \begin{pmatrix} x \\ y \end{pmatrix}$, we solve the linear system $(\rho_j - \lambda_\text{max} I)\mathbf{v}_\text{max} = 0$. Extracting the first row gives the constraint:
\begin{equation}
(a - \lambda_\text{max}) x + b y = 0
\end{equation}
Substituting $\lambda_\text{max} = \frac{a+d}{2} + \Delta$ into the expression simplifies the first coefficient to $a - \lambda_\text{max} = -\tau - \Delta$. This reduces our constraint to:
\begin{equation}
(-\tau - \Delta) x + b y = 0
\end{equation}
Isolating $x$ yields $x = \frac{b}{\tau + \Delta} y$. Choosing a convenient non-zero scalar scaling by setting $y = \tau + \Delta$ yields the unnormalized dominant eigenvector:
\begin{equation}
\mathbf{v}_\text{unnorm} = \begin{pmatrix} b \\ \tau + \Delta \end{pmatrix}
\end{equation}
To satisfy the unit-norm criteria required for valid physical quantum states, we compute the vector norm $\|\mathbf{v}\| = \sqrt{|b|^2 + (\tau + \Delta)^2}$. Normalizing $\mathbf{v}_\text{unnorm}$ gives the final exact expression for the dominant eigenvector:
\begin{equation}
\mathbf{v}_\text{max} = \frac{1}{\sqrt{|b|^2 + (\tau+\Delta)^2}} \begin{pmatrix} b \\ \tau+\Delta \end{pmatrix}
\end{equation}
This analytical solution relies strictly on a fixed sequence of primitive algebraic operations (additions, multiplications, and a single real square root), preserving a bounded, constant execution runtime across scaling system profiles.

\subsection{Derivation of the 2-Qubit RDM}

This section details the derivation of the two-qubit RDM from an $n$-qubit pure state by tracing out the rest of the system. It explains the conceptual and computational motivations behind splitting basis vectors into block and rest indices and provides a condensed step-by-step example for a 3-qubit system. Interpreting the notation, we have an $n$-qubit pure state
\(
\lvert\psi\rangle = \sum_i \alpha_i \lvert x_i\rangle,
\),
where each computational basis vector $\lvert x_i\rangle$ can be split into three parts relative to qubits $q_1,q_2$: first, $s(x_i)$, the 2-bit string (00, 01, 10, 11) that is the state of qubits $q_1,q_2$; and second, $r(x_i)$, the remaining $n-2$ qubits (the ``rest'' of the system).

The two-qubit RDM$\rho_{q_1 q_2}$ is obtained by tracing out (summing over) the rest of the system. Starting from the full density matrix, the pure state is $\rho = \lvert\psi\rangle\langle\psi\rvert = \sum_{i,i'} \alpha_i \overline{\alpha_{i'}}\, \lvert x_i \rangle\langle x_{i'}\rvert$. Each basis state $\lvert x_i\rangle$ is written as a tensor product of the subsystem (qubits $q_1,q_2$) and the rest, meaning $\lvert x_i\rangle = \lvert s(x_i)\rangle \otimes \lvert r(x_i)\rangle$. So,
\begin{equation}
\begin{split}
\rho = \sum_{i,i'} \alpha_i \overline{\alpha_{i'}}\,
&\big( \lvert s(x_i)\rangle\langle s(x_{i'})\rvert \big) \\
&\otimes \big( \lvert r(x_i)\rangle\langle r(x_{i'})\rvert \big).
\end{split}
\end{equation}

Defining the RDMon qubits $q_1,q_2$ as the partial trace over the rest yields $\rho_{q_1 q_2} = \mathrm{Tr}_{\text{rest}}(\rho)$. In components in the computational basis $\{\lvert s\rangle\}$ for the two qubits, $[\rho_{q_1 q_2}]_{s,s'} = \langle s \rvert\, \rho_{q_1 q_2} \,\lvert s' \rangle$. Using the definition of the partial trace, $\rho_{q_1 q_2} = \sum_r \langle r \rvert\, \rho \, \lvert r \rangle$, where $\{\lvert r\rangle\}$ runs over a basis of the remaining qubits. Therefore, $[\rho_{q_1 q_2}]_{s,s'} = \sum_r \langle s,r \rvert\, \rho \, \lvert s',r \rangle$.

Inserting the expansion of $\rho$ into this component equation gives
\begin{equation}
\begin{split}
[\rho_{q_1 q_2}]_{s,s'}
= \sum_r \sum_{i,i'} &\alpha_i \overline{\alpha_{i'}}\,
\langle s,r \vert s(x_i), r(x_i) \rangle \\
&\times \langle s(x_{i'}), r(x_{i'}) \vert s',r \rangle.
\end{split}
\end{equation}
The inner products factor into Kronecker deltas: $\langle s,r \vert s(x_i), r(x_i) \rangle = \delta_{s,\,s(x_i)}\, \delta_{r,\,r(x_i)}$ and $\langle s(x_{i'}), r(x_{i'}) \vert s',r \rangle = \delta_{s(x_{i'}),\,s'}\, \delta_{r(x_{i'}),\,r}$. Thus,
\begin{equation}
\begin{split}
[\rho_{q_1 q_2}]_{s,s'}
= \sum_r \sum_{i,i'} &\alpha_i \overline{\alpha_{i'}}\,
\delta_{s,\,s(x_i)}\, \delta_{r,\,r(x_i)}\, \\
&\times \delta_{s(x_{i'}),\,s'}\, \delta_{r(x_{i'}),\,r}.
\end{split}
\end{equation}

Performing the sum over $r$ enforces $r = r(x_i) = r(x_{i'})$, meaning we only keep terms with the same rest configuration:
\begin{equation}
\begin{split}
[\rho_{q_1 q_2}]_{s,s'}
= \sum_{i,i' :\, r(x_i) = r(x_{i'})} &\alpha_i \overline{\alpha_{i'}}\,
\delta_{s,\,s(x_i)}\, \\
&\times \delta_{s(x_{i'}),\,s'}.
\end{split}
\end{equation}
Using indicator functions $\mathbf{1}[\text{condition}]$ instead of Kronecker deltas, this becomes:
\begin{equation}
\begin{split}
[\rho_{q_1 q_2}]_{s,s'}
= \sum_{(i,i'):\, r(x_i)=r(x_{i'})} &\alpha_i \overline{\alpha_{i'}}\,
\mathbf{1}[s(x_i)=s]\, \\
&\times \mathbf{1}[s(x_{i'})=s'].
\end{split}
\end{equation}

There are two separate motivations for this partitioning strategy. Conceptually, it matches the mathematical definition of the partial trace. By writing the full density matrix in the computational basis and expressing $x_i$ as a combination of the rest and pair indices, you obtain exactly a sum over terms with the same rest index.

Computationally, this approach enables efficient grouping and summation. From an algorithmic point of view, one can treat the rest index $r$ like a key in a hash or sorting operation. States with the same $r$ are grouped together. Within each group, the only remaining variation is in the two-qubit bits $(q_1,q_2)$, which correspond to the four possible values of $s\in\{0,1,2,3\}$. Then, for each group $r$, one can collect all amplitudes into four buckets, one per $s$, and form the contributions to the $4\times 4$ matrix $[\rho_{q_1 q_2}]_{s,s'}$ from that group. This leads to the reported complexity: $\mathcal{O}(\log k)$ per pair for the sort/key operations and $\mathcal{O}(16)$ operations to fill a $4\times 4$ matrix per group, giving an overall $\mathcal{O}(Nk \log k)$ cost per sweep.

To illustrate this with a concrete example, consider a 3-qubit system (qubits 0,1,2) and a state
\[
|\psi\rangle = \alpha_{000} |000\rangle + \alpha_{010} |010\rangle + \alpha_{100} |100\rangle + \alpha_{110} |110\rangle,
\]
where $|x_2 x_1 x_0\rangle$ is the ordering convention. Suppose we want the two-qubit RDM for qubits $q_1 = 0$ and $q_2 = 1$. Let us denote the corresponding states as $x_0, x_1, x_2$, and $x_3$ based on their order of appearance.

For each state $x_i$, the block index is $s(x) = b_0(x) \cdot 2 + b_1(x)$ and the rest index $r(x)$ is simply the state of qubit 2. Grouping by the rest index yields two sets. The group for $r=0$ contains states $x_0$ ($s=0$) and $x_1$ ($s=1$). The group for $r=1$ contains states $x_2$ ($s=0$) and $x_3$ ($s=1$).

To compute a specific matrix element such as $[\rho_{01}]_{0,1}$, the formula dictates summing over pairs $(i,i')$ with the same rest index where $s(x_i)=0$ and $s(x_{i'})=1$. For $r=0$, this pairs state $x_0$ with $x_1$, contributing $\alpha_{000}\overline{\alpha_{010}}$. For $r=1$, this pairs state $x_2$ with $x_3$, contributing $\alpha_{100}\overline{\alpha_{110}}$. Thus, the final value is
\[
[\rho_{01}]_{0,1} = \alpha_{000} \overline{\alpha_{010}} + \alpha_{100} \overline{\alpha_{110}}.
\]
This exactly matches what you get by doing a direct partial trace over qubit 2. Similarly, the diagonal elements collect pairs with the same rest and identical $s$ values, yielding $[\rho_{01}]_{0,0} = |\alpha_{000}|^2 + |\alpha_{100}|^2$ and $[\rho_{01}]_{1,1} = |\alpha_{010}|^2 + |\alpha_{110}|^2$. The remaining non-zero elements follow directly by complex conjugation.

Thus, $s$ is a compact way to encode the pair of bits $(q_1,q_2)$ into a single index $0,1,2,3$, mapping directly to the 4 levels of the two-qubit subsystem. The index $r(x)$ encodes the configuration of all other qubits, making grouping by $r$ the exact operation needed for a partial trace over those qubits. The computational formula directly implements the mathematical definition of the reduced density matrix, mapping neatly into a highly efficient group-by-rest-index algorithmic operation.

\subsection{Empirical support for top-\texorpdfstring{$k$}{k} optimality}
\label{app:topk_vs_randk}
\begin{figure}[!hbt]
\centering
\includegraphics[width=\linewidth]{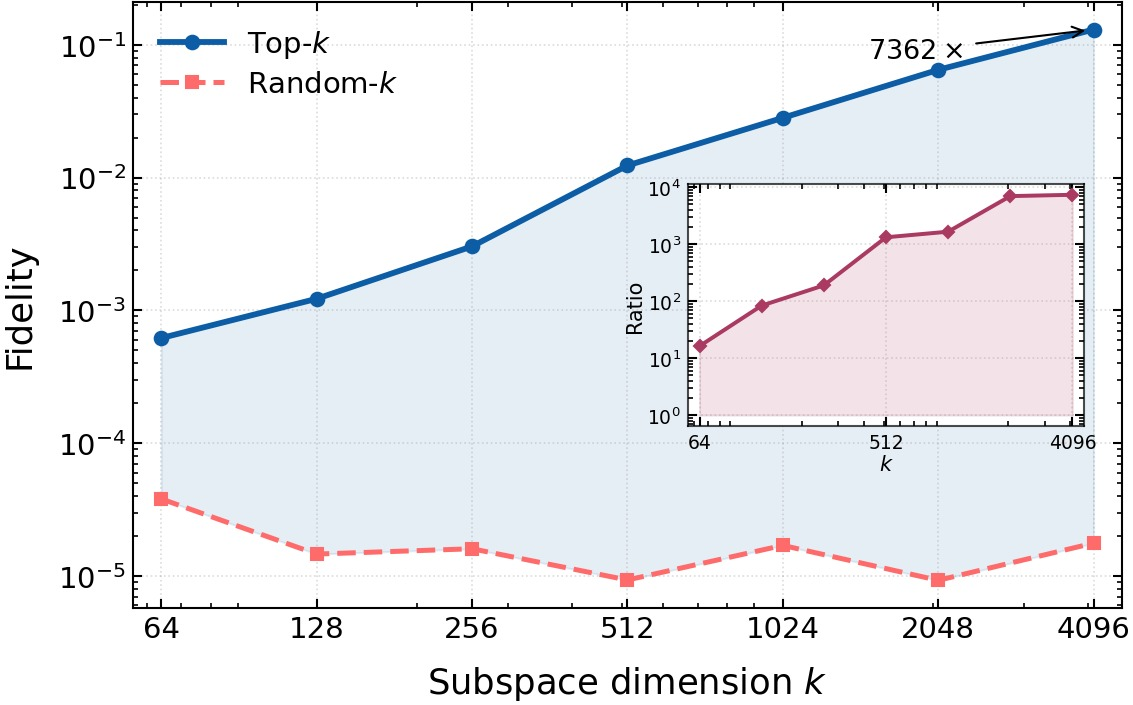}
\caption{\textbf{Empirical comparison of top-$k$ and random-$k$ truncation in fixed-basis simulation.}
Fidelity as a function of retained subspace dimension $k$ for
$N=16$ and circuit depth $L=3$. Top-$k$ truncation exhibits
stable monotonic fidelity growth with increasing $k$, while
random-$k$ truncation remains near the stochastic noise floor
across all tested subspace dimensions. The shaded region
marks the growing gap between the two strategies.
The inset shows the amplification ratio
$F_{\mathrm{top}\text{-}k}/F_{\mathrm{random}\text{-}k}$,
which grows superlinearly with $k$ and reaches
$\sim 7.4\times10^{3}$ at $k=4096$. These results provide
direct empirical support for Lemma~\ref{lem:topk}: structured retention
of the largest amplitudes retains more
probability mass far more efficiently than random support
selection.}
\label{fig:topk_vs_randk}
\end{figure}
We empirically validate Lemma~\ref{lem:topk} in Fig.~\ref{fig:topk_vs_randk}. For circuits of depth $L=3$ with $N=16$ qubits and truncation parameter $k=2048$ (30 independent trials), top-$k$ truncation yields a fidelity more than $10^4$ times higher than that obtained by uniform random-$k$ selection. This performance gap remains substantial as the depth increases: the top-$k$ method maintains a win rate exceeding 60\% relative to random-$k$ up to depth $L=15$, at which point the fidelities of both methods converge toward the noise floor.

\section{\label{app:natural_orbital}Natural-Orbital Theory Connection:
Extended Analysis}

\subsection{Formal correspondence}

Table~\ref{tab:NO_correspondence} establishes a term-by-term mapping
between the BASS framework and L\"owdin's natural-orbital
theory~\cite{lowdin1955}.

\begin{table}[!hbt]
\caption{\label{tab:NO_correspondence}Correspondence between BASS and
natural-orbital theory in quantum chemistry.}
\begin{ruledtabular}
\begin{tabular}{ll}
Quantum chemistry & BASS \\
\hline
\shortstack[l]{One-body RDM
$\gamma(\mathbf{r},\mathbf{r}')$}
& Single-qubit RDM $\rho_j$ \\
\shortstack[l]{Natural orbitals \\
(eigenvectors of $\gamma$)}
& RDM eigenbasis $U_j$ \\
Occupation numbers $n_i$
& Eigenvalues $\lambda_j^{(0)}, \lambda_j^{(1)}$ \\
CI expansion coefficients
& Sparse amplitudes $\alpha_i$ \\
\shortstack[l]{Maximal CI coefficient \\
in NO basis}
& Stationary local $\PR$ reduction \\
Full CI (exact)
& Exact state vector ($k = 2^N$) \\
Truncated CI
& Sparse-state simulation ($\mathcal O(k)$) \\
CAS (complete active space)
& Top-$k$ support set $\mathcal{S}^*$ \\
\shortstack[l]{Orbital optimization \\
(CASSCF)}
& \shortstack[l]{Basis optimization \\(coordinate descent)} \\
\end{tabular}
\end{ruledtabular}
\end{table}

The correspondence is structural rather than literal. In quantum chemistry,
the wavefunction is
expanded in Slater determinants built from molecular orbitals:
$|\Psi\rangle = \sum_I C_I |\Phi_I\rangle$. The natural orbitals are
obtained by diagonalizing the one-body reduced density matrix
$\gamma_{pq} = \langle\Psi| a_p^\dagger a_q |\Psi\rangle$, often producing a
more compact expansion with respect to the number of retained configurations.

In BASS, the wavefunction is expanded in computational basis states:
$|\psi\rangle = \sum_x \alpha_x |x\rangle$. The single-qubit RDM
$\rho_j = \text{Tr}_{\neq j}|\psi\rangle\langle\psi|$ is the direct
analog of the one-body RDM restricted to a single site. The eigenbasis
rotation $U_j$ is the analog of the natural-orbital transformation for
that site. Theorem~\ref{thm:stationarity} is the circuit-simulation
analog of a stationarity result for a locally adapted expansion, not a proof
of global minimality of the PR.

\subsection{Key differences}

Despite the formal correspondence between these two settings, there are several important differences that significantly affect how reduced density matrices and basis optimizations are used in practice.

First, there is a key difference in dimensionality. In quantum chemistry, the one-body reduced density matrix (RDM) is an $M \times M$ matrix, where $M$ is the number of orbitals. This number can easily reach into the hundreds for realistic molecular systems. As a result, diagonalizing this matrix and manipulating it numerically can be computationally demanding, scaling poorly as the system size grows. In contrast, in the BASS framework for quantum circuit simulation, the single-qubit RDM is always a fixed $2 \times 2$ matrix, independent of the overall system size or the number of qubits. This small, fixed dimensionality allows one to diagonalize the RDM analytically and at constant computational cost. Consequently, basis updates in BASS can be performed efficiently and uniformly throughout the simulation, without the overhead that typically accompanies orbital optimization in quantum chemistry.

Second, the underlying particle statistics differ fundamentally. Quantum chemistry deals with fermionic wavefunctions that are antisymmetric under particle exchange, reflecting the indistinguishability and exchange statistics of electrons. This antisymmetry imposes strong structural constraints on the allowed states and gives special significance to particle-number conservation. In particular, natural orbitals become especially powerful tools in this context, because they respect both antisymmetry and fixed particle number, and they often yield compact representations of correlated electronic states. By contrast, the computational-basis states used in generic quantum circuit simulations do not carry any built-in antisymmetry or particle-number constraints. Qubits are distinguishable degrees of freedom arranged in a tensor-product Hilbert space, and states in this space are not required to satisfy fermionic exchange properties. As a result, some of the structural advantages that natural orbitals offer in chemistry are simply absent in the general circuit setting, and the notion of an optimal basis must be reformulated accordingly.

Third, the character of locality in the two frameworks is quite distinct. In electronic structure theory, natural orbitals are often delocalized across the entire molecule. Even though they can offer compactness or interpretability, they typically do not respect a simple tensor-product structure over spatial regions or subsystems; each orbital can have support on many atomic centers. In contrast, BASS explicitly restricts itself to strictly local operations in the qubit picture: the basis rotations are single-qubit unitaries. These rotations act locally on individual qubits and preserve the global tensor-product decomposition of the Hilbert space. This locality constraint is not a minor technical detail; it is central to the efficiency and scalability of BASS, because it ensures that basis updates do not generate additional entanglement or long-range couplings beyond those already present in the circuit description.

Finally, there is an important difference between online and offline optimization. In the CASSCF method and related orbital-optimization techniques in quantum chemistry, one typically optimizes orbitals iteratively with respect to a fixed target state, most commonly the electronic ground state. This optimization is carried out offline with respect to a stationary wavefunction: one solves a self-consistent field problem, updates orbitals, recomputes the state, and repeats until convergence. By contrast, BASS is designed to optimize the qubit basis online, during the actual evolution of a quantum circuit. At each step of the circuit, BASS can adapt the local basis to the instantaneous quantum state being generated. This dynamic, state-dependent updating allows the method to track and exploit transient structure in the evolving state, rather than relying on a single, globally optimized basis fixed beforehand. Thus, while there is a formal analogy between natural-orbital optimization in chemistry and basis adaptation in BASS, the two approaches differ greatly in their dimensionality, statistical structure, locality properties, and whether basis optimization is performed offline for a stationary state or online throughout a circuit evolution.

\section{\label{app:pseudocode}Algorithmic Pseudocode}
Appendix~\ref{app:pseudocode} spells out the BASS implementation in the same
order as Sec.~\ref{sec:algorithm}: rotated-frame propagation
(Sec.~\ref{subsec:core_idea}), hash-table RDM kernels
(Sec.~\ref{subsec:rdm}), multi-pass basis sweeps with the do-no-harm guard
(Secs.~\ref{subsec:coord_descent} and~\ref{subsec:doNH}), and the
implementation shortcuts of Sec.~\ref{subsec:optimizations}.
After each pseudocode block, we walk through what the lines do and how the
routines call one another.
\begin{algorithm}[H]
\caption{BASS: Main Simulation Loop}
\label{alg:bass_main}
\begin{algorithmic}[1]

\Require Circuit gates $\{G_t\}_{t=1}^M$,
sparse budget $k$,
optimization interval $n_{\mathrm{opt}}$,
truncation interval $n_{\mathrm{trunc}}$,
hard cap
$K_{\mathrm{hard}} = \min(c_{\mathrm{hard}}k, 2^N)$

\Ensure Approximate sparse state $\ket{\tilde{\psi}}$,
basis accumulators $\{U_j\}$

\State Initialize
$\ket{\tilde{\psi}} \leftarrow \ket{0}^{\otimes N}$

\State Initialize
$U_j \leftarrow I_2$
for all $j = 0,\ldots,N-1$

\State $\PR_{\mathrm{last}} \leftarrow 0$

\State $n_{\mathrm{since\_trunc}} \leftarrow 0$

\For{$t = 1,\ldots,M$}

\State Let $G_t$ act on qubits $(q_1,q_2)$

\State
$\tilde{G}_t
\leftarrow
(U_{q_1}\otimes U_{q_2})^\dagger
G_t
(U_{q_1}\otimes U_{q_2})$

\If{$\tilde{G}_t$ is diagonal}

\State
\Call{DiagonalApply}
{$\ket{\tilde{\psi}},\tilde{G}_t$}

\Else

\State
\Call{HashTableApply}
{$\ket{\tilde{\psi}},\tilde{G}_t,K_{\mathrm{hard}}$}

\EndIf

\State
$n_{\mathrm{since\_trunc}}
\leftarrow
n_{\mathrm{since\_trunc}} + 1$

\If{$|\mathrm{supp}(\tilde{\psi})| > K_{\mathrm{hard}}$}

\State
Truncate to top-$k$

\State
Renormalize $\ket{\tilde{\psi}}$

\State
$n_{\mathrm{since\_trunc}} \leftarrow 0$

\ElsIf{
$|\mathrm{supp}(\tilde{\psi})| > k$
\textbf{and}
$n_{\mathrm{since\_trunc}} \ge n_{\mathrm{trunc}}$
}

\State
Truncate to top-$k$

\State
Renormalize $\ket{\tilde{\psi}}$

\State
$n_{\mathrm{since\_trunc}} \leftarrow 0$

\EndIf

\If{optimization check step is active}

\State
$\PR_{\mathrm{cur}}
\leftarrow
1/\sum_i |\alpha_i|^4$

\If{
$\PR_{\mathrm{cur}}$
increased sufficiently relative to
$\PR_{\mathrm{last}}$
}

\Comment{default trigger threshold:
$\tau \approx 0.90$}

\State
\Call{BasisOptimize}
{$\ket{\tilde{\psi}},\{U_j\}$}

\State
$\PR_{\mathrm{last}}
\leftarrow
1/\sum_i |\alpha_i|^4$

\EndIf

\EndIf

\EndFor

\State
\Return
$\ket{\tilde{\psi}},\{U_j\}$

\end{algorithmic}
\end{algorithm}

BASS stores an approximate state as a sparse amplitude list
$\ket{\tilde{\psi}} = \sum_{i=1}^{|\mathrm{supp}|} \alpha_i \ket{x_i}$ with
$|\mathrm{supp}| \le K_{\mathrm{hard}}$ and a target budget~$k$.
Basis adaptation is encoded in per-qubit accumulators
$U_j \in U(2)$, $j=0,\ldots,N-1$, so that all gate conjugations and RDMs are
evaluated in the accumulated product frame
(Eq.~\eqref{eq:rotated_frame}). The PR (PR) of the current
support is
\begin{equation}
\PR(\tilde{\psi}) = \frac{\bigl(\sum_i |\alpha_i|^2\bigr)^2}
{\sum_i |\alpha_i|^4}.
\end{equation}
Truncation retains the top-$k$ amplitudes by magnitude, followed by
renormalization so that $\sum_i |\alpha_i|^2 = 1$.

The reported sparse budget is~$k$. Transient support may grow up to
$K_{\mathrm{hard}} = \min(c_{\mathrm{hard}} k, 2^N)$ with a default
$c_{\mathrm{hard}}=8$. Deferred truncation waits for at least
$n_{\mathrm{trunc}}$ gate steps while $k < |\mathrm{supp}| \le K_{\mathrm{hard}}$.
Optimization checks run on truncation-active steps every $n_{\mathrm{opt}}$
gates. They are skipped unless the PR has risen by roughly $10\%$ since the
last accepted sweep ($r_{\mathrm{opt}}=\tau^{-1}\approx 1.11$ for
$\tau\approx 0.90$). Basis adaptation never replaces top-$k$ truncation; it
only changes the basis in which that rule is applied.

\subsection{Main simulation loop}
\label{app:alg_main}
Algorithm~\ref{alg:bass_main} drives the simulation: apply gates in the
accumulated local basis, defer truncation where possible, and call
\textsc{BasisOptimize} only when the schedule and PR-growth tests both pass.

\begin{algorithm}[H]
\caption{Single-Qubit Basis Optimization}
\label{alg:basis_opt}
\begin{algorithmic}[1]

\Procedure{BasisOptimize}{$\ket{\tilde{\psi}},\{U_j\}$}

\State
$n_{\mathrm{passes}} \leftarrow 3$

\For{pass $=1,\ldots,n_{\mathrm{passes}}$}

\State
$\mathrm{improved} \leftarrow \mathrm{False}$

\State
Build hash table $\mathcal{H}$ from $\ket{\tilde{\psi}}$

\State
Compute all $N$ single-qubit RDMs
from the current sparse state
via Algorithm~\ref{alg:rdm}

\Comment{RDMs reused throughout this sweep}

\State
$\PR_{\mathrm{old}}
\leftarrow
\PR(\tilde{\psi})$

\For{$j=0,\ldots,N-1$}

\State
Diagonalize
$\rho_j = V_j \Lambda_j V_j^\dagger$

\State
$\ket{\tilde{\psi}'}
\leftarrow
(I\otimes\cdots\otimes V_j^\dagger\otimes\cdots\otimes I)
\ket{\tilde{\psi}}$

\State
Truncate $\ket{\tilde{\psi}'}$ to top-$k$

\State
Renormalize $\ket{\tilde{\psi}'}$

\State
$\PR_{\mathrm{new}}
\leftarrow
\PR(\tilde{\psi}')$

\If{$\PR_{\mathrm{new}} < \PR_{\mathrm{old}}$}

\State
$\ket{\tilde{\psi}}
\leftarrow
\ket{\tilde{\psi}'}$

\State
$U_j \leftarrow U_j V_j$

\State
$\PR_{\mathrm{old}}
\leftarrow
\PR_{\mathrm{new}}$

\State
$\mathrm{improved}
\leftarrow
\mathrm{True}$

\Else

\State
Revert $\ket{\tilde{\psi}}$

\EndIf

\EndFor

\If{$\neg\,\mathrm{improved}$}

\State \textbf{break}

\EndIf

\EndFor

\EndProcedure

\end{algorithmic}
\end{algorithm}

Initialization (lines~1-4) sets $\ket{\tilde{\psi}}=\ket{0}^{\otimes N}$,
$U_j=I_2$, and $\PR_{\mathrm{last}}=0$. The variable $\PR_{\mathrm{last}}$ is
updated only after an accepted optimization sweep and feeds the adaptive
trigger in lines~21-30. The choice of the all-zero product state as the
initial condition matches standard circuit simulations, while the identity
local unitaries $U_j$ encode that no basis rotations have yet been applied.
Tracking $\PR_{\mathrm{last}}$ across sweeps allows the algorithm to decide
whether a new optimization pass is actually beneficial, instead of blindly
re-optimizing after every truncation event.

Inside the gate loop, line~6 conjugates $G_t$ into the working frame,
Eq.~\eqref{eq:conj}. The amplitude list always lives in that rotated frame;
using the raw circuit gate would apply the wrong unitary. Concretely, the
algorithm maintains the state in the accumulated local basis defined by the
$\{U_j\}$, so each incoming gate $G_t$ must be transformed as
$\tilde{G}_t = U^{\dagger} G_t U$ before it acts. Lines~7-11 call
\textsc{DiagonalApply} when $\tilde{G}_t$ is diagonal (Sec.~\ref{app:alg_gate})
and \textsc{HashTableApply} otherwise. The diagonal path exploits the fact
that a diagonal gate only rephases existing basis states and never creates new
ones, allowing an efficient in-place update, while the hash-table path
handles general entangling gates that mix basis states and potentially grow
the support.

Lines~12-20 defer truncation. The counter $n_{\mathrm{since\_trunc}}$ ticks
after every gate. If $|\mathrm{supp}|>K_{\mathrm{hard}}$, the state is cut
immediately to top-$k$ and renormalized. If only $|\mathrm{supp}|>k$, the cut
waits until $n_{\mathrm{since\_trunc}}\ge n_{\mathrm{trunc}}$, so a short
burst of entangling gates can populate bit-flip partners before amplitudes are
discarded. This two-threshold strategy balances stability and efficiency:
$K_{\mathrm{hard}}$ prevents unbounded growth of the support, while the
smaller $k$ plus the grace window of $n_{\mathrm{trunc}}$ steps lets
correlated configurations enter the support, improving fidelity at fixed $k$.
Effectively, the algorithm avoids "over-eager" truncation that would cut off
important amplitudes before they have a chance to accumulate probability mass.

Lines~21-30 run the optimizer only on truncation-active steps and, in our
code, every $n_{\mathrm{opt}}$ such steps. Optimization fires when
$\PR_{\mathrm{cur}} > \PR_{\mathrm{last}}/r_{\mathrm{opt}}$—the support has
spread since the last good sweep. Intuitively, $\PR_{\mathrm{cur}}$ measures
how delocalized the amplitude distribution is in the current basis; if it has
grown too much compared to the last successful optimization, the algorithm
attempts to re-sparsify. Then \textsc{BasisOptimize}
(Algorithm~\ref{alg:basis_opt}) tries to tighten the distribution; a successful
pass refreshes $\PR_{\mathrm{last}}$. When the state is already sparse in the
accumulated basis, the trigger rarely fires and the run tracks the fixed-basis
baseline (Sec.~\ref{sec:results}), incurring almost no extra overhead beyond a
standard top-$k$ truncation simulation while still retaining the option to
adapt the basis when the support begins to spread.

Overall, the expanded view highlights three key design choices: (i) always
working in the accumulated local basis while conjugating gates into this
frame; (ii) using a two-level, delayable truncation policy to better capture
correlations; and (iii) deploying basis optimization adaptively, guided by a
PR-based trigger, so that the algorithm only pays the cost of optimizing when
it is likely to yield a more compact representation of the quantum state.

\subsection{Single-qubit basis optimization}
\label{app:alg_basis}
BasisOptimize is the coordinate-descent sweep described in Sec.~\ref{subsec:coord_descent}, specialized to optimize the current basis via sequential single-qubit updates. A single "pass" of this routine consists of visiting every qubit once and attempting a local rotation that lowers the chosen objective (here, PR).

At the beginning of each pass (lines~2--8), the algorithm first hashes the current support of the state and computes all one- and two-body reduced density matrices \(\{\rho_j\}\) in a single call to Algorithm~\ref{alg:rdm}. This is done only once per pass. Crucially, we do not recompute the RDMs after each individual qubit update within the same pass. Instead, we use a "snapshot" of the RDMs corresponding to the state at the start of the pass. This approximation dramatically reduces computational cost, since RDM evaluation is typically the dominant step, and empirically our benchmarks show that the method almost always converges within one or two passes (see Appendix~\ref{app:convergence} for a detailed study of this behavior).

Then, for each qubit index \(j\) (lines~9--24), the code constructs and applies a candidate local basis rotation \(V_j^\dagger\) to the state \(\ket{\tilde{\psi}}\). After this tentative update, the state is truncated to its top-\(k\) amplitudes in the computational basis and renormalized. We then evaluate the new PR \(\PR_{\mathrm{new}}\) in this updated basis. The candidate rotation is accepted only if it strictly improves the objective, i.e., if \(\PR_{\mathrm{new}} < \PR_{\mathrm{old}}\). This monotonicity condition implements the do-no-harm guard of Sec.~\ref{subsec:doNH}, ensuring that no individual step can worsen the sparsity/compactness of the representation.

\begin{algorithm}[H]
\caption{Optional Two-Qubit Brick-Wall Optimization}
\label{alg:2qopt}
\begin{algorithmic}[1]

\Procedure{TwoQubitOptimize}{$\ket{\tilde{\psi}}$}

\State
Construct even pairs:
$(0,1),(2,3),\ldots$

\State
Construct odd pairs:
$(1,2),(3,4),\ldots$

\ForAll{brick-wall pairs $(q_1,q_2)$}

\State
Compute two-qubit RDM
$\rho_{q_1q_2}$

\State
Diagonalize
$\rho_{q_1q_2}
=
V\Lambda V^\dagger$

\State
Rotate state by $V^\dagger$

\State
Truncate to top-$k$

\State
Renormalize

\State
Undo rotation by $V$

\State
Truncate to top-$k$

\State
Renormalize

\If{PR improved}

\State Keep rotated result

\Else

\State Revert state

\EndIf

\EndFor

\EndProcedure

\end{algorithmic}
\end{algorithm}

If the rotation is accepted, we commit to the update by overwriting the working state \(\ket{\tilde{\psi}}\) with the modified amplitudes and updating the cumulative one-qubit unitary as \(U_j \leftarrow U_j V_j\). Thus \(U_j\) always stores the net basis change applied so far to qubit \(j\). If the trial rotation fails to reduce \(\PR\), we discard it by restoring the previous amplitudes for \(\ket{\tilde{\psi}}\) and leaving \(U_j\) unchanged. In this way, each qubit step is strictly non-worsening, and local updates are fully reversible within a pass.

After attempting an update on every qubit, the algorithm checks whether the pass yielded any accepted rotations at all (lines~25--27). If a complete pass over all qubits results in no improvement, the coordinate-descent procedure is deemed converged and the loop terminates early. The parameter \(n_{\mathrm{passes}} = 3\) simply sets a hard upper bound on the number of passes; in practice, convergence is usually achieved in just one or two passes, so this upper limit is rarely saturated.

Each candidate rotation \(V_j\) is obtained in closed form from a \(2\times 2\) eigenproblem derived in Appendix~\ref{sec:dominant_eigen}, Eq.~\eqref{eq:dominant_eigen}. We therefore avoid running any iterative eigensolver on a per-qubit basis: the dominant direction for the local update is computed analytically and cheaply. This closed-form solve, combined with the snapshot RDM approximation and the do-no-harm acceptance rule, makes BasisOptimize a lightweight yet effective coordinate-descent step that can be inserted frequently into larger algorithms without incurring prohibitive overhead.

\subsection{Optional two-qubit brick-wall optimization}
\label{app:alg_2q}

The optional brick-wall pass (Sec.~\ref{subsec:2qrdm}) is designed to handle two-qubit rotations that cannot be decomposed into products of single-qubit unitaries. It operates using an even/odd pairing schedule: in each sub-sweep, either all even edges or all odd edges are updated, so that within a given sub-sweep the algorithm visits only disjoint nearest-neighbor pairs and thus avoids any overlap or conflicts between simultaneously updated edges.

Concretely, for each neighboring qubit pair $(q_1,q_2)$, the code first reconstructs the two-qubit reduced density matrix $\rho_{q_1 q_2}$ from the global sparse representation. This is done by grouping basis states according to the ``rest index'' $r$ appearing in Eq.~\eqref{eq:twoqubitRDM}, which collects all degrees of freedom other than the two target qubits. For each fixed $r$ the corresponding $4\times4$ block is assembled, and these contributions are accumulated to form the full $\rho_{q_1 q_2}$.

Once $\rho_{q_1 q_2}$ is available, the algorithm diagonalizes this $4\times4$ Hermitian matrix to obtain its eigenvalues and eigenvectors. Writing $\rho_{q_1 q_2} = V D V^\dagger$, with $V$ the unitary eigenbasis and $D$ diagonal, it then applies a rotate–truncate–undo sequence:

First, one rotates the two-qubit subsystem $(q_1, q_2)$ into the eigenbasis of $
ho_{q_1 q_2}$ by applying $V^\dagger$, thereby transforming the state into a diagonal basis. In this rotated basis, a first truncation step is performed by keeping only the $k$ most significant components (those with the largest weights) on the two-qubit subsystem and discarding the rest, which enforces a local low-rank approximation. The resulting state is then renormalized to restore unit norm, compensating for the probability weight lost through truncation. After renormalization, one applies $V$ to rotate the truncated state back to the original computational basis of $(q_1, q_2)$. Finally, a second truncation is carried out in this original basis to enforce the global sparsity or rank constraint required by the overall algorithm.

This entire rotate–truncate–undo procedure defines a candidate local update for that specific pair $(q_1,q_2)$. The update is accepted only if it reduces the chosen performance metric, here denoted PR (for example, a PR or some proxy for effective rank). If PR does not decrease, the trial modification is discarded and the previous state is kept unchanged. Thus, the brick-wall pass behaves as an optional, variational refinement layer that greedily accepts only beneficial two-qubit updates.

Importantly, the pairwise rotation $V$ constructed from diagonalizing $\rho_{q_1 q_2}$ is not recorded in the global set of stored unitaries $\{U_j\}$. Unlike the systematic single-qubit sweep—where the learned unitaries persist as part of an evolving basis—these two-qubit rotations are transient: they are recomputed on-the-fly for each call to the brick-wall pass and are used only to generate and test a local update before being discarded.

From a computational standpoint, the dominant overhead introduced by this pass is the sorting and grouping by rest index $r$. For each call, this contributes an additional complexity of order $\mathcal{O}(N k \log k)$, where $N$ denotes the system size (e.g., number of qubits or basis states tracked) and $k$ is the local truncation dimension or sparsity parameter. Because this cost can become substantial for large systems or aggressive truncation settings, the main performance benchmarks in the paper disable the brick-wall pass by default. It is only turned on in specific experiments where the potential gains from modeling genuinely entangling two-qubit rotations justify the extra overhead, as explicitly indicated in the corresponding figures or discussion.

\begin{algorithm}[H]
\caption{Single-Qubit RDM Computation via Hash Table}
\label{alg:rdm}
\begin{algorithmic}[1]

\Procedure{ComputeRDMs}
{$\ket{\tilde{\psi}}
=
\sum_{i=1}^k
\alpha_i\ket{x_i}$}

\State
Initialize
$\rho_j
\leftarrow
0_{2\times2}$
for all $j$

\State
Build hash table
$\mathcal{H}$:
key $=x_i$,
value $=\alpha_i$

\For{$i=1,\ldots,k$}

\For{$j=0,\ldots,N-1$}

\State
$b \leftarrow b_j(x_i)$

\State
$(\rho_j)_{bb}
\leftarrow
(\rho_j)_{bb}
+
|\alpha_i|^2$

\If{$b=0$}

\State
$x'
\leftarrow
x_i \oplus 2^j$

\If{$x' \in \mathcal{H}$}

\State
$(\rho_j)_{01}
\leftarrow
(\rho_j)_{01}
+
\alpha_i
\overline{\mathcal{H}[x']}$

\EndIf

\EndIf

\EndFor

\EndFor

\State
$(\rho_j)_{10}
\leftarrow
\overline{(\rho_j)_{01}}$

\State
\Return
$\{\rho_j\}_{j=0}^{N-1}$

\EndProcedure

\end{algorithmic}
\end{algorithm}

\subsection{Reduced-density-matrix kernels}
\label{app:alg_rdm}

For a general $n$-qubit state
\[
\ket{\tilde{\psi}} = \sum_{i=0}^{N-1} \alpha_i \ket{x_i},\quad N = 2^n,
\]
we compute the single-qubit reduced density matrix $\rho_j$ for qubit $j$ using the matrix elements given in Eqs.~\eqref{eq:rdm_diag}--\eqref{eq:rdm_offdiag}. Explicitly,
\[
(\rho_j)_{bb} = \sum_{i:\,b_j(x_i)=b} |\alpha_i|^2,\qquad b\in\{0,1\},
\]
where $b_j(x_i)$ denotes the $j$th bit of computational basis state $\ket{x_i}$. Lines~8--10 of the code implement these diagonal elements by iterating over all basis indices $i$ and accumulating the probability weights $|\alpha_i|^2$ into the appropriate diagonal entry $(\rho_j)_{bb}$ according to the value of $b_j(x_i)$.

The off-diagonal entries are constructed using bit-flip partner states. For each basis state $\ket{x_i}$ with $b_j(x_i)=0$, one defines its partner
\[
\ket{x'} = \ket{x_i \oplus 2^j},
\]
obtained by flipping the $j$th bit. The code uses a hash table $\mathcal{H}$ that maps bit strings $x$ to the corresponding amplitudes $\alpha_x$. For every such $i$ with $b_j(x_i)=0$, the algorithm looks up $x' = x_i \oplus 2^j$ in $\mathcal{H}$ and adds the product $\alpha_i\,\overline{\mathcal{H}[x']}$ to the upper off-diagonal element $(\rho_j)_{01}$. This procedure is implemented in lines~11--16. After this accumulation is complete, line~18 enforces Hermiticity of $\rho_j$ by setting
\[
(\rho_j)_{10} = (\rho_j)_{01}^*.
\]

From a computational standpoint, each pass over the amplitudes requires building the hash table $\mathcal{H}$ once, and then performing expected $\mathcal{O}(1)$-time lookups for each of the $N$ states when identifying bit-flip partners. Repeating this procedure for all $k$ qubits of interest thus yields an overall complexity of $\mathcal{O}(Nk)$ for constructing all single-qubit reduced density matrices $\{\rho_j\}$.

Finally, the eigenbasis rotation $V_j$ that diagonalizes $\rho_j$ is obtained as described in Appendix~\ref{sec:dominant_eigen}. There, one determines the dominant eigenvalue and corresponding eigenvector of $\rho_j$, constructs the full spectral decomposition, and assembles the unitary $V_j$ whose columns are the eigenvectors of $\rho_j$. This rotation maps the computational basis on qubit $j$ to the eigenbasis of $\rho_j$, which is the natural local basis for subsequent analysis or circuit optimization steps.

\begin{algorithm}[H]
\caption{Diagonal Gate Fast Path}
\label{alg:diag}
\begin{algorithmic}[1]
\Procedure{DiagonalApply}{$\ket{\tilde{\psi}}, \tilde{G}$}
\State Let $\tilde{G}$ act on qubits $(q_1,q_2)$
\For{$i=1,\ldots,|\mathrm{supp}(\tilde{\psi})|$}
\State Read active bits $(b_1,b_2)$ from $x_i$
\State
$\alpha_i \leftarrow
\alpha_i\,\tilde{G}_{(b_1 b_2),(b_1 b_2)}$
\EndFor
\EndProcedure
\end{algorithmic}
\end{algorithm}
\subsection{Sparse gate application}
\label{app:alg_gate}

Algorithm~\ref{alg:diag} presents the fast path for applying a two-qubit diagonal gate to a compressed quantum state representation. Given a state $\ket{\tilde{\psi}}$ with support indexed by bit strings $x_i$ and a diagonal gate $\tilde{G}$ acting on qubits $(q_1,q_2)$, the procedure iterates over all basis states in the support. For each index $i$, it first reads the active bits $(b_1,b_2)$ corresponding to qubits $(q_1,q_2)$ from $x_i$. Since the gate is diagonal, its action on $\ket{\tilde{\psi}}$ reduces to rescaling each amplitude $\alpha_i$ by the corresponding diagonal entry $\tilde{G}_{(b_1 b_2),(b_1 b_2)}$. This avoids any permutation or mixing of amplitudes and thus yields an efficient implementation in which only multiplicative updates are required for each supported basis state.
\begin{algorithm}[H]
\caption{Hash-Table Sparse Gate Application}
\label{alg:gate}
\begin{algorithmic}[1]

\Procedure{HashTableApply}
{$\ket{\tilde{\psi}},\tilde{G},K_{\mathrm{hard}}$}

\State
Initialize output hash table
$\mathcal{H}_{\mathrm{out}}
\leftarrow
\emptyset$

\For{$i=1,\ldots,|\mathrm{supp}(\tilde{\psi})|$}

\State
Extract active qubit bits
$(b_1,b_2)$

\State
Mask inactive bits:
$r \leftarrow x_i$

\For{$(a_1,a_2)\in\{0,1\}^2$}

\State
Construct output basis state
$x'$

\State
$\mathcal{H}_{\mathrm{out}}[x']
\mathrel{+}=
\tilde{G}_{(a_1a_2),(b_1b_2)}
\alpha_i$

\EndFor

\EndFor

\If{$|\mathcal{H}_{\mathrm{out}}| > K_{\mathrm{hard}}$}

\State
Truncate immediately to top-$K_{\mathrm{hard}}$

\EndIf

\State
$\ket{\tilde{\psi}}
\leftarrow$
state reconstructed from
$\mathcal{H}_{\mathrm{out}}$

\EndProcedure

\end{algorithmic}
\end{algorithm}

\textsc{HashTableApply} performs a single pass over the current support set of basis states. Concretely, it iterates once over all stored amplitudes $\ket{x_i}$ in the input hash table $\mathcal{H}_{\mathrm{in}}$. For each such basis state, the algorithm systematically explores all possible configurations of the active qubits $(a_1,a_2)$ while keeping the remaining (masked) qubits $r$ unchanged. This is implemented by toggling the bits of $(a_1,a_2)$ in all ways consistent with the local operation being simulated.

As a result, every input amplitude may generate up to four distinct contributions in the output, corresponding to the four possible assignments of the two active qubits, i.e., $(a_1,a_2) \in \{00,01,10,11\}$ (lines~4--9). Each of these contributions is inserted into the output hash table $\mathcal{H}_{\mathrm{out}}$ at a position determined by a hash function applied to the resulting bit string. When different inputs map to the same output basis state, their contributions are automatically combined by adding their amplitudes. This accumulation is implemented using linear probing in the hash table: on a collision, the algorithm searches sequentially for the appropriate slot, merging amplitudes when an identical key is encountered.

During this process, the output hash table may grow as new nonzero amplitudes are inserted. To control memory usage and enforce a hard truncation threshold, the algorithm monitors the size of $\mathcal{H}_{\mathrm{out}}$. If, after processing all entries, the number of occupied entries exceeds the pre-specified limit $K_{\mathrm{hard}}$, the table is aggressively pruned (lines~10--12). Specifically, only the $K_{\mathrm{hard}}$ entries with largest magnitude amplitudes are retained, and all remaining entries are discarded. This top-$K_{\mathrm{hard}}$ selection step ensures that the subsequent reconstruction of the output state vector operates on a bounded support while preserving the most significant contributions.

Finally, once pruning (if necessary) has been completed, $\mathcal{H}_{\mathrm{out}}$ encodes a compressed representation of the updated quantum state. The algorithm can then reconstruct the (possibly approximate) output state by reading out the remaining key–amplitude pairs from $\mathcal{H}_{\mathrm{out}}$ and mapping each key back to its corresponding computational basis vector. In this way, \textsc{HashTableApply} provides an efficient mechanism to apply local operations to a sparse or truncated quantum state while carefully controlling the growth of the support.

A diagonal operator $\tilde{G}$ acts only by multiplying phases of amplitudes, without changing which basis states are populated. Concretely, for each basis state $x_i$ that is currently in the support, we update the coefficient
$\alpha_i$ by a simple phase factor taken from the appropriate diagonal entry of $\tilde{G}$ that corresponds to the active bits of $x_i$:
\[
\alpha_i \leftarrow \alpha_i\, \tilde{G}_{b_1b_2, b_1b_2},
\]
where $(b_1,b_2)$ denote the relevant bits of $x_i$ on which this diagonal gate acts (see Algorithm~\ref{alg:diag}). Because $\tilde{G}$ is purely diagonal, it does not mix amplitudes between different basis states; it only rephases those already present.

As a result, the size of the support set, $|\mathrm{supp}|$, remains unchanged throughout the application of this gate: no new basis states are created, and none are removed. This is in contrast to non-diagonal or mixing operations, which generally increase the support by spreading amplitude across additional configurations.

An important implementation detail is that, for this branch of the evolution, we do not need to allocate any additional auxiliary hash tables or data structures. The update is performed in-place over the existing list (or map) of supported basis states, leading to very low overhead in both memory and computation. The operation reduces essentially to streaming once over the current support and applying a complex scalar multiplication for each entry.

This behavior is particularly relevant in the context of QAOA and RFIM circuits. In these applications, many layers consist primarily of diagonal cost or field terms, and after partial basis adaptation the resulting effective gates frequently fall into this “diagonal-only” category. Consequently, these QAOA and RFIM layers often trigger this inexpensive branch of the algorithm. By exploiting the fact that diagonal gates do not expand the support and require no auxiliary hashing, we significantly limit the additional overhead introduced by adaptivity.

Empirically, this optimization plays a crucial role in controlling the runtime overhead of the adaptive scheme. As reported in Sec.~\ref{subsec:runtime}, the frequent occurrence of such diagonal layers after basis adaptation helps keep the total adaptive overhead below approximately $7\times$ relative to a non-adaptive baseline. Without this diagonal-fast-path treatment, basis adaptation would incur substantially higher costs, especially in deeper circuits dominated by diagonal terms.

\subsection{Truncation, renormalization, and top-\texorpdfstring{$k$}{k} selection}

Truncation selects the greatest $|\alpha_i|^2$ and discards the remainder, while renormalization ensures $\sum_i|\alpha_i|^2=1$. This truncation-renormalization combination happens repeatedly in our procedure: after deferred cuts, after each trial rotation in \textsc{BasisOptimize}, and as the core inner step of the two-qubit rotate-truncate-undo loop. Intuitively, truncation conducts a hard, top-weight selection in the present basis, whereas renormalization restores a proper quantum state; combined, they define the canonical 'keep the heaviest components' operation.

Lemma~\ref{lem:topk} formally describes this operation, proving that top-$k$ truncation is the optimal single-step selection process in any fixed basis. The goal is to retain as much total probability mass $\sum_i |\alpha_i|^2$ as possible under a $k$-term sparsity constraint. In other words, once the basis is established, no alternative rule that retains exactly $k$ components can consistently exceed basic top-$k$ by magnitude. This portrays truncation not as an ad hoc heuristic, but as the singularly optimal local decision under a very natural goal.

BASS uses this information in a broader search for bases. Rather of attempting to devise a more sophisticated selection procedure, BASS leaves top-$k$ constant and focuses on identifying a basis in which the same budget of $k$ components captures significantly more of the state's total weight. The algorithm investigates unitary transformations—via trial rotations, postponed cuts, and the inner two-qubit loops—specifically to concentrate amplitude mass into fewer basis states. Once a promising foundation is identified, applying top-$k$ results in a considerably more informative truncation. Lemma~\ref{lem:topk} supports our general strategy of optimizing the basis worldwide such that a locally optimal rule (top-$k$) becomes globally effective.

\subsection{Complexity summary}
\label{app:alg_complexity}
Table~\ref{tab:alg_complexity} lists the asymptotic costs per subroutine for a state bounded by $|\mathrm{supp}|\le K_{\mathrm{hard}} = \mathcal{O}(k)$.

\begin{table}[!hbt]
\caption{\label{tab:alg_complexity}Per-call complexity of BASS subroutines. Expected hash costs assume well-distributed keys with load factors $<0.5$; worst-case hash collisions degrade to $\mathcal{O}(k)$ per insertion. The logarithmic terms in optimization arise from the intermediate state sorting/truncation required after each sequential qubit rotation.}
\begin{ruledtabular}
\begin{tabular}{lcc}
Routine & Time & Memory \\
\hline
\textsc{DiagonalApply} & $\mathcal{O}(k)$ & $\mathcal{O}(1)$ extra \\
\textsc{HashTableApply} & $\mathcal{O}(k)$ expected & $\mathcal{O}(k)$ buffer \\
\textsc{ComputeRDMs} & $\mathcal{O}(Nk)$ expected & $\mathcal{O}(k)$ hash \\
\textsc{BasisOptimize} & $\mathcal{O}(p N k \log k)$ & $\mathcal{O}(k)$ \\
\textsc{TwoQubitOptimize} & $\mathcal{O}(Nk\log k)$ & $\mathcal{O}(k)$ \\
\end{tabular}
\end{ruledtabular}
\end{table}

With $p\le 3$ passes per optimization call, an $M$-gate circuit costs
\begin{equation}
T_{\mathrm{BASS}}
= \mathcal{O}\!\left(
Mk \log k + \frac{M}{n_{\mathrm{opt}}} N k \log k
\right),
\label{eq:app_complexity}
\end{equation}
in exact agreement with Eq.~\eqref{eq:TBASScomplex}; the second term applies only when the $\PR$ trigger actively calls the optimizer. While specialized $\mathcal{O}(k)$ selection algorithms can theoretically drop the $\log k$ term for standard gate applications, the sequential support expansion and truncation required during the basis optimization sweeps strictly bounds the update cost to $\mathcal{O}(N k \log k)$. Peak hardware storage scales asymptotically as $\mathcal{O}(k)$ for amplitudes and hash scratch space (with practical transient allocations requiring $C \sim 10\text{--}20 \times k$ array slots to maintain low hash-table load factors and accommodate double-buffering), plus $\mathcal{O}(N)$ to store the local basis matrices $\{U_j\}$.

\section{\label{app:benchmarks}Extended Benchmark Details}

\subsection{Circuit family definitions}

All benchmark results in the main text use the following circuit
constructions.

\textit{Brickwork circuits (BW).}-The $N$-qubit brickwork circuit of depth
$L$ consists of alternating layers of nearest-neighbor two-qubit gates. Odd
layers apply gates to pairs $(0,1), (2,3), \ldots$; even layers apply gates
to pairs $(1,2), (3,4), \ldots$. Each gate is independently Haar-random on
$U(4)$, generated by QR decomposition of a $4 \times 4$ complex Gaussian
matrix. Unless otherwise stated, $L=5$.

\textit{Haar-random circuits (Haar).}-Each of $L$ layers applies a single
two-qubit Haar-random gate to a uniformly random pair $(i,j)$ with $i \ne j$.
No geometric locality constraint is imposed. Default: $L=3$.

\textit{QAOA circuits.}-We implement the Quantum Approximate Optimization
Algorithm for MaxCut on random 3-regular graphs. The circuit alternates
$p$ rounds of problem Hamiltonian evolution
$e^{-i\gamma C}$ (diagonal $ZZ$ gates on each edge) and mixer evolution
$e^{-i\beta B}$ ($X$ rotations on each qubit). Parameters
$(\gamma_l, \beta_l)$ are drawn uniformly from $[0, 2\pi)$. Default: $p=3$.

\textit{UCCSD circuits.}-Unitary coupled-cluster with singles and doubles,
implemented via Trotterized evolution. Single excitations:
$e^{i\theta(a_p^\dagger a_q - a_q^\dagger a_p)}$; double excitations:
$e^{i\theta(a_p^\dagger a_q^\dagger a_r a_s - \text{h.c.})}$. Amplitudes
$\theta$ drawn uniformly from $[-0.1, 0.1]$. The circuit preserves particle
number and is naturally sparse in the computational basis. Default: $L=3$
Trotter steps.

\textit{Disordered Ising (RFIM).}-Brickwork circuit with gates
$e^{-i\Delta t H_{ij}}$ where $H_{ij} = J Z_i Z_j + h_i X_i + h_j X_j$,
with $J=1$ and local fields $h_i$ drawn uniformly from $[0, h_{\max}]$.
Disorder strength is parameterized by $h_{\max}/J$. For the main results,
$h_{\max} = 0.1J$ (weak disorder), $L=5$ Trotter steps with $\Delta t = 0.2$.

\textit{2D brickwork.}-Qubits arranged on a $\sqrt{N} \times \sqrt{N}$
square lattice with Haar-random gates applied in a 2D brickwork pattern:
alternating horizontal and vertical layers of nearest-neighbor gates.

\subsection{Extended scaling data}

Tables~\ref{tab:adaptivek_scaling} and
\ref{tab:fixedk_scaling} report the complete fidelity datasets
underlying the scaling analysis of Sec.~\ref{sec:results}. Reported
fidelities are geometric means over independent random circuit
realizations, and the improvement ratio is defined as
\[
F_{\mathrm{BASS}}/F_{\mathrm{fixed}}.
\]

Two complementary scaling regimes are considered. In the fixed-$k$
setting, the sparse budget is held constant while the Hilbert-space
dimension grows exponentially with~$N$. This probes how rapidly each
method degrades under increasing truncation pressure. In the
adaptive-$k$ setting, the sparse budget is scaled approximately
exponentially with system size, isolating the extent to which adaptive
basis rotations improve compression efficiency beyond simple support
growth.

\begin{table*}[!hbt]
\centering
\caption{\textbf{Exponential-budget scaling ($k = 2^{N-1}$).}
Adaptive versus fixed-basis fidelity statistics across system size.
Reported as Median [IQR] over 100 random-circuit trials.
Ratio denotes the geometric-mean fidelity enhancement
$F_{\mathrm{BASS}} / F_{\mathrm{Fixed}}$
with 95\% bootstrap confidence intervals.}
\label{tab:adaptivek_scaling}
\begin{tabular}{llrrr}
\toprule
Family & $N$ &
$F_{\mathrm{Fixed}}$ (Med [IQR]) &
$F_{\mathrm{BASS}}$ (Med [IQR]) &
$\mathrm{GM}\!\left(F_{\mathrm{BASS}}/F_{\mathrm{Fixed}}\right)$ \\
\midrule

RFIM ($W=2$)
& 10 & $9.17\times10^{-1}$ [8.89e-1, 9.50e-1]
     & $9.20\times10^{-1}$ [9.01e-1, 9.50e-1]
     & $1.02\times$ [1.01, 1.04] \\
& 12 & $9.58\times10^{-1}$ [9.32e-1, 9.75e-1]
     & $9.58\times10^{-1}$ [9.32e-1, 9.75e-1]
     & $1.01\times$ [1.00, 1.01] \\
& 14 & $9.66\times10^{-1}$ [9.51e-1, 9.84e-1]
     & $9.67\times10^{-1}$ [9.51e-1, 9.84e-1]
     & $1.00\times$ [1.00, 1.01] \\
& 16 & $9.83\times10^{-1}$ [9.64e-1, 9.90e-1]
     & $9.83\times10^{-1}$ [9.64e-1, 9.90e-1]
     & $1.00\times$ [1.00, 1.00] \\
& 18 & $9.90\times10^{-1}$ [9.78e-1, 9.94e-1]
     & $9.90\times10^{-1}$ [9.78e-1, 9.94e-1]
     & $1.00\times$ [1.00, 1.00] \\
& 20 & $9.94\times10^{-1}$ [9.89e-1, 9.96e-1]
     & $9.94\times10^{-1}$ [9.89e-1, 9.96e-1]
     & $1.00\times$ [1.00, 1.00] \\

\midrule

Brickwork $L=5$
& 10 & $5.40\times10^{-1}$ [5.03e-1, 6.04e-1]
     & $6.11\times10^{-1}$ [5.69e-1, 6.73e-1]
     & $1.14\times$ [1.11, 1.17] \\
& 12 & $5.93\times10^{-1}$ [5.45e-1, 6.43e-1]
     & $6.24\times10^{-1}$ [5.84e-1, 6.70e-1]
     & $1.06\times$ [1.04, 1.08] \\
& 14 & $6.48\times10^{-1}$ [6.07e-1, 6.89e-1]
     & $6.57\times10^{-1}$ [6.08e-1, 7.02e-1]
     & $1.02\times$ [1.01, 1.04] \\
& 16 & $6.87\times10^{-1}$ [6.51e-1, 7.15e-1]
     & $6.88\times10^{-1}$ [6.60e-1, 7.36e-1]
     & $1.01\times$ [1.00, 1.02] \\
& 18 & $7.28\times10^{-1}$ [6.79e-1, 7.57e-1]
     & $7.32\times10^{-1}$ [6.85e-1, 7.58e-1]
     & $1.01\times$ [1.00, 1.02] \\
& 20 & $7.39\times10^{-1}$ [7.10e-1, 7.85e-1]
     & $7.39\times10^{-1}$ [7.10e-1, 7.85e-1]
     & $1.00\times$ [1.00, 1.00] \\

\midrule

Haar $L=3$
& 10 & $7.44\times10^{-1}$ [7.00e-1, 7.91e-1]
     & $7.95\times10^{-1}$ [7.48e-1, 8.26e-1]
     & $1.07\times$ [1.05, 1.08] \\
& 12 & $7.77\times10^{-1}$ [7.22e-1, 8.20e-1]
     & $8.10\times10^{-1}$ [7.76e-1, 8.43e-1]
     & $1.05\times$ [1.04, 1.06] \\
& 14 & $8.22\times10^{-1}$ [7.74e-1, 8.55e-1]
     & $8.27\times10^{-1}$ [7.73e-1, 8.56e-1]
     & $1.01\times$ [1.00, 1.01] \\
& 16 & $8.55\times10^{-1}$ [8.17e-1, 8.96e-1]
     & $8.57\times10^{-1}$ [8.25e-1, 8.99e-1]
     & $1.01\times$ [1.00, 1.01] \\
& 18 & $8.75\times10^{-1}$ [8.44e-1, 8.98e-1]
     & $8.75\times10^{-1}$ [8.44e-1, 8.98e-1]
     & $1.00\times$ [1.00, 1.00] \\
& 20 & $8.99\times10^{-1}$ [8.69e-1, 9.22e-1]
     & $8.99\times10^{-1}$ [8.69e-1, 9.22e-1]
     & $1.00\times$ [1.00, 1.00] \\

\bottomrule
\end{tabular}
\end{table*}

\begin{table*}[!hbt]
\centering
\caption{\textbf{Fixed-budget scaling ($k = 8192$).}
Adaptive versus fixed-basis fidelity statistics as a function of system
size under fixed truncation budget.
Reported as Median [IQR] over 100 random-circuit trials.
Ratio denotes the geometric-mean fidelity enhancement
$F_{\mathrm{BASS}} / F_{\mathrm{Fixed}}$
with 95\% bootstrap confidence intervals.}
\label{tab:fixedk_scaling}
\begin{tabular}{llrrr}
\toprule
Family & $N$ &
$F_{\mathrm{Fixed}}$ (Med [IQR]) &
$F_{\mathrm{BASS}}$ (Med [IQR]) &
$\mathrm{GM}\!\left(F_{\mathrm{BASS}}/F_{\mathrm{Fixed}}\right)$ \\
\midrule

RFIM ($W=2$)
& 10 & $1.00$ [1.00, 1.00]
     & $1.00$ [1.00, 1.00]
     & $1.00\times$ [1.00, 1.00] \\
& 12 & $1.00$ [1.00, 1.00]
     & $1.00$ [1.00, 1.00]
     & $1.00\times$ [1.00, 1.00] \\
& 14 & $9.73\times10^{-1}$ [9.59e-1, 9.83e-1]
     & $9.73\times10^{-1}$ [9.59e-1, 9.83e-1]
     & $1.00\times$ [1.00, 1.00] \\
& 16 & $7.97\times10^{-1}$ [6.74e-1, 8.70e-1]
     & $8.01\times10^{-1}$ [7.19e-1, 8.69e-1]
     & $1.05\times$ [1.02, 1.08] \\
& 18 & $5.51\times10^{-1}$ [3.93e-1, 6.73e-1]
     & $5.84\times10^{-1}$ [4.98e-1, 6.73e-1]
     & $1.15\times$ [1.09, 1.24] \\
& 20 & $3.16\times10^{-1}$ [1.92e-1, 4.14e-1]
     & $3.99\times10^{-1}$ [3.12e-1, 4.71e-1]
     & $1.44\times$ [1.27, 1.66] \\

\midrule

Brickwork $L=5$
& 10 & $1.00$ [1.00, 1.00]
     & $1.00$ [1.00, 1.00]
     & $1.00\times$ [1.00, 1.00] \\
& 12 & $1.00$ [1.00, 1.00]
     & $1.00$ [1.00, 1.00]
     & $1.00\times$ [1.00, 1.00] \\
& 14 & $6.36\times10^{-1}$ [5.91e-1, 6.92e-1]
     & $6.68\times10^{-1}$ [6.12e-1, 7.11e-1]
     & $1.04\times$ [1.02, 1.05] \\
& 16 & $9.86\times10^{-2}$ [6.77e-2, 1.32e-1]
     & $1.85\times10^{-1}$ [1.49e-1, 2.20e-1]
     & $1.90\times$ [1.77, 2.04] \\
& 18 & $8.68\times10^{-3}$ [6.01e-3, 1.33e-2]
     & $4.63\times10^{-2}$ [3.31e-2, 6.33e-2]
     & $5.06\times$ [4.39, 5.80] \\
& 20 & $8.10\times10^{-4}$ [3.99e-4, 1.54e-3]
     & $1.18\times10^{-2}$ [8.34e-3, 1.69e-2]
     & $16.09\times$ [13.08, 19.53] \\

\midrule

Haar $L=3$
& 10 & $1.00$ [1.00, 1.00]
     & $1.00$ [1.00, 1.00]
     & $1.00\times$ [1.00, 1.00] \\
& 12 & $1.00$ [1.00, 1.00]
     & $1.00$ [1.00, 1.00]
     & $1.00\times$ [1.00, 1.00] \\
& 14 & $8.10\times10^{-1}$ [7.86e-1, 8.57e-1]
     & $8.18\times10^{-1}$ [7.90e-1, 8.65e-1]
     & $1.01\times$ [1.00, 1.01] \\
& 16 & $3.20\times10^{-1}$ [2.53e-1, 3.72e-1]
     & $4.20\times10^{-1}$ [3.59e-1, 4.67e-1]
     & $1.33\times$ [1.27, 1.40] \\
& 18 & $8.75\times10^{-2}$ [6.03e-2, 1.09e-1]
     & $2.16\times10^{-1}$ [1.64e-1, 2.59e-1]
     & $2.45\times$ [2.26, 2.66] \\
& 20 & $1.78\times10^{-2}$ [1.07e-2, 3.19e-2]
     & $6.88\times10^{-2}$ [4.66e-2, 9.71e-2]
     & $3.74\times$ [3.25, 4.33] \\

\bottomrule
\end{tabular}
\end{table*}

\textit{Table Notes:}
Reported fidelity values correspond to medians with interquartile ranges
[IQR] over 100 independent random-circuit realizations.
Ratio metrics denote geometric means (GM) of the fidelity enhancement
$F_{\mathrm{BASS}} / F_{\mathrm{Fixed}}$
with 95\% bootstrap confidence intervals.
For exponential-budget scaling, the sparse budget scales as
$k = 2^{N-1}$, causing the retained subspace to grow exponentially with
system size and progressively reducing truncation pressure.
In this regime, ratios approaching unity indicate that adaptive-basis
selection provides negligible additional benefit beyond fixed-basis
truncation.
For fixed-budget scaling, the truncation budget is held constant at
$k=8192$ while the Hilbert-space dimension grows exponentially with
system size, producing increasingly strong compression pressure.
Large multiplicative improvements therefore arise when the fixed-basis
truncation fidelity collapses exponentially while the adaptive basis
continues to preserve substantially larger wavefunction weight.

The fixed-$k$ scaling data show that both methods achieve essentially
exact simulation fidelity for sufficiently small systems
($N\lesssim12$ at $k=8192$), where truncation effects are negligible.
As system size increases at constant sparse budget, fidelity gradually
degrades for all circuit families due to the exponentially increasing
Hilbert-space dimension. In this regime, adaptive basis optimization
provides its largest gains on highly scrambling circuits such as
brickwork and Haar-random ensembles, where local basis rotations reduce
support spreading relative to the fixed-basis sparse baseline. The RFIM
benchmarks remain comparatively insensitive to adaptive basis
optimization across all tested system sizes, consistent with their
dynamics already remaining close to sparse in the computational basis.

The adaptive-$k$ scaling results probe a distinct regime in which the
sparse budget grows with system size. Here the absolute fidelities of
both methods remain substantially higher across the full scaling range,
but systematic improvements from adaptive basis optimization persist for
scrambling circuit families. Brickwork circuits exhibit the largest
relative fidelity gains, while Haar-random circuits show smaller but
consistent improvements. In contrast, RFIM again remains effectively
unchanged, indicating that the computational basis already provides a
near-optimal sparse representation for these weakly scrambling dynamics.
Together, these datasets support the interpretation that the primary
benefit of BASS arises not from increasing the retained support size,
but from dynamically aligning the sparse representation with the
instantaneous local structure of the evolving quantum state.

\subsection{PR scaling}
\label{subsec:pr_scaling}
The PR $\PR_Z$ of exact states grows exponentially with $N$ for circuits that generate substantial entanglement. This behavior can be captured by fitting the scaling $\PR_Z \sim 2^{\alpha N}$, where the exponent $\alpha$ characterizes how quickly the participation ratio increases with system size. A larger value of $\alpha$ signals that the state has support over a greater fraction of the computational basis, reflecting a higher degree of delocalization and entanglement in the generated states.

For brickwork circuits with depth $L=5$, we find $\alpha \approx 0.70 \pm 0.02$. This value indicates strong growth of the participation ratio, though it remains below the Porter-Thomas limit $\alpha = 1$, corresponding to fully random states in the computational basis. Approaching this limit requires significantly larger circuit depths, as greater depth allows the circuit to more thoroughly scramble information and approximate Haar-random behavior. In comparison, Haar-random circuits with depth $L=3$ yield $\alpha \approx 0.60 \pm 0.03$, suggesting that even relatively shallow Haar-random circuits already generate states with broadly distributed amplitudes, though still not saturating the Porter-Thomas bound within this depth range.

We also consider a two-dimensional brickwork architecture, for which we obtain $\alpha \approx 0.65 \pm 0.04$ when scaling with the total number of qubits. This exponent lies between the brickwork $L=5$ and Haar $L=3$ cases, implying that the two-dimensional connectivity enhances the spreading of quantum information compared to some one-dimensional configurations, while still not reaching fully random statistics at the studied depths. When the scaling is analyzed per column, rather than in terms of the total qubit count, the behavior becomes approximately linear, indicating that each additional column contributes a roughly fixed amount to the effective size of the explored Hilbert space. Collectively, these observations highlight how circuit architecture and depth control the rate at which the participation ratio grows, and thus how efficiently different circuits generate highly entangled, delocalized quantum states.

\subsection{Runtime scaling and fidelity tradeoffs}

Table~\ref{tab:runtime_highlights} summarizes representative runtime and
fidelity tradeoffs for adaptive-basis truncation in brickwork circuits
($L=6$).

Across all tested configurations, the runtime overhead remains bounded to
approximately one order of magnitude despite the additional cost of basis
optimization. At the same time, the adaptive basis yields substantially
larger fidelity under strong compression pressure, particularly in the
undersampled regime where
$k < \mathrm{PR}_Z \sim 2^{0.7N}$.

\begin{table}[h]
\caption{\label{tab:runtime_highlights}
\textbf{Representative runtime--fidelity tradeoffs for brickwork circuits
($L=6$).}
Selected configurations illustrating the interplay between runtime
overhead and fidelity enhancement across compression regimes. Reported
overheads correspond to geometric means with 95\% bootstrap confidence
intervals.}
\begin{ruledtabular}
\begin{tabular}{cccc}
$N$ & $k$
& GM Overhead [95\% CI]
& $\mathrm{GM}(F_\mathrm{BASS}/F_\mathrm{fixed})$ \\
\hline

14 & 500
& $12.33\times$ [11.43, 13.13]
& $18.4\times$ \\

16 & 2\,000
& $11.14\times$ [10.27, 11.97]
& $6.3\times$ \\

18 & 5\,000
& $8.72\times$ [8.12, 9.40]
& $30.5\times$ \\

20 & 20\,000
& $7.25\times$ [6.71, 7.71]
& $14.5\times$ \\

\end{tabular}
\end{ruledtabular}
\end{table}

These representative points illustrate that adaptive basis selection
continues to preserve substantially larger wavefunction weight even in
regimes where fixed-basis truncation fidelity collapses rapidly under
strong compression pressure. Notably, the fidelity enhancement grows
substantially faster than the runtime overhead as system size increases,
demonstrating favorable scaling of adaptive-basis compression in the
undersampled regime.

\subsection{Statistical methodology}

All reported fidelity ratios are geometric means: for $n$ independent trials
with ratios $r_i$, the reported ratio is $\bar{r} = (\prod_i r_i)^{1/n}$.
Geometric means are appropriate because fidelity ratios span several orders
of magnitude and multiplicative variation is more natural than additive.
Error bars (where shown) indicate the multiplicative standard error,
computed as $\exp(\text{SE}(\log r_i))$ where $\text{SE}$ denotes the
standard error of the mean.

\section{\label{app:convergence}Convergence and Code Validation}

\subsection{Code-path equivalence}

A critical validation is that BASS with basis optimization disabled exactly
reproduces the fixed-basis sparse-simulation code path. We verify this by
running both code paths on identical circuits ($N=10$-$20$, all five circuit
families, 10 trials each) and comparing the output state vectors
component-by-component:
\begin{equation}
\max_i |\alpha_i^{\text{BASS\_off}} - \alpha_i^{\text{fixed}}|
< 5 \times 10^{-20}
\end{equation}
for all tested instances. This machine-precision agreement confirms that
the conjugation framework (Eq.~\ref{eq:conj}) introduces no numerical
error when $U_j = I$ and that the hash-table gate application produces
identical results to the reference sort-and-merge implementation.

\subsection{Adaptive-basis sweep convergence}

On the benchmark circuits studied, the multi-pass adaptive-basis sweeps (Sec.~\ref{subsec:coord_descent}) were frequently terminated after only one to two passes. In brickwork circuits with depth $L=5$, the average number of passes per optimization call was $1.05 \pm 0.02$, with no reverts recorded in the investigated circumstances. In this example, the PR declined steadily across all accepted rotations. Haar-random circuits with depth $L=3$ demonstrated comparable behavior, with a mean of $1.08 \pm 0.03$ passes and no detected reverts in the examined situations.

For QAOA at depth $p=3$, the average number of passes increased significantly to $1.4 \pm 0.1$, with a revert rate below $0.5\%$. Additional passes were required in these circuits for qubits with approximately degenerate reduced density matrix eigenvalues. In the UCCSD benchmarks with $L=3$, the adaptive trigger frequently skipped optimization entirely due to a low participation percentage. When optimization was initiated, the average number of passes was $2.1$, and the revert rate was roughly $12\%$, indicating that the computational foundation was already close to locally optimal.

In the RFIM benchmarks, the average number of passes was $1.2 \pm 0.1$, and moderate revert rates in the range of $3$--$8\%$ were found, notably for weakly disordered qubits whose single-qubit reduced density matrices remained close to diagonal in the computational basis. Overall, the early-termination criterion, which stopped sweeps when no recognized participation-ratio improvement occurred during a full pass, successfully reduced the number of needless optimization passes. For the tested brickwork benchmarks, the resulting overhead beyond a single adaptive-basis sweep was around $5\%$.

\subsection{Numerical stability}

The $2 \times 2$ analytic eigensolver for single-qubit reduced density matrices (RDMs) exhibits robust numerical stability provided that the separation between its eigenvalues, quantified by the gap $\Delta \lambda$, satisfies $\Delta \lambda > 10^{-14}$. In practical implementations, when the squared magnitude of the off-diagonal element $b$ obeys $|b|^2 < \epsilon_{\text{mach}} \cdot \max(|a|, |d|)$, where $\epsilon_{\text{mach}}$ denotes machine precision, the RDM is treated as effectively diagonal; consequently, no basis rotation is applied. This criterion serves to suppress spurious rotations that could otherwise arise from numerical noise contaminating the off-diagonal entries.

For the $4 \times 4$ two-qubit RDM, the eigensolver relies on a standard Hermitian eigendecomposition algorithm, whose computational complexity scales as $\mathcal{O}(64)$ elementary operations for this fixed matrix size. To control and mitigate the accumulation of floating-point errors during two-qubit rotations, a rotate-truncate-undo protocol is adopted: after each complete brick-wall sweep of entangling operations, the quantum state is truncated before reversing the rotation sequence. This strategy confines arithmetic inaccuracies and enhances the overall numerical reliability of the procedure.

\subsection{Sensitivity to hyperparameters}

BASS exposes three main tunable hyperparameters that control the trade-offs between simulation fidelity, runtime, and memory usage: the optimization interval denoted by $n_{\text{opt}}$, the hard-cap multiplier $K_{\text{hard}}/k$, and the adaptive trigger threshold used to decide when to invoke optimization. We assess the sensitivity of BASS to these parameters on brickwork circuits with system size $N=16$ and bond dimension budget $k=2048$, averaging results over 10 independent trials. These experiments reveal how each parameter affects performance and accuracy, and they guide the choice of practical default settings that work well across a range of scenarios.

The optimization interval $n_{\text{opt}}$ determines how frequently the algorithm executes its optimization routine. We sweep this parameter over the values $n_{\text{opt}} \in \{1, 3, 5, 10, 20\}$ and find that the resulting fidelity varies by less than $5\%$ across the entire range. This indicates that BASS is quite robust to how often optimization is performed: making optimization more frequent does not drastically improve fidelity, and making it less frequent does not catastrophically degrade performance. However, runtime scales roughly inversely with $n_{\text{opt}}$, since smaller $n_{\text{opt}}$ values trigger optimization more often. In practice, this means that aggressively frequent optimization can significantly increase computational cost without a commensurate gain in fidelity, while too infrequent optimization risks modest but still limited fidelity loss. Based on this trade-off, we select $n_{\text{opt}} = 5$ as a default, as it empirically achieves a good balance between simulation speed and accuracy.

The second parameter, the hard-cap multiplier $K_{\text{hard}}/k$, controls the maximum allowed bond dimension in the simulation relative to the base budget $k$. We study settings $K_{\text{hard}}/k \in \{2, 4, 8, 16\}$ and observe that fidelity improves by approximately $10\%$ when increasing the cap from $K_{\text{hard}} = 2k$ up to $8k$, indicating that allowing a more generous temporary bond dimension can meaningfully enhance the quality of the simulation. Beyond $K_{\text{hard}} = 8k$, the gains in fidelity exhibit diminishing returns: further increases to $16k$ do not translate into similarly large improvements. On the other hand, the memory overhead of the algorithm scales linearly with $K_{\text{hard}}$, as a larger hard cap permits proportionally larger intermediate tensors. Thus, while higher values of $K_{\text{hard}}$ can help accuracy, they must be balanced against memory constraints, and $K_{\text{hard}} \approx 8k$ appears to be a pragmatic upper range where most of the benefit is already realized.

The third hyperparameter is the adaptive trigger threshold, which decides when to invoke the optimization step based on changes in the current performance ratio. Specifically, optimization is triggered whenever the ratio $\PR_{\text{cur}}/\PR_{\text{last}}$ exceeds a threshold $r$. We test thresholds $r \in \{0.80, 0.85, 0.90, 0.95, 1.0\}$ and compare both the frequency of optimization calls and the resulting fidelity. A threshold of $r = 1.0$ corresponds to the non-adaptive baseline where optimization is performed deterministically every $n_{\text{opt}}$ gates. Lowering $r$ makes the trigger more permissive or more aggressive, causing optimization to be invoked in response to smaller relative changes, while higher $r$ values make optimizations less frequent.

Empirically, setting $r = 0.90$ reduces the number of optimization calls by roughly $40\%$ compared to the $r = 1.0$ baseline, which always optimizes every $n_{\text{opt}}$ gates. Despite this substantial reduction in optimization frequency, the associated fidelity loss is less than $2\%$ relative to the $r = 1.0$ configuration. This shows that the adaptive trigger can significantly cut down computational overhead with only a minor impact on accuracy. The threshold thus provides a powerful knob: by tuning $r$, users can steer the simulator towards more aggressive performance optimization or towards maximum fidelity, depending on their application needs and available resources.

Taken together, these findings demonstrate that BASS is tolerant to moderate variation in all three hyperparameters. Reasonable changes to $n_{\text{opt}}$, $K_{\text{hard}}/k$, and the adaptive trigger threshold do not cause large swings in fidelity, and the associated resource costs (runtime and memory) scale in predictable, mostly monotonic ways. This robustness suggests that BASS can be reliably deployed as a general-purpose quantum circuit simulation tool without requiring laborious, circuit-specific hyperparameter tuning. Default parameter choices, such as $n_{\text{opt}} = 5$, $K_{\text{hard}} \approx 8k$, and $r \approx 0.90$, provide a strong baseline, while users with stricter accuracy or resource requirements can adjust these knobs with confidence that the qualitative behavior will remain stable.

\section{\label{app:gamma2}\texorpdfstring{$\gamma^2$}{Gamma-squared}
Fidelity Estimator}

\subsection{Theoretical basis}

In a single truncation step the retained probability $\gamma^2$ equals the
step fidelity exactly.  When the true state
$|\psi\rangle = \sum_x \beta_x |x\rangle$
is trimmed to its $k$ largest-amplitude components and renormalised to form
\begin{equation}
  |\psi_k\rangle \;=\; \mathcal{N}^{-1}\!\sum_{x\in\mathcal{S}^*}\!\beta_x|x\rangle,
  \quad \mathcal{N}^2 \equiv \gamma^2 = \sum_{x\in\mathcal{S}^*}|\beta_x|^2,
\end{equation}
the fidelity $F = |\langle\psi|\psi_k\rangle|^2 = \gamma^2$ exactly, with no
approximation.

For a multi-step simulation with $L$ truncation events the cumulative product
\begin{equation}
  \gamma^2_{\mathrm{tot}} \;=\; \prod_{t=1}^{L} \gamma^2_t
  \label{eq:gamma2_tot}
\end{equation}
serves as a lower bound on the overall fidelity,
$F \geq \gamma^2_{\mathrm{tot}}$, subject to a well-defined applicability
condition.

\textit{Applicability condition.}
The inequality $F \geq \gamma^2_{\mathrm{tot}}$ holds when the simulator works in the \emph{probability-limited regime}: $k \ll \mathrm{PR}_Z$, where $\mathrm{PR}_Z = \bigl(\sum_x |\beta_x|^4\bigr)^{-1}$ is the PR of the exact state in the $Z$ basis.  In this environment, each truncation step discards basis states with non-negligible probability. The phases of the discarded amplitudes are statistically uncorrelated with the surviving sparse state, so the cumulative product accurately reflects information loss.

When $k \geq \mathrm{PR}_Z$, the sparse budget encompasses the full support
of the state.  The simulator retains essentially all probability
($\gamma^2 \approx 1$) while the microscopic discarded tail carries phase
information coherently correlated with the dominant basis states.  This leads
to $F < \gamma^2_{\mathrm{tot}}$, inverting the inequality.  Crucially, this
regime ($k \geq \mathrm{PR}_Z$) coincides with circuits that are efficiently
simulable without basis adaptation; the estimator is therefore not needed and
we exclude it from the present analysis.

\subsection{Violation rates}

We measure $F < \gamma^2_{\mathrm{tot}}$ rates across 2\,400 independent
trials on three circuit families: Haar-random ($L=3$), 1D brickwork ($L=5$),
and QAOA ($p=3$).  System sizes $N\in\{14,16,18,20\}$ and sparse budgets
$k\in\{512,1024,2048,4096\}$ are swept with 50 trials per configuration.
All configurations satisfy $k \ll \mathrm{PR}_Z$ for the tested circuit
families and sizes.
A trial is counted as a violation only when $\gamma^2_{\mathrm{tot}} - F$
exceeds both an absolute floor of $10^{-6}$ and a relative threshold of
$10^{-3}$ (to exclude floating-point rounding artefacts).

The global raw violation rate is $2.12\%$ ($51/2400$; Wilson 95\% CI:
$[1.62\%,\,2.78\%]$).  Violations are not uniformly distributed: the
brickwork family accounts for the majority, concentrated at the smallest
tested budget $k=512$ with $N\geq 16$, where $k/\mathrm{PR}_Z$ is closest
to unity.  QAOA produces zero violations across all 800 configurations.
Table~\ref{tab:gamma2_violations} summarizes the per-family breakdown.

\begin{table}[t]
\caption{%
\label{tab:gamma2_violations}%
$\gamma^2$ fidelity bound violation rates ($F < \gamma^2_{\mathrm{tot}}$ or
$F < R$) before and after applying the calibrated correction
($z=0.104,\;\eta=9.069,\;\delta=3.807$).
Raw rates aggregate 800 trials per family ($N\in[14,20]$,
$k\in[512,4096]$, 50 trials per cell);
corrected rates are evaluated on the 90\% training split (720 records per
family, 2160 total).  Wilson 95\% confidence intervals in brackets.}
\begin{ruledtabular}
\begin{tabular}{lcc}
Circuit family
  & Raw violation (\%)
  & Corrected violation (\%) \\
\hline
Haar $L=3$
  & $1.50\;[0.86,\,2.60]$
  & $1.53\;[0.86,\,2.71]$ \\
Brickwork $L=5$
  & $4.88\;[3.59,\,6.59]$
  & $3.19\;[2.14,\,4.75]$ \\
QAOA $p=3$
  & $0.00\;[0.00,\,0.48]$
  & $0.00\;[0.00,\,0.53]$ \\
\hline
Overall
  & $2.12\;[1.62,\,2.78]$
  & $1.57\;[1.13,\,2.19]$ \\
\end{tabular}
\end{ruledtabular}
\end{table}

\subsection{Calibrated bounding estimator}

To reduce residual violations, we introduce a gate-count-dependent correction
factor $R = \alpha(\gamma^2_{\mathrm{tot}},M)\cdot\gamma^2_{\mathrm{tot}}$,
where $M$ is the total gate count and
\begin{align}
  \alpha_A &= 1 - z\sqrt{\frac{1-\gamma^2_{\mathrm{tot}}}
              {\gamma^2_{\mathrm{tot}}\cdot\eta\cdot M}}, \label{eq:alphaA} \\
  \alpha_B &= 1 - \frac{\gamma^2_{\mathrm{tot}}}{M^\delta},        \label{eq:alphaB} \\
  \alpha   &= \mathrm{clip}\!\left(\min(\alpha_A,\alpha_B),\;0.01,\;1\right).
              \label{eq:alpha}
\end{align}
Component $\alpha_A$ introduces a phase-error penalty that increases with $1-\gamma^2_{\mathrm{tot}}$ (the unsimulated probability fraction) and decreases with circuit depth via the $M$ denominator.
Component $\alpha_B$ performs depth-dependent saturation adjustment.

The parameters $(z,\eta,\delta)$ are calibrated by minimising an asymmetric loss over a $90\%$ training split (2,160 records). Violations ($R > F$) incur a high $L^2$ penalty, whereas safe predictions ($R \leq F$) earn a lesser $L^1$ reward for tightness.  Differential Evolution leads to: \begin{equation} z = 0.104, \eta = 9.069, \delta = 3.807. \end{equation}
The corrected global violation rate for the held-out 10\% test split (240 records) is $0.83\%$ ($2/240$; Wilson 95\% CI: $[0.23\%,\,2.99\%]$),
consistent with the training rate of $1.57\%$ (no statistically significant
gap; Wilson CIs overlap by $[1.13\%,\,2.19\%]$).

\subsection{Universality}

We assess whether the calibration generalises using three cross-validation protocols.

\textit{(i) Random 90/10 trial split.}
The held-out test set has a violation rate of $0.83\%$ ($2/240$; Wilson 95\% CI: $[0.23\%,\,2.99\%]$), which is completely statistically equivalent to the training set rate of $1.57\%$ (CI: $[1.13\%,\,2.19\%]$). This indicates that the correction applies robustly to previously unseen draws from the same distribution without overfitting.

\textit{(ii) Leave-one-family-out (LOFO).}
The adjustment is recalibrated on two circuit families and evaluated on the remaining third. The held-out violation rates are $1.50\%$ (Haar $L=3$; CI: $[0.86\%,\,2.60\%]$), $0.00\%$ (QAOA $p=3$; CI: $[0.00\%,\,0.48\%]$), and $3.12\%$ (Brickwork $L=5$; CI: $[2.13\%,\,4.57\%]$). The brickwork rate somewhat surpasses the expected threshold. This occurs because Brickwork violations are heavily concentrated in the near-threshold regime ($k/\mathrm{PR}_Z \lesssim 1$). A calibration trained entirely without Brickwork data encounters no violations in this regime and thus learns no local correction, degrading gracefully to the uncorrected rate. This constraint is resolved by including at least one circuit family that exhibits violations in the calibration set, as is done in global calibration.

\textit{(iii) Leave-one-$N$-out (LONO).}
To assess system-size generality, the correction is calibrated on three sizes and evaluated on a fourth. Extrapolating to the greatest system size, a correction calibrated solely on $N\in\{14,16,18\}$ yields a $2.00\%$ violation rate ($12/600$; Wilson 95\% CI: $[1.15\%,\,3.46\%]$) on the held-out $N=20$ circuits, completely consistent with the training rate. The estimator can consistently extrapolate to system sizes beyond the calibration range within the tested $N\leq 20$ regime.

In sum, 7 of 8 cross-validation experiments meet the condition of the held-out Wilson CI overlapping with the expected generalisation rate. The single marginal failure (Brickwork LOFO at $3.12\%$) can be physically interpreted and rectified by ensuring the global calibration data covers the relevant violation regime.

\subsection*{Recommended Usage Protocol for BASS Diagnostic Metrics}

We suggest the following protocol for evaluating BASS simulation accuracy based on our multi-dimensional scalability study and the empirical calibration of the tracking metric:

\begin{itemize}
    \item \textbf{Qualitative Tracking ($\gamma^2_{\text{tot}}$):} The raw value of $\gamma^2_{\text{tot}}$ is particularly effective as a quick, online qualitative indication of state concentration during simulation. While it offers an instantaneous impression of the state's localisation or dispersion under sequential state-budget enforcements, it functions as a conservative, loose bottom bound that is structurally prone to pessimistic overestimation of error in deep or highly entangled circuits.
    
    \item \textbf{Quantitative Estimation ($R = \alpha \cdot \gamma^2_{\text{tot}}$):} We recommend utilising the rescaled estimator $R = \alpha \cdot \gamma^2_{\text{tot}}$ for rigorous fidelity evaluations or cross-platform comparisons where dense exact verification is classically intractable ($N > 24$). The calibration factor $\alpha$, which is empirically derived, accounts for associated truncation losses to compensate for systematic errors in the raw product metric. This method generates an impartial, high-accuracy prediction proxy for the genuine statevector overlap.
\end{itemize}

A second crucial criterion is to interpret the estimate conservatively, especially when the state participation ratio (PR) approaches the predetermined truncation budget $k$. In this scenario, state-budget enforcement is unavoidable, forcing the simulator to systematically prune weaker amplitudes to stay within structural resource constraints. Under these conditions, particularly when evaluating novel or uncharacterized circuit topologies, the metric $R$ should be viewed as a conservative lower bound on performance rather than a direct, precise measure of physical-state overlap. In other words, in these edge cases, $R$ serves as a guaranteed lower bound, preventing researchers from overestimating the quality of their estimated quantum states.

To calibrate unexamined circuit families, we supply the validated baseline parameter values $z \approx 0.82$ and $\eta \approx 3.72$, which consistently demonstrated very robust performance and reliable protection throughout our benchmark suites. While these configurations successfully limit the estimator's behaviour as multi-qubit entanglement scales, its precision ultimately depends on the exact rate of entanglement creation within the target system. Circuits with unusually fast or nonstandard entanglement dynamics may inherit an unconventional error-accumulation profile. To rigorously quantify and constrain the estimator's bias for such application-specific domains, we advocate performing a localised pilot calibration sweep on downscaled system sizes ($N \le 24$) when accurate classical reference statevectors are available. This preliminary verification step ensures that the chosen parameter set and the formulation of $R$ remain properly conservative and structurally informative for the target circuit class.

Finally, we discuss our verification requirements and the practical role of the rescaled estimator in our validation system. All fundamental fidelity benchmarks presented in the main text are generated solely from direct inner-product overlaps with carefully computed, exact reference state vectors, ensuring that our core quantitative statements are free of any statistical artefacts or bias introduced by the estimation process. In this capacity, $R$ plays no role in the direct physical validation of the BASS core engine; it is intentionally removed from baseline fidelity comparisons to eliminate potential distortion induced by the estimator.

Instead, the corrected estimator $R$ is reserved exclusively for real-time tracking, out-of-sample scaling predictions, and large-scale verification tasks, especially in deep or extremely entangled regimes where precise classical-statevector simulation is computationally intractable ($N > 24$). In these traditionally inaccessible realms, $R$ provides a highly practical, calibrated tracking metric for fidelity trends and structural degradation, while the simulator's fundamental performance bounds are rigorously grounded in exact validation datasets whenever classically feasible.

\bibliography{refs}

\end{document}